\documentclass[12pt]{article}

\usepackage{commath}
\usepackage[numbers]{natbib}
\usepackage[colorlinks=true,allcolors=blue,pdfencoding=auto]{hyperref}
\usepackage{amssymb}
\usepackage{makecell}
\usepackage{amsbsy}
\usepackage{amsmath}
\usepackage{amsthm}
\usepackage{amsfonts}
\usepackage{comment}
\usepackage{enumerate}
\usepackage{todonotes}
\usepackage{thmtools}
\usepackage{thm-restate}

\usepackage{setspace}

\usepackage{statistics-macros}
\usepackage[top=2.54cm,left=3cm,right=3cm,bottom=2.54cm]{geometry}

\usepackage{color}
\newcommand{\red}[1]{\textcolor{red}{#1}}
\definecolor{innerboxcolor}{rgb}{.9,.95,1}
\definecolor{outerlinecolor}{rgb}{.6,0,.2}
\definecolor{outerlinecoloreb}{rgb}{0,1,0}
\definecolor{innerboxcoloreb}{rgb}{1,1,1}

\declaretheorem[name=Theorem,numberwithin=section]{thm}

\declaretheorem[name=Lemma,numberwithin=section]{lemma}

\begin{document}

\title{Bounds on the conditional and average treatment effect with unobserved confounding factors}

\author{Steve Yadlowsky, Hongseok Namkoong, Sanjay Basu, John Duchi, Lu Tian}
  \begin{center}
  {\LARGE Bounds on the conditional and average treatment effect with unobserved confounding factors} \\
  \vspace{.5cm} {\large Steve Yadlowsky$^1$, Hongseok Namkoong$^2$, Sanjay Basu$^3$, John Duchi$^4$, Lu Tian$^5$} \\
  \vspace{.2cm}
  $^{1}$Google Research, Brain Team, \texttt{yadlowsky@google.com}\\
  $^2$Decision, Risk, and Operations Division, Columbia Business School, \texttt{namkoong@gsb.columbia.edu} \\
  $^3$School of Public Health,
  Imperial College London, \texttt{s.basu@imperial.ac.uk} \\
  $^4$Statistics and Electrical Engineering, Stanford University, \texttt{jduchi@stanford.edu}\\
  $^5$Biomedical Data Science, Stanford University, \texttt{lutian@stanford.edu}
\end{center}

\begin{abstract}
  For observational studies, we study the sensitivity of causal inference when
  treatment assignments may depend on unobserved confounders. We develop a
  loss minimization approach for estimating bounds on the conditional average
  treatment effect (CATE) when unobserved confounders have a bounded effect on
  the odds ratio of treatment selection. Our approach is scalable and allows
  flexible use of model classes in estimation, including nonparametric and black-box machine
  learning methods. Based on these bounds for the CATE, we propose a sensitivity
  analysis for the average treatment effect (ATE). Our semi-parametric
  estimator extends/bounds the augmented inverse propensity weighted (AIPW)
  estimator for the ATE under bounded unobserved confounding. By constructing
  a Neyman orthogonal score, our estimator of the bound for the ATE is a regular root-$n$ estimator so
  long as the nuisance parameters are estimated at the $o_p(n^{-1/4})$
  rate. We complement our methodology with optimality results
  showing that our proposed bounds are tight in certain cases. We demonstrate
  our method on simulated and real data examples, and show accurate coverage
  of our confidence intervals in practical finite sample regimes with rich
  covariate information.
\end{abstract}

\section{Introduction}

Consider a causal inference problem with
treatment indicator $Z\in \{0, 1\}$ representing control or
intervention, potential outcomes $Y(1) \in \R$
under intervention and $Y(0) \in \R$ under control,
and a set of observed covariates $X \in \mathcal{X} \subseteq \R^d$.
Our interest is in studying confounding bias in estimators of the \emph{conditional average treatment effect} (CATE)
\begin{equation*}
  \tau(x) \defeq \E[Y(1)-Y(0)\mid X=x],
\end{equation*} and estimation and inference of the \emph{average
  treatment effect} (ATE)
\begin{equation*}
  \tau \defeq \E[Y(1) - Y(0)]
\end{equation*}
based on $n$ i.i.d.\ observations $\{Y = Y(Z), Z, X\}$.\footnote{Together, these imply the stable unit treatment value assumption, which will be assumed throughout.} Many methods provide
consistent estimators for the ATE~\cite{ImbensRu15} under the independence
assumption
\begin{equation}
  \{Y(1), Y(0)\} \independent Z \mid X, \label{eq:indepsim}
\end{equation}
that all confounding factors are observed or equivalently, that observed
covariates $X$ account for all dependence between the potential outcomes and
treatment assignments.
Estimation of the CATE, $\tau(x),$ under the
independence assumption~\eqref{eq:indepsim} has recently generated substantial
interest~\cite{hill2011bayesian, athey2016recursive, kunzel2017meta,
  wager2017estimation, nie2017learning}.
  

Confounding bias is ubiquitous in observational studies, and the assumption
\eqref{eq:indepsim} is frequently too restrictive: in practice, there is
almost always an unobserved confounding factor $U \in \mc{U}$ affecting both
treatment selection and outcome. Consequently, we consider an unobserved
confounding factor $U$ such that
\begin{equation}
  \{Y(1), Y(0)\} \independent Z \mid X,U. \label{eq:indep}
\end{equation}
This allows there to be a common cause $U$ of the treatment $Z$ and potential outcomes $\{Y(0), Y(1)\}$ that contains the relevant information about the potential outcomes that influence the treatment assignment. More abstractly, it allows the treatment assignment $Z$ to depend directly on the unobserved
potential outcome; a multivariate unobserved confounder $U$ satisfying
condition~\eqref{eq:indep} always exists by letting $U=(Y(1), Y(0))$. Under
this assumption, neither the ATE $\tau$ nor the CATE $\tau(x)$ is
identifiable, and traditional estimators can be arbitrarily
biased~\cite{Rosenbaum02,imbensreview,robins2000sensitivity}.  Yet it may be
plausible that there is not ``too much'' confounding, so it is interesting to
provide bounds on the possible range of treatment effects under such
scenarios.  We take this approach to propose a sensitivity analysis linking
the posited strength of unobserved confounding to the range of possible values
of the ATE $\tau$ and CATE $\tau(x)$.

We consider unobserved confounders that have bounded influence on the odds of
treatment assignment, following Rosenbaum's ideas~\cite{Rosenbaum02}.
\begin{definition}
  \label{definition:bounded-selection}
  A distribution $P$ over $\{Y(1), Y(0), X, U, Z\}$ satisfies the
  \emph{$\Gamma$-\cornfield} condition with $1 \le \Gamma < \infty$ if
  for $P$-almost all $u, \tilde{u} \in \mc{U}$ and $X \in \mc{X}$,
  \begin{equation}
    \frac{1}{\Gamma} \le
    \frac{P(Z = 1 \mid  X, U=u)}{P(Z = 0 \mid  X, U=u)}
    \frac{P(Z = 0 \mid  X, U=\tilde{u})}{P(Z = 1 \mid  X, U=\tilde{u})}
    \le \Gamma.
    \label{eq:cornfield}
  \end{equation}
\end{definition}
\noindent
Condition~\eqref{eq:cornfield} limits departures from the independence
assumption \eqref{eq:indepsim}, and is equivalent to a regression
model for the treatment selection probability~\cite[Prop.~12]{Rosenbaum02}
where the log odds ratio for treatment is
\begin{equation}
  \label{eq:stat-model}
  \log \frac{P(Z=1 \mid  X, U)}{P(Z=0 \mid  X, U)}
  = \kappa(X) + \log(\Gamma) b(U, X),
\end{equation}
for some function $\kappa : \mathcal{X} \to \R$ of observed
covariates $X$ and a bounded function $b: \mathcal{U}\times \mathcal{X} \to [0, 1]$ of the
unobserved, and observed confounders, $U$ and $X$ respectively. Such odds ratios are common, for example, in
medicine, where they reflect associations between risk factors and
outcomes~\cite{norton2018odds}. Practice requires choosing a realistic value
of $\Gamma$ to interpret the sensitivity analysis; we discuss this in more
detail in Section~\ref{section:discussion}. One common approach by
practitioners is to look at the level of $\Gamma$ when bounds on the ATE
$\tau$ crosses a certain level of interest (e.g. $0$), which measures the
robustness of the findings to unobserved confounding ~\cite{Rosenbaum02}, and
then consider how plausible that choice of $\Gamma$ would be for the data
generating process.

The ATE $\tau,$ and CATE $\tau(x),$ are partially identified under the $\Gamma$-\cornfield condition~\eqref{eq:cornfield}, so we focus instead on estimating bounds for them.  This perspective on sensitivity
to unobserved confounding traces to \citeauthor{CornfieldHaHaLiShWy59}'s
analysis demonstrating that if an unmeasured hormone can explain the observed
association between smoking and lung cancer, it would need to increase the
probability of smoking by nine-fold (an unrealistic
amount)~\cite{CornfieldHaHaLiShWy59}.  Contemporary medical informatics and
epidemiological studies focusing on small effect sizes require a more nuanced
approach for estimating the causal effect in the presence of unobserved
confounding than the simple one used by \citet{CornfieldHaHaLiShWy59}.  For
example, observational data is often used for post-market drug surveillance,
but \citet{bosco2010most} shows that unobserved confounding presents a
particularly high risk in these data, motivating the need for sensitivity
analysis to contextualize findings. \citet{coloma2012electronic} show that
effect sizes are often small, as adverse events for approved drugs are
relatively rare. Therefore, to draw confident and precise conclusions when
there is mild confounding, it is important to avoid applying an overly conservative sensitivity
analysis. Motivated to provide the most precise
conclusions possible in the presence of confounding, we seek methods that
provide optimal (tight) bounds on the CATE and ATE under the
$\Gamma$-\cornfield{} condition~\eqref{eq:cornfield}.

\subsection{Bounding treatment effects}

In what follows, we bound the confounding bias using analogues of the plug-in
treatment contrast estimator for the the CATE, and the augmented inverse
probability weighted (AIPW) estimator for the
ATE~\cite{bang2005doublerobust}. We treat each potential outcome separately,
focusing on lower bounds on $\mu_1 = \E[Y(1)]$ (other cases are
symmetric). Based on observed data, all parameters necessary to estimate $\mu_1$ can be non-parametrically identified, except the conditional mean of the unobserved potential outcome, $\E[Y(1) \mid X, Z=0]$. Since this quantity is
not identifiable in the presence of unobserved confounding, we develop a
worst-case bound under the $\Gamma$-\cornfield{} condition~\eqref{eq:cornfield}, and develop estimators
based on the observed data. Specifically, let
\begin{equation}
\theta_1(x) \defeq \inf\{\E_Q[Y(1) \mid X=x, Z=0] : Q \in \mathcal{Q}_x\}
    \label{eq:cate-lower-bound}
\end{equation}
where $\mathcal{Q}_x$ is the set of all distributions for $(Y(0), Y(1), Z)$
conditional on $X=x$ satisfying the independence assumption \eqref{eq:indep}
and the bound \eqref{eq:cornfield} for $X = x$, and matching the conditional
distributions that are identified in the observed data $P$:
$Q(Z=1 \mid X)=P(Z=1 \mid X)$ and
$Q(Y(1) \in \cdot \mid Z=1, X) = P(Y(1) \in \cdot \mid Z=1, X)$.  By
definition, $ \theta_1(x) \le \E_P[Y(1) \mid X=x, Z=0]$ under the
bounded unobserved confounding ($\Gamma$-\cornfield{}
condition~\eqref{eq:cornfield}). Lower bounds on $\E[Y(1) \mid X=x]$ and $\E[Y(1)]$
follow from plugging in $\theta_1(x)$ in place of the unknown $\E_P[Y(1) \mid X=x, Z=0].$


Our first main result (Section~\ref{sec:cate-sensitivity}) shows that $\theta_1(x)$ can be expressed as the solution to the loss
minimization problem with a reweighted squared loss
\begin{equation*}
    \minimize_{\funcparam(\cdot)} 
    ~\half \E\left[ \hinge{Y(1) - \funcparam(X)}^2 +
    \Gamma \neghinge{Y(1) - \funcparam(X)}^2 \mid \treatmentrv=1 \right], 
\end{equation*}
where $a_+=a \ind{a>0}$, $a_-=-a \ind{a<0}, a \in \R$ and $\ind{\cdot}$ is the
indicator function.  
The scalable loss minimization approach allows us to use flexible model
classes to estimate the lower bound, including many nonparametric and machine
learning methods.  Intuitively, the preceding display upweights the penalty for negative residuals, therefore increasing the impact of smaller observed outcomes on the minimizer $\theta_1(x),$ correcting for the fact that selection bias from confounding may have decreased the frequency of smaller observed outcomes.

Our second main result defines a semiparametric estimator
\eqref{eq:orthogonal-estimator} for the lower bound on the expected outcome
$Y(1)$ under the $\Gamma$-\cornfield condition~\eqref{eq:cornfield}
\begin{equation}
  \label{eq:outcome-bound}
  \mu_1^- \defeq \E[Z Y(1) + (1-Z) \theta_1(X)] \le \E[Y(1)].
\end{equation}
Our estimation approach (Section~\ref{sec:semiparametric}) builds out of a
line of work~\cite{bang2005doublerobust, chernozhukov2018double} for
statistical inference on $\tau$ when all confounders are
observed~\eqref{eq:indepsim}; we adapt \citet{chernozhukov2018double}'s
cross-fitting procedure to allow large model classes to estimate nuisance
parameters. Our semiparametric estimator satisfies \emph{Neyman
  orthogonality}~\cite{Neyman59}, and is insensitive to estimation errors in
nuisance parameters.
By virtue of this orthogonality, our estimator is root-$n$ consistent and
asymptotically normal so long as the nuisance parameters are estimated at a
slower-than-parametric $o_p(n^{-1/4})$ rate of convergence.  Our result gives
asymptotically exact confidence intervals (CIs) for the lower bound $\mu_1^-$
~\eqref{eq:outcome-bound}.

Coupling the asymptotic distribution for $\hat{\mu}_1^-$ with the symmetrically defined upper and lower bounds $\what{\mu}_z^\pm$ for $\E[Y(z)]$, we can construct a
CI for the ATE $\tau$ under the
$\Gamma$-\cornfield condition~\eqref{eq:cornfield}. In general, the boundary of our
interval never shrinks to $\tau$ even in the large sample limit due to
unobserved confounding. However, when there is no
unmeasured confounding ($\Gamma = 1$), our method is equivalent to the AIPW estimator for the ATE $\tau$.

Our population-level bound is unimprovable for bounding each expected potential outcome and
their conditional counterparts, $\E[Y(z)]$ and $\E[Y(z) \mid X=x]$,
$z \in \{0, 1\}$, but may not be always optimal in bounding their difference, the ATE
$\tau = \E[Y(1) - Y(0)]$. On the other hand, when the potential outcomes are symmetric in the sense that $Y(0) \overset{d}{=} C(1-Y(1))$ for some constant $C$, then our bounds
on the treatment effect are also unimprovable
(Section~\ref{sec:hypothesis-test}), thereby guaranteeing that our CI converges (in the large sample limit) to
the smallest possible interval containing $\tau$ under the $\Gamma$-\cornfield
condition~\eqref{eq:cornfield}.

Finally, we supplement our theoretical analysis with an experimental
investigation of the proposed approaches in Section~\ref{sec:experiments}. On
both simulated and real-world data, we show that the CIs have
good coverage and reasonable length.

\subsection{Related Work}

The semiparametric literature~\cite{bang2005doublerobust,
  chernozhukov2018double} have shown that the augmented inverse probability
weighted (AIPW) estimator allows the use of flexible nonparametric and machine
learning models to estimate the nuisance parameters: conditional means
$\E[Y(z) \mid X]$, $z \in \{0, 1\}$, and the propensity score
$\P(Z = 1\mid X).$ By exploiting certain orthogonality properties,
\citet{chernozhukov2018double} showed how to obtain root-$n$ consistency and
asymptotic normality for estimated $\tau$ even when involved estimates of the
nuisance parameters converge at slower nonparametric rates. We generalize this
approach under the $\Gamma$-\cornfield{} condition.

A number of authors have studied nonparametric and semiparametric models for
sensitivity analysis. These works consider alternatives to our choice of
model~\eqref{eq:cornfield} in characterizing the relationship between
unobserved confounders, treatment, and
outcomes~\cite{,franks2019flexible,richardson2014nonparametric,
  robins2000sensitivity, zhao2017sensitivity, shen2011sensitivity,
  brumback2004sensitivity}.  We focus on the model of \citet{Rosenbaum02}
because of its appealing interpretation as a regression model~\eqref{eq:stat-model}.  

\citet{Imbens03} derived a sensitivity analysis for the treatment effect in the presence of unobserved confounding. His approach requires specifying parametric models for the effect of an unobserved confounder on both the treatment selection and outcome. Specifically, the relationship between the unmeasured confounder and treatment assignment is modelled via a logistic regression, which is a special case of condition~\eqref{eq:cornfield}.

\citet{aronow2012interval} and \citet{miratrix2017shape} study the bias due
to unknown selection probabilities in survey analysis, with
an
approach similar to ours.
In the survey setting, only
surveyed individuals provide covariates $X$,
so the papers~\cite{aronow2012interval,miratrix2017shape}
consider a simplified model for selection bias,
\begin{equation} \label{cond:simple-gamma-selection}
  \frac{1}{\Gamma} \le \frac{P(Z=1 \mid U=u)}{P(Z=0 \mid U=u)} \frac{P(Z=0 \mid U=\tilde{u})}{P(Z=1 \mid U=\tilde{u})} \le \Gamma.
\end{equation}
\citet{zhao2017sensitivity} and \citet{shen2011sensitivity} consider the
sensitivity of inverse probability weighted estimates of the ATE $\tau$ to
unobserved confounding by varying the propensity score estimates around their
estimated values. \citet{zhao2017sensitivity} discuss the relationship between
their model of bounded unobserved confounding---which they call the marginal
sensitivity model---and that based on the
$\Gamma$-\cornfield{}~\eqref{eq:cornfield}. 
Compared to our semiparametric estimator, the complexity of the asymptotic distribution of their estimator necessitates
using a bootstrap method for inference. A interesting
future direction is to extend the methods in this paper to improve
statistical inference under their model. 

The most common approach to sensitivity analysis for the ATE under
condition~\eqref{eq:cornfield} is to use matched
observations~\cite{Rosenbaum02, Rosenbaum10,
  rosenbaum2011new,rosenbaum2014weighted,fogarty2016sensitivity}. Unfortunately,
exactly matched pairs rarely exist in practice, even for covariate vectors of
moderate dimension; when considering continuous covariates, the probability of
finding exactly matched pairs is zero. \citet{abadie2006large} show that under
appropriate regularity conditions on the functions $\mu_{z}(x)$ and $e_1(x)$ (defined in Eqs.  \eqref{eq:true-mean-function} and \eqref{eq:prop-score}),
estimators of $\tau$ using approximately matched pairs can have a bias of
order $\Omega(n^{-1/d})$ for $d$-dimensional continuous covariates. For these
data, the AIPW method is a more appropriate statistical analysis tool. The
AIPW estimator and other semiparametric methods can provide
$\sqrt{n}$-consistent estimates of the ATE without unmeasured confounding
\cite{imbensreview, hahn98, scharfstein1999adjusting,
  chernozhukov2018double}. The semi-parametric approach for the lower bound on
the ATE that we present in Section~\ref{sec:semiparametric} is
$\sqrt{n}$-consistent under analogous regularity conditions. 
Therefore, when analyzing an observational study using the AIPW estimator, one should perform a sensitivity analysis using the semiparametric method we provide here. When finding good matched pairs is feasible, many analysts prefer matching due to the transparency of the results and the simplicity of confounding adjustment. If analyzing an observational study using matching methods, it would be natural to also use a matching-based method for sensitivity analysis, such as the ones described above. In summary, our proposed method and matching based sensitivity analysis approaches can be coupled with different main analyses in practice, and are complementary to each other.

Most work~\cite{hill2011bayesian, athey2016recursive, kunzel2017meta,
  wager2017estimation, nie2017learning} directly study estimation of the CATE
$\tau(x) = \mu_1(x) - \mu_0(x)$ assuming that all confounders are observed.
More recently, \citet{kallus2018confounding} present an approach to learning
personalized decision policy in the presence of unobserved confounding, and a
contemporaneous work with this paper~\cite{kallus2019interval} derive bounds
on the CATE; their methods are based on the marginal sensitivity model of
\citet{zhao2017sensitivity}.

\paragraph*{Notation}
We use $\P_n$ and $\P_n(\cdot \mid Z=z)$ to represent the empirical
probabilities of $\{(Y_i(Z_i), X_i, Z_i)\}_{i = 1}^n$ and
$\{(Y_i(Z_i), X_i) \mid Z_i=z\}$, respectively, and
$\empE[\cdot \mid Z=z]$ is the expectation with respect to $\P_n(\cdot \mid
Z=z)$ for $z=0, 1.$ We let $n_z=\sum_{i=1}^n \ind{Z_i=z}$ be the count
of observations with $Z_i = z$, where
$\ind{\cdot}$ is the indicator function.
For a distribution $P$ and function $f : \mc{X} \to \R$, we use
$\norm{f}_{2,P} = (\int_{\mc{X}} f^2(x) \dif{P}(x))^{1/2}$. For functions
$f : \Omega \to \R$ and $g : \Omega \to \R$ with arbitrary domain $\Omega$, we write $f \lesssim g$ if there
exists constant $C < \infty$ such that $f(t) \le C g(t)$ for all $t \in \Omega$,
and $f \asymp g$ if $g \lesssim f \lesssim g$. 
We use $P_z$ and $\E_z$ to denote the conditional distribution
$P(\cdot \mid Z=z)$ and associated expectation, respectively. We write
$\E_{Q}$ for the expectation under the probability $Q$, and omit the
subscript under the data-generating distribution $P$.


\section{Bounds on Conditional Average Treatment Effect}
\label{sec:cate-sensitivity}


To bound the CATE $\tau(X) = \E[Y(1) - Y(0)\mid X]$, we begin by separately
bounding
\begin{equation}
  \mu_1(X) = \E[Y(1)\mid X] ~~~\mbox{and}~~~ \mu_0(X) = \E[Y(0)\mid X].
  \label{eq:true-mean-function}
\end{equation}
We focus on $\mu_1(\cdot)$ as these two cases are symmetric.   Henceforth,
our statements hold for $P$-almost every $X$ and $P_z$-almost every $Y$ (where $z$ should be inferred from context).

\subsection{Bounding the unobserved potential outcome}
\label{sec:no-covariate-lower-bound}

Decompose $\mu_1(\cdot)$ into observed and unobserved components
\begin{equation}
  \label{eq:mu-1-decomposition}
  \mu_1(X) = \E[Y(1) \mid Z = 1, X] P(Z = 1 \mid X) 
 +\E[Y(1) \mid Z = 0, X] P(Z = 0 \mid X).
\end{equation}
The mean functions and the nominal propensity score,
\begin{gather}
  \label{eq:mean-function}
    \mu_{z,z}(X) = \E[Y(z) \mid Z=z, X],
    \\
  \label{eq:prop-score}
  e_z(X) = P(Z=z\mid X),
\end{gather}
are standard regression functions estimable based on observed data
~\cite{imbensreview,nie2017learning,
  scharfstein1999adjusting,wager2015adaptive}. The key difficulty in
estimating the CATE is that one potential outcome is always unobserved: we
never observe data to directly estimate $\E[Y(1) \mid Z=0, X]$.

We begin by reformulating the worst-case lower
bound~\eqref{eq:cate-lower-bound}, $\theta_1(\cdot)$, based on the likelihood
ratio between the observed and unobserved potential outcomes. We take a worst
case optimization approach over likelihood ratios to bound the unobserved
conditional mean. Using Lemma~\ref{lem:bdd-lr-no-cov} to come, the conditional
distribution $P(Y(1) \in \cdot \mid X, Z=1)$ is absolutely continuous with
respect to $P(Y(1) \in \cdot \mid X, Z=0)$ under condition
\eqref{eq:cornfield}, so
\begin{equation}
  \label{eq:lr-form}
  \E[Y(1) \mid Z = 0, X] = \E\left[Y(1) L(Y(1),X) \mid Z = 1, X \right],
\end{equation}
where $L$ is the likelihood ratio
\begin{equation}
  L(y,x)=\frac{\dif{P(Y(1) \in \cdot \mid Z = 0, X=x)}}{\dif{P(Y(1) \in \cdot \mid Z=1, X=x)}}(y).
  \label{eq:likratio}
\end{equation} 
While $L$ is unknown, the $\Gamma$-\cornfield{} condition \eqref{eq:cornfield}
constrains it, inducing a lower bound on the unobserved
quantity~\eqref{eq:lr-form}.
\begin{restatable}{lemma}{lembddlrnocov}
  \label{lem:bdd-lr-no-cov}
  Let $P$ satisfy the $\Gamma$-\cornfield{} condition~\eqref{eq:cornfield},
  and the conditional independence~\eqref{eq:indep}. Then $P_{Y(1)|Z = 0,X=x}$
  is absolutely continuous with respect to $P_{Y(1)|Z = 1,X=x}$, and the
  likelihood ratio~\eqref{eq:likratio} satisfies
  $0\le L(y,x) \le \Gamma L(\tilde{y},x)$ for almost every $y$, $\tilde{y}$
  and $x$.
  
  Furthermore, for any likelihood ratio $L$ satisfies
  $0\le L(y,x) \le \Gamma L(\tilde{y},x)$ for almost every $y$, $\tilde{y}$
  and $x$, there is a distribution $P$ satisfying the $\Gamma$-\cornfield{}
  condition~\eqref{eq:cornfield}, and the independence
  assumption~\eqref{eq:indep}, such that Eq.~\eqref{eq:likratio} holds.
\end{restatable}
See Appendix~\ref{sec:proof-bounded-likelihood} for a proof of the absolute continuity. The rest of the results are illuminating, so we provide them here, assuming absolute continuity.

\begin{proof}
For simplicity in notation and without loss
  of generality, we assume there are no covariates $x$. Define the
  likelihood ratio for the unobserved
  $U$ by
  $r(u) \defeq \frac{q_0(u)}{q_1(u)},$ where $q_z(u)$ is the probability density function for $U\mid Z=z.$ Note that by applying Bayes rule in the inequality \eqref{eq:cornfield}, for any $u,~\wt{u}$, 
  \begin{equation}
  r(u) \le \Gamma r(\wt{u}).
    \label{eq:u-bound}
  \end{equation}
  Then, for
  any set $B \in \sigma(Y(1))$, the sigma algebra of $Y(1)$, we have
  \begin{align*}
    \E[\ind{B} \mid Z = 0]
    = \E\left[\E[r(U) \mid Y(1), Z = 1]
      \ind{B} \mid Z = 1\right],
  \end{align*}
  so that almost everywhere, the likelihood
  ratio $L(y) = \frac{dP_{Y(1) \mid Z = 0}}{dP_{Y(1) \mid Z = 1}}(y)$
  satisfies
  \begin{equation}
    L(y) = \E\left[r(U) \mid Y(1) = y, Z = 1\right]
    \label{eqn:remember-radon}
  \end{equation}
  by the Radon-Nikodym theorem.  Now, for an arbitrary $\epsilon > 0$, and $y,~ \wt{y}$ satisfying the equality~\eqref{eqn:remember-radon}, let $u_0$ be such that $r(u_0) \le
  \inf_u r(u) + \epsilon$.
  Then
  \begin{equation*}
    L(y) \stackrel{(i)}{=}
    \E[r(U) \mid Y(1) = y, Z = 1]
    = r(u_0) \E\left[\frac{r(U)}{r(u_0)} \mid Y(1) = y, Z = 1\right]
    \stackrel{(ii)}{\le} \Gamma r(u_0)
  \end{equation*}
  where equality~$(i)$ is simply Eq.~\eqref{eqn:remember-radon} and
  inequality~$(ii)$ follows from the bound \eqref{eq:u-bound}.
  We also have $L(\wt{y}) \ge \inf_u r(u) \ge r(u_0) - \epsilon$
  by equality~\eqref{eqn:remember-radon}. Consequently,
  $L(y) \le \Gamma r(u_0) \le \Gamma (L(\wt{y}) + \epsilon)$, and
  as $\epsilon$ was arbitrary, this completes the proof.
  
The converse follows easily as well: given a likelihood ratio satisfying the
above constraint, the $\Gamma$-\cornfield~\eqref{eq:cornfield} condition and
the independence $\{Y(1), Y(0)\} \,\indep\, Z \mid X, U$ is satisfied for $U \defeq (Y(1), Y(0))$, and $P(Z=1 \mid Z=z, U=u)$ only depending on the $Y(1)$ component of $U$, and defined by applying Bayes rule to the likelihood ratio.
\end{proof}

Lemma~\ref{lem:bdd-lr-no-cov} implies that the lower bound $\theta_1(x)$ from
Eq.~\eqref{eq:cate-lower-bound} on the unobserved conditional expectation
$\E[Y(1)\mid X, Z = 0]$ is:
\begin{align}
  \label{eq:population-theta}
  \theta_1(X) &=
  \inf \left\{ \E[Y(1) L(Y(1)) \mid Z = 1, X] ~:~ L \in \likeratioset\right\}
\end{align}
where
\begin{equation*}
   \likeratioset = \left\{ L :\R \to \R~\text{measurable} ~:~\begin{aligned} & 0 \le L(y) \le \Gamma L(\tilde{y})~\text{for all}~y,\tilde{y},\\&\E[L(Y(1)) \mid Z=1, X] = 1\end{aligned}\right\}.
\end{equation*}
The first constraint in $\likeratioset$ comes from the $\Gamma$-\cornfield{} condition
(Lemma~\ref{lem:bdd-lr-no-cov}), and the second normalization constraint
guarantees that $L$ is a likelihood ratio; the objectives and constraints are
linear in $L$. Applying Lagrangian duality to these constraints and
simplifying the resulting dual problem shows that the solution to this
optimization problem is the solution to an estimating equation in terms of the
function
\begin{equation}
  \psi_{\theta}(y) \defeq \hinge{y - \theta} - \Gamma \neghinge{y -
    \theta}.
  \label{eq:psi-defn}
\end{equation}

\begin{restatable}{lemma}{populationduality}
  \label{lemma:duality}
  Let $\theta_1(X)$ be defined as in
  \eqref{eq:population-theta}. If $|\theta_1(X)|<\infty$, then
  $\theta_1(X)$ solves 
  $$\E[ \psi_{\theta_1(X)}(Y(1)) \mid Z=1, X]=0$$ 
  whenever this solution is unique. If the solution is not unique, 
  \begin{equation}
    \theta_1(X) 
    = \sup \left\{\mu \in \R
    ~ :~ \E[ \psi_\mu(Y(1)) \mid Z=1, X] \ge 0
    \right\}.
    \label{eq:dual-constrained}
  \end{equation}
\end{restatable}

While $\theta_1(\cdot)$ could be estimated using a local estimating equation approach (eg., as in \cite{newey1994kernel} and \cite{athey2019generalized}) for the equations
$\E[ \psi_{\theta_1(X)}(Y(1))\mid Z = 1, X]=0$ for each $X$, we go further to provide an
alternative loss minimization method to estimate $\theta_1(\cdot)$.
This enables the application of a broad class of computationally and statistically efficient estimators.

\begin{figure}
    \centering
    \includegraphics[width=3in]{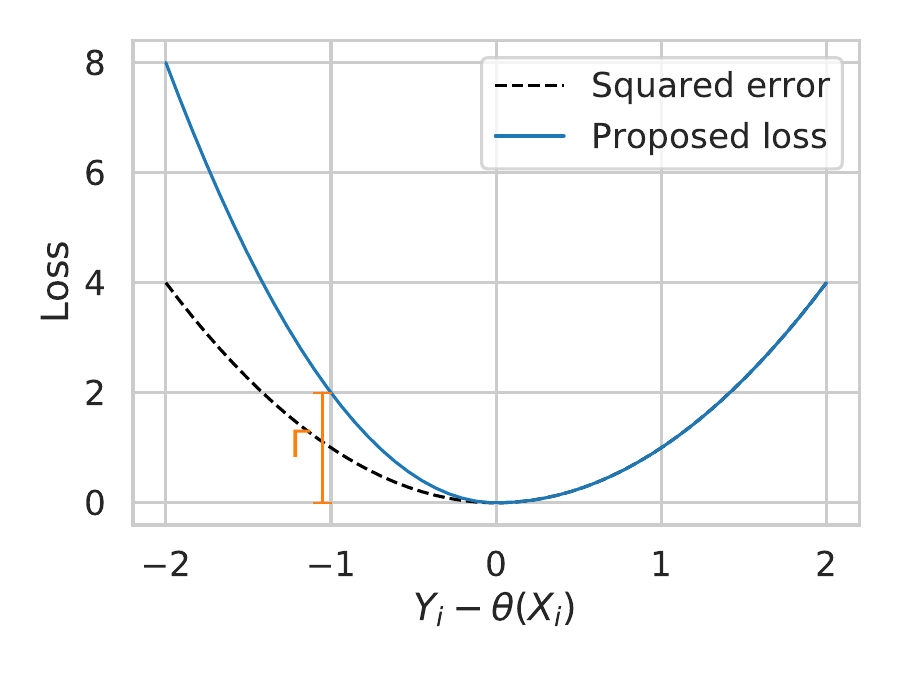}
    \caption{Loss function \eqref{eq:gamma-loss} to minimize to lower bound conditional mean of unobserved potential outcome under the $\Gamma$-\cornfield{} condition. Illustrated here for $\Gamma=2$. This loss penalizes negative residuals more than positive residuals, to account for the the fact that confounding could be already up-weighting positive residuals.}
    \label{fig:loss}
\end{figure}

The lower bound $\theta_1(\cdot)$ is the solution to
the convex loss minimization problem
  \begin{equation}
    \label{eqn:opt}
    \minimize_{\funcparam(\cdot)} 
    ~\E[  \loss_\Gamma\left(\funcparam(X), Y(1)\right) \mid \treatmentrv=1],
  \end{equation}
  where $\loss_{\Gamma}$ is the weighted squared loss
  \begin{equation}
    \label{eq:gamma-loss}
  \loss_\Gamma(\theta, \outcome)
  \defeq \half \left[ \hinge{\outcome - \theta}^2 +
    \Gamma \neghinge{\outcome - \theta}^2 \right],
\end{equation}
illustrated in Figure~\ref{fig:loss}. Noting that
$\frac{\dif{}}{\dif{\theta}} \loss_\Gamma(\theta, \outcome) =
-\psi_{\theta}(y)$, we have the following lemma on the uniqueness properties
and structure of $\funcparam_1$ solving the optimization
problem~\eqref{eqn:opt}.
\begin{restatable}{lemma}{lemoptisgood}
  \label{lemma:opt-is-good}
  Assume $(t, x) \mapsto
  \E[\loss_\Gamma(t, Y(1)) \mid X = x, Z = 1]$ is continuous
  on $\R \times \mc{X}$.
  If $\treatedE[\loss_\Gamma(\popfunc(\covariaterv), Y(1))] < \infty$, then
  $\popfunc(\cdot)$ solves
  $\E\left[ \psi_{\theta}(Y(1)) \mid X = x, Z=1\right]=0$
  for almost every $x$ if and only if it
  solves \eqref{eqn:opt}.
  Such a minimizer $\theta_1(\cdot): \mc{X} \to \R$ exists and is unique up to
  measure-0 sets.
\end{restatable}
\noindent
See Appendix~\ref{section:proof-of-opt-is-good} for proof. Our approach allows
both classical techniques, such as sieves, and flexible use of modern machine
learning methods to estimate $\theta_1(x)$; in our experiments, we demonstrate
how to approximately solve the loss minimization problem~\eqref{eqn:opt} using
gradient boosted decision trees.

\subsection{Nonparametric Estimation with Sieves}
\label{sec:nonparametric-sieves}
To obtain concrete nonparametric guarantees, we consider the method of
sieves~\cite{GemanHw82}, which considers an increasing sequence
$\funcspace_1 \subset \funcspace_2 \subset \cdots \subset \funcspace$ of
spaces of (smooth) functions, where $\funcspace$ denotes all measurable
functions. Here, for a sample size $n$, we take the estimator
$\what{\theta}_1(\cdot)$ solving
\begin{equation}
  \label{eqn:opt-emp}
  \minimize_{ \funcparam \in \funcspace_{n}}
    ~\empE[  \loss_\Gamma\left(\funcparam (\covariaterv), Y(1)\right) \mid Z=1].
\end{equation}
With appropriate choices of the function spaces $\funcspace_n$, it is possible
to provide strong approximation and estimation guarantees. As the loss
$\funcparam \mapsto \loss_\Gamma(\funcparam(\covariate), \outcome)$ is convex,
the empirical optimization problem~\eqref{eqn:opt-emp} is convex when
$\funcspace_{n}$ is a finite dimensional linear space (eg. polynomials,
splines), which facilitates efficient computation~\cite{BoydVa04}.

In Appendix~\ref{sec:sieve-method}, we adapt results for sieve estimators~\cite{Chen07} to show convergence rates for the
solution $\what{\theta}_1(\cdot)$ to the empirical
problem~\eqref{eqn:opt-emp}.  When $\theta_1(X)$ belongs in a
$\holdersmooth$-smooth H\"{o}lder space, in Theorem~\ref{thm:sieve} of the Supplementary Materials, we prove
that the empirical solution $\what{\theta_1}(\cdot)$ is consistent and achieves the following convergence rate (up to logarithmic factors):
\begin{equation*}
  \norm{\what{\theta}_1(\cdot) - \theta_1(\cdot)}_{2, P_1} = O_P\left( \left(\frac{\log n}{n} \right)^{\frac{p}{2p+d}} \right).
\end{equation*}
In the interest of space, we defer a comprehensive treatment to
Appendix~\ref{sec:sieve-method}.

\subsection{Bounding the CATE}

Since $\theta_1(\cdot)$ satisfies $\theta_1(X) \le \E[Y(1) \mid X, Z=0]$ under
the $\Gamma$-\cornfield{} condition, altogether $\mu_1^{-}(\cdot)$ defined
below provides the lower bound
\begin{equation*}
  \mu_1^{-}(X) = \mu_{1,1}(X) e_1(X) + \theta_1(X) e_0(X) \le \mu_1(X).
\end{equation*}
By symmetry, letting
$\mu_0^+(X)=\mu_{0,0}(X) e_0(X) + \theta_0(X) e_1(X)$
where
\begin{align}
  \label{eq:population-mu-0-upper}
  \theta_0(X) & =
  \sup_{L~\textup{measurable}} \E[Y(0) L(Y(0)) \mid Z = 0, X] \\
  & \qquad ~
  \text{s.t.} ~ 0 \le L(y) \le \Gamma L(\tilde{y})~\text{all}~y,\tilde{y},
  ~~~ \E[L(Y(1)) \mid Z=0, X] = 1,
  \nonumber
\end{align}
we have the parallel conclusion that $\mu_0^+(X) \ge \mu_0(X)$ holds under
$\Gamma$-\cornfield{} condition. Similar to the above,
$\theta_0(\cdot)$ is a unique minimizer of
$\E[\loss_{\Gamma^{-1}}(\funcparam(\covariaterv), Y(0)) \mid \treatmentrv=0]$.\footnote{Convergence results for sieve estimators of $\theta_0(\cdot)$ again fall out
of our results in Section~\ref{sec:sieve-method}.}

Thus, under the $\Gamma$-\cornfield{} condition~\eqref{eq:cornfield}, a valid
lower bound on the CATE is simply
\begin{equation}
  \tau^{-}(X) = \mu_1^{-}(X) - \mu_0^{+}(X).
  \label{eq:ate-lower-bound}
\end{equation}  
We summarize our developments in the theorem below.
\begin{restatable}{thm}{lowerboundproperty}
  \label{thm:lower-bound-property}
  Let $\Gamma \ge 1$ and $\{Y(1), Y(0), Z, X, U\}$ satisfy
  condition~\eqref{eq:cornfield} and the conditional independence
  assumption~\eqref{eq:indep}. Let $\tau^{-}(X)$ in \eqref{eq:ate-lower-bound}
  use $\theta_1(X)$ and $\theta_0(X)$ solving the optimization problems
  \eqref{eq:population-theta} and \eqref{eq:population-mu-0-upper} with the
  same $\Gamma$. When $\E[|Y(z)| \mid X]<\infty$ for $z = 0, 1$ and $0 < e_1(X)<1$,
  \begin{equation*}
    \tau^{-}(X) \le \E[Y(1) - Y(0) \mid X].
  \end{equation*}
\end{restatable}

A natural estimator for $\tau^{-}(x)$ is the difference in conditional
expected potential outcomes
\begin{gather*}
\what{\tau}^{-}(x)= \what{\mu}_1^{-}(x)-\what{\mu}_0^+(x),\\
\what\mu_1^-(x) = \what{\mu}_{1, 1}(x)\what{e}_1(x) + \what{\theta}_1(x) \what{e}_0(x), ~~\mbox{and}~~
\what\mu_0^+(x) = \what{\mu}_{0, 0}(x) \what{e}_0(x) + \what{\theta}_0(x) \what{e}_1(x),
\end{gather*}
where $\what{e}_z(\cdot)$ and $\what{\mu}_{z, z}(\cdot)$ are suitable
estimators for the nominal propensity score $e_z(\cdot)$ and the observed
potential outcome's mean function $\mu_{z,z}(\cdot),$ respectively.  A
variety of classical nonparametric methods and machine learning methods can estimate these
regression
functions~\cite{chernozhukov2018double,wager2015adaptive,chen1999improved}.
To understand convergence of $\what{\tau}^{-}(\cdot)$, consider the
convergence of these regression estimates. Specifically,
assume that the estimators $\what{e}_1(\cdot)$
and $\what{\mu}_{z, z}(\cdot)$
satisfy that
\begin{align*}
&\norm{\what{e}_1(\cdot) - e_1(\cdot)}_{2, P} = O_P(r_{n, 1}), \\
  &\norm{\what{\mu}_{1, 1}(\cdot) - \mu_{1, 1}(\cdot)}_{2, P_1} = O_P(r_{n, 2}),
    ~~\norm{\what{\mu}_{0, 0}(\cdot) - \mu_{0, 0}(\cdot)}_{2, P_0} = O_P(r_{n, 3}), \\
  & \norm{\what{\theta}_1(\cdot) - \theta_1(\cdot)}_{2, P_1} = O_p(r_{n, 4}), ~~\norm{\what{\theta}_0(\cdot) - \theta_0(\cdot)}_{2, P_0} = O_p(r_{n, 5}), 
\end{align*}
where $r_{n,j}$ depend on the model assumptions and estimation method. We
assume $0 < \epsilon \le e_1(x) \le 1-\epsilon$, so 
$\|\cdot\|_{2,P_1} \asymp \|\cdot\|_{2,P_0} \asymp \|\cdot\|_{2,P}$.  Then,
$\what{\tau}^-(\cdot)$ is a consistent estimator, and
\begin{equation*}
  \norm{\what{\tau}^-(\cdot) - \tau^-(\cdot)}_{2, P} =
  O_p\left(r_{n, 1} + r_{n, 2} + r_{n, 3} + r_{n, 4} + r_{n, 5}\right).
\end{equation*}

Under assumptions stated in Appendix~\ref{sec:sieve-method}
(\ref{assumption:holder-smooth}--\ref{assumption:lebesgue-equiv}, including
that $\theta_z$ belongs in a $\holdersmooth$-smooth H\"{o}lder space), our
sieve estimators~\eqref{eqn:opt-emp} for $\theta_z$ achieves the asymptotic
convergence rate
\begin{equation*}
  \norm{\what{\theta}_z - \theta_z}_{2, P_z} = \widetilde{O}_P\left( n^{-\frac{p}{2p+d}} \right) ~~z \in \{0, 1\},
\end{equation*}
where the notation $\widetilde{O}_P(\cdot)$ hides logarithmic factors.  Under
similar smoothness and regularity assumptions, \citet{chen1999improved} establish
that sieve estimators $\what{e}_z(\cdot)$ and
$\what{\mu}_{z, z}(\cdot)$ for $e_z$ and $\mu_{z,z}$ can also achieve a convergence rate of
$r_{n,j}=\widetilde{O}(n^{-\frac{\holdersmooth}{2\holdersmooth + \covdim}}).$ Consequently,
$$\norm{\what{\tau}^-(\cdot) - \tau^-(\cdot)}_{2, P} = \widetilde{O}_P\left(n^{-\frac{p}{2p+d}}\right),$$
where the convergence rates reflect typical behavior of (minimax optimal) non-parametric estimators of a regression
function~\cite{newey1997convergence,Stone80}.
These constitute the high order
terms of the approximation error for estimating the CATE $\tau(x)$ without
unobserved confounding~\cite{kunzel2017meta}, if the smoothness of the CATE $\tau(\cdot)$ is of a similar order to the individual parameters $\theta_z(\cdot)$, $\mu_z(\cdot)$ and $e_z(\cdot)$. Interesting future work would be to develop a method that adapts to the complexity of $\tau^{-}(\cdot)$, itself, as done by \citet{nie2017learning} and \citet{kennedy2020optimal}.







\section{Bounds on the Average Treatment Effect}
\label{sec:semiparametric}

Given the bounds developed in Section~\ref{sec:cate-sensitivity} for the
conditional average treatment effect $\tau(\cdot)$, we now turn to bounding
the average treatment effect (ATE) $\tau$ by marginalizing over $X$
\begin{align}
\tau^- & \defeq \E[\tau^-(X)] = \E\left[ \mu_1^-(X) - \mu_0^+(X) \right].
\label{eq:ate-lower-lazy}
\end{align}
Because $\tau^-(x) \le \tau(x)$ for any
$x$, $\tau^- \le \tau$ is a lower bound of the ATE. Rewriting $\tau^-$
as
\begin{equation}
  \tau^- = \mu_1^- - \mu_0^+
  ~~ \mbox{where} ~~
  \Bigg\{\begin{array}{l}
    \mu_1^- = \E[ \mu_1^-(X)] = \E\left[ Z Y(1) + (1-Z)\funcparam_1(X)
      \right] \\
    \mu_0^+ = \E[ \mu_0^+(X)] = \E\left[ (1-Z)Y(0) +
      Z\funcparam_0(X) \right],
  \end{array}
  \label{eq:one-side-functional}
\end{equation}
we estimate $\mu_1^-$ and $\mu_0^+$ separately and combine the resulting
estimators.

In Section~\ref{sec:semiparametric-method}, we construct a
semiparametric estimator for the bound $\tau^-$ that is conceptually similar to the AIPW estimator under unconfoundedness. We show in Section~\ref{sec:semiparametric-convergence} that it achieves
$\sqrt{n}$-consistency even when (nonparametric) estimates of the nuisance
parameters (e.g.  $e_z(\cdot)$, $\mu_{z, z}(\cdot)$, $\funcparam_1(\cdot)$)
only converge at slower rates. We focus on lower bounds for the
potential outcome $Y(1)$ as other cases are symmetric. We conclude our theoretical discussion by complementing our methodological
developments with optimality guarantees (Section~\ref{sec:hypothesis-test}). We show that our approach is
asymptotically unimprovable for testing a null of no treatment effect and
unobserved confounding satisfying the $\Gamma$-\cornfield{} condition against a
positive alternative.

\subsection{Estimation procedure}
\label{sec:semiparametric-method}


We construct a score $T(V; \eta)$ to estimate $\mu_1^-$ 
similar to the AIPW estimator of the ATE in the absence of unobserved confounding,  where $V = (X, Y, Z)$ and $\eta$ represents a set of nuisance parameters defined below. The score $T(V; \eta)$ comes from calculating the semiparametric influence function for $\mu_1^-$ from representation in \eqref{eq:one-side-functional} using the method described by \citet{newey1994asymptotic}, and augmenting the representation with the influence function. To this end, by computing the pathwise derivative of the functional in \eqref{eq:one-side-functional} with respect to a parametric subfamily of the nonparametric model, and matching to the form derived by \citet{newey1994asymptotic}, we see that the remaining term in the influence function is
\begin{equation*}
    \alpha_1(V; \theta_1, e_1, \nu_1) = \treatmentrv \frac{\psi_{\theta_1(X)}(Y)(1-e_1(X))}{\nu_1(X)e_1(X)},
\end{equation*}
which depends on the nuisance parameters $\theta_1(x)$ and $e_1(x)$, and a new nuisance parameter,
\begin{equation}
  \label{eqn:weight-normalization-def-a}
  \nu_1(x)=P(Y \ge \funcparam_1(x) \mid Z=1, X=x) +
  \Gamma P(Y < \funcparam_1(x) \mid Z=1, X=x),
\end{equation}
which serves as a weight
normalization factor. In this, the function $\psi_\theta(y)$ refers to the one defined in Eq.~\eqref{eq:psi-defn}. Adding the term $\alpha_1(V; \eta)$ to the representation in \eqref{eq:one-side-functional} gives the augmented score
\begin{equation}
  T(X, Y, Z; \theta_1, e_1, \nu_1) \defeq  \treatmentrv Y + (1-\treatmentrv)\funcparam_1(X) + \treatmentrv \frac{\psi_{\theta_1(X)}(Y)(1-e_1(X))}{\nu_1(X)e_1(X)},
  \label{eq:aug-score}
\end{equation}
that we use for estimation.
We have $E_P[T(X, Y, Z; \theta_1, e_1, \nu_1)] = \mu_1^-$ since
$\E_P[\psi_{\theta_1(X)}(Y) \mid Z=1, X] = 0.$ 
By virtue of its augmented form, the score $T(\cdot; \cdot)$ is
insensitive to estimates in the nuisance parameters, formalized by the Neyman
orthogonality condition~\citep{Neyman59}:
\begin{definition}
  \label{def:neyman}
  Let $Q, \eta \mapsto \E_Q[S(V; \eta)]$ be a statistical functional with $Q$ a distribution over $V$, and nuisance
  parameter $\eta \in \Lambda$, where we take $\Lambda$ to be a subset of a
  normed vector space containing the true nuisance parameter $\eta_0$. The
  score $S$ is Neyman orthogonal at $P$ if for all $\eta \in \Lambda$,
  the derivative $\frac{d}{dr} S(P; \eta_0 + r(\eta-\eta_0))$
  exists for $r \in [0, 1)$, and is zero at $r = 0$.
\end{definition}
As \citet[Section 2.2.5]{chernozhukov2018double} shows, a score formed by
adding the influence function adjustment $\alpha_1(v)$ from the pathwise derivative as in
\citet{newey1994asymptotic} is Neyman orthogonal. Therefore, we
expect Neyman orthogonality of the functional~\eqref{eq:aug-score} constructed
in this way; we verify this formally in the proof of
Theorem~\ref{thm:semiparametric} in the Supplementary Materials.

We construct a semiparameteric plug-in estimator for the augmented functional,
and show that estimation errors of the nuisance parameters multiply to reduce
their influence on our final estimator. Concretely, we prove that our
augmented esitmator preserves $\sqrt{n}$ consistency provided that our
nuisance estimates converge at a rate of $o_P(n^{-1/4})$ in $\|\cdot\|_{2,P}$
norm.  This draws important connections to the classical doubly-robust
AIPW estimator under no unobserved
confounding. Recalling the definitions~\eqref{eq:mean-function} and
\eqref{eq:prop-score} of $\mu_{z,z}(x)$ and $e_1(x)$ (respectively), the standard AIPW
estimator for $\E[Y(1)]$ is
\begin{equation}
  \label{eq:aipw}
  \what{\mu}_{1,\text{AIPW}} = \frac{1}{n} \sum_{i=1}^{n}\left[
    \what{\mu}_{1,1}(X_i) + \frac{Z_i}{\what{e}_1(X_i)}\left(Y_i -
    \what{\mu}_{1,1}(X_i)\right)\right].
\end{equation}
Assuming all confounding variables are observed \eqref{eq:indepsim}, the
AIPW~\eqref{eq:aipw} is an asymptotically efficient estimator of
$\mu_1$~\cite{HiranoImRe03}.  The AIPW also satisfies the Neyman orthogonality
condition, which \citet{chernozhukov2018double} used to show that the AIPW
estimator~\eqref{eq:aipw} with cross-fitting (described below) enjoys the root $n$ rate so long as the nuisance
parameters can be estimated at the rate $o_p(n^{-1/4})$.  Our approach
generalize the AIPW estimator~\eqref{eq:aipw} under the $\Gamma$-\cornfield{}
condition, and reduces to the AIPW when $\Gamma = 1$.

We use an efficient sample-splitting recipe for constructing an augmented
estimator for $\mu_1^-$ by adapting \citet{chernozhukov2018double}'s cross-fitting
meta-procedure for Neyman-orthogonal functionals to our augmented score $T(\cdot)$. Letting $K \in \naturals$ be the number of
folds for cross-fitting, randomly split the data into $K$ folds of
approximately equal size. With slight abuse of notation, let $\mathcal{I}_k$
be the indices corresponding to the observations in the $k$-th part as well as the
corresponding observation themselves.

For each $k$, using the sample $\mathcal{I}_{-k}$ of observations \emph{not} in the $k$-th fold, we compute
\begin{enumerate}
\item an estimator of $\theta_1(x)$, denoted by $\what\funcparam_{1,k}(x),$ using
  the procedure described in Section~\ref{sec:cate-sensitivity}
\item  an estimator of $e_1(x)$, denoted by $\what{e}_{1,k}(x)$, and
  let $\what{e}_{0, k}(x)=1-\what{e}_{1, k}(x);$
\item an estimator of $\nu_1(\cdot)$, denoted by $\what\nu_{1,k}(\cdot)$, using the procedure described in Section~\ref{sec:nu-est}
\end{enumerate}
Estimating $\nu_{1}(\cdot)$ in the last step is more involved, as it depends
on $\theta_1(\cdot)$, so we defer the construction of $\what\nu_{1,k}(\cdot)$
to Section~\ref{sec:nu-est}. Under appropriate regularity conditions---for
example, sufficient smoothness of $\funcparam_1(x)$, $e_1(x)$, and
$\nu_1(x)$---these estimators attain $o_P(n^{-1/4})$ convergence in
$\norm{\cdot}_{2,P}$.
%
In the end, our proposed cross-fitting estimator of $\mu_1^-$ is
\begin{align}
  \label{eq:orthogonal-estimator}
  \what{\mu}_1^- =
  \frac{1}{n} \sum_{k=1}^K \sum_{i\in\mathcal{I}_k} \left\{\treatmentrv_i Y_i + (1-\treatmentrv_i)\what{\funcparam}_{1,k}(X_i) + \treatmentrv_i \frac{\psi_{\what{\funcparam}_{1,k}(X_i)}(Y_i)\what{e}_{0,k}(X_i)}{\what{\nu}_{1,k}(X_i)\what{e}_{1,k}(X_i)}\right\},
\end{align}
with an estimator $\what{\mu}_0^+$ for $\mu_0^+$ constructed similarly.  This
estimator is natural; when $\Gamma = 1$, we recover the cross-fitting version
of the standard doubly robust AIPW estimator~\eqref{eq:aipw}. While the
estimator satisfies the orthogonality conditions of
\citet{chernozhukov2018double} that imply a form of local robustness for
$\widehat{\theta}_1(\cdot)$ near $\theta_1(\cdot)$, we explain below why it isn't doubly robust.

\subsection{Asymptotic properties and inference}
\label{sec:semiparametric-convergence}

To establish asymptotic normality of $\what{\mu}_1^-$, we require a few
assumptions. Consistency of $\what{\mu}_1^-$ follows from weak regularity
conditions and the consistency of $\what{\theta}_1(\cdot)$, which we address
via Assumption~\ref{assumption:bounded-variance}. Asymptotic normality
requires stronger conditions (Assumptions~\ref{assumption:problem-regularity}
and~\ref{assumption:nuisance-est}), in turn allowing us to establish
Theorem~\ref{thm:semiparametric} on the asymptotic normality of
$\what{\mu}_1$.

\begin{assumption}
  \label{assumption:bounded-variance}
  There exist $\epsilon>0$ and  $0<\constlow <\constup $
  such that (a) $\E[|Y(1)|] < \infty,$ (b) $\|
  \what{\theta}_{1}(\cdot) - \theta_1(\cdot) \|_{1,P} \cp 0$, (c)
  $e_1(X) \in [\epsilon, 1-\epsilon]$ almost surely, (d) $\P([\essinf \what{e}_1(X), \esssup
    \what{e}_1(X)] \subset [\epsilon, 1-\epsilon]) \to 1,$ and (e) $\P\left(\constlow
  \le \what{\nu}_1(x)
  \le \constup ~ \mbox{for~all~}x\right) \to 1$.
\end{assumption}

Assumption~\ref{assumption:bounded-variance}(a-c) are
slightly stronger than the usual assumptions for justifying consistency of
the AIPW estimator for the ATE $\tau$ in the absence of unobserved confounding
\cite{bang2005doublerobust,chernozhukov2018double}.
When $\|\what{e}_1(\cdot) - e_1(\cdot)\|_{\infty,P} \cp 0$,
Assumption~\ref{assumption:bounded-variance}(c) implies
Assumption~\ref{assumption:bounded-variance}(d), and similarly when
$\|\what{\nu}_1(\cdot) - \nu_1(\cdot)\|_{\infty,P} \cp 0$, $\nu_1(x) \in [1,
  \Gamma]$ implies Assumption~\ref{assumption:bounded-variance}(e).



Assumption~\ref{assumption:bounded-variance}(b) is necessary, and cannot be removed by alternatively assuming consistency of the other nuisance parameters. The $\alpha(V; \eta)$ with the true $\theta_1(\cdot)$ plugged in has mean zero regardless of the nominal propensity score used. In this case, the proposed estimator is consistent in estimating $\mu_1^-.$  However, if an incorrect $\theta_1(\cdot)$ is plugged in to $T(V; \eta),$ straightforward computation shows that $\E[T(V; \eta)]$ depends on the $\theta_1(\cdot)$ plugged in, even with the correct $e_1(\cdot)$ and $\nu_1(\cdot)$. Therefore, $\hat{\mu}_1^-$ is not globally doubly robust; the Neyman orthogonality condition only guarantees a local form of robustness.


\begin{restatable}{thm}{thmconsistency}
  \label{thm:consistency}
  Under Assumption~\ref{assumption:bounded-variance}, the
  estimator~\eqref{eq:orthogonal-estimator} satisfies
  $\what{\mu}_1^- \cp \mu_1^-$.
\end{restatable}
\noindent See the Supplementary Materials (Section~\ref{sec:proof-consistency}) for the proof.
We now turn to stronger regularity assumptions for the weak convergence of $\widehat{\mu}_1^-.$

\begin{assumption}
  \label{assumption:problem-regularity}
  (a) There exist $q > 2$, and $C_q < \infty$ such that
  $\E[|Y(1)|^q] \le C_q$, and (b) $Y(1)$ has a conditional density
  $p_{Y(1)}(y \mid X=x, \treatmentrv=1)$ with respect to the Lebesgue measure
  and $\sup_{x,y} p_{Y(1)}(y \mid \treatmentrv=1, X=x) < \infty$.
\end{assumption}

\begin{assumption}
  \label{assumption:nuisance-est}
  $\what{\eta}_1 = (\what{\funcparam}_1, \what{\nu}_1, \what{e}_1)$ is a
  consistent estimator of $\eta_1 \defeq (\funcparam_1, \nu_1, e_1)$ and (a)
  $\| \what{\eta}_1(\cdot) - \eta_1(\cdot) \|_{2,P}= o_P(n^{-1/4})$,
  (b)~$\| \what{\eta}_1(\cdot) - \eta_1(\cdot) \|_{\infty,P} = O_P(1)$.
\end{assumption}

Assumptions~\ref{assumption:problem-regularity} (a) is no stronger than the
standard regularity conditions needed for existence of asymptotically normal estimators of the ATE without
unobserved confounding~\cite{chernozhukov2018double}. Assumption
\ref{assumption:problem-regularity}(b) ensures that the term
$\theta(\cdot) \mapsto \E[ Z\psi_{\theta(X)}\{Y(1)\} \mid X]$
is sufficiently smooth to control fluctuations due to estimating
$\theta_1(\cdot)$.  Inspection of the proof of
Theorem~\ref{thm:semiparametric} to come shows that we may relax
Assumption~\ref{assumption:problem-regularity}(b): if $\theta_1(x)$ and
$\what{\theta}_1(x)$ have range ${\cal A}_1(x)$, we may replace
\ref{assumption:problem-regularity}(b) with
\begin{equation}
  \label{eqn:condition-replacing-density}
  \esssup_{X} \sup_{y \in {\cal A}_1(X)} p_{Y(1)}(y \mid
  \treatmentrv=1, X) < \infty,
\end{equation}
which is satisfied, eg., when the outcome $Y(z)$ is binary and
$P(Y(z) = y \mid \treatmentrv = z, X)<1$ for $y \in \{0,1\}$, because
$\what{\theta}(X) \in (0, 1)$ eventually and $p(y \mid Z=1, X) = 0$ for
$y \not \in \{0,1\}$.

The convergence rate conditions for estimating nuisance parameters in
Assumption~\ref{assumption:nuisance-est} are relatively standard in
semi-parametric estimation
\cite{newey1994asymptotic,chernozhukov2018double}, but nonetheless this theoretical requirement can be restrictive and hard to achieve to certain applications. For example, while for $e_1(\cdot),$ the
conditional mean of observed random variables, a variety of methods can provide
$o_P(n^{-1/4})$ consistency, they still require the data generating distribution to meet appropriate conditions and the sample size to be large relative to the dimension of covariate
\cite{wager2017estimation,chernozhukov2018double}.  The estimators
$\what{\theta}_1(\cdot)$ from Section~\ref{sec:cate-sensitivity} and
$\what{\nu}_1(\cdot)$ from Section~\ref{sec:semiparametric-method} achieve the
convergence rates in Assumption~\ref{assumption:nuisance-est} under
appropriate smoothness conditions on $\theta_1(\cdot)$ and $\nu_1(\cdot).$ For
instance, if Assumptions~\ref{assumption:holder-smooth},
\ref{assumption:bdd-error}, and \ref{assumption:lebesgue-equiv} hold with
$p > d/2$, then Theorem~\ref{thm:sieve} shows that estimating
$\funcparam_1(x)$ as in Section~\ref{sec:cate-sensitivity} with linear sieves
(see Examples~\ref{example:polynomials} and \ref{example:splines}) will
satisfy Assumption~\ref{assumption:nuisance-est}. Section~\ref{sec:nu-est}
provides an efficient enough estimator of $\nu_1(\cdot)$ when $p > d/2$.
Under these assumptions, the following theorem gives the asymptotic
distribution of the estimator $\what{\mu}_1^-$
in~\eqref{eq:orthogonal-estimator}, with asymptotic variance
\begin{equation*}
  \sigma_1^2 \defeq \var\left[
    \treatmentrv Y + (1-\treatmentrv)\funcparam_1(X) + \treatmentrv
    \frac{\psi_{\funcparam_1(X)}(Y)(1-e_1(X))}{\nu_1(X)e_1(X)}
  \right].
\end{equation*}
We use the following consistent estimator of the asymptotic variance
\begin{equation*}
\what\sigma_1^2 \defeq \frac{1}{n}\sum_{k=1}^K\sum_{i\in \mathcal{I}_k}\left[ \treatmentrv_i Y_i + (1-\treatmentrv_i)\what{\funcparam}_{1,k}(X_i) + \treatmentrv_i \frac{\psi_{\what{\funcparam}_{1,k}(X_i)}(Y_i)\what{e}_{0,k}(X_i)}{\what{\nu}_{1,k}(X_i)\what e_{1,k}(X_i)} -\what{\mu}_1^-\right]^2.
  \end{equation*}
\begin{restatable}{thm}{semiparametricnormality}
  \label{thm:semiparametric}
  Let Assumptions~\ref{assumption:bounded-variance},
  \ref{assumption:problem-regularity}, and \ref{assumption:nuisance-est} hold.
  Then, $\what{\mu}_1^-$ given in Eq.~\eqref{eq:orthogonal-estimator} is
  asymptotically normal with $ \sqrt{n}(\what{\mu}_1^- - \mu_1^-)
  \overset{d}{\to} \normal(0, \sigma_1^2)$.
  Furthermore, $\what{\sigma}_1^2 \cp \sigma_1^2$, and
  $ \frac{\sqrt{n}}{\what{\sigma_1}} (\what{\mu}_1^- - \mu_1^-)
  \overset{d}{\to} \normal(0, 1)$.
\end{restatable}
\noindent See Section~\ref{sec:proof-semiparametric} in the Supplementary Materials for a
proof. To bound $\tau$ from below, let
$$\what{\tau}^- = \what\mu_1^- - \what\mu_0^+, $$
where $\tilde{\psi}_\theta(y)=\Gamma (y-\theta)_+-(y-\theta)_-$,
$$\what\mu_0^+=\frac{1}{n}\sum_{k=1}^K \sum_{i \in {\cal I}_k} \left[(1-\treatmentrv_i)Y_i + \treatmentrv_i \what\funcparam_{0,k}(X_i) + (1-\treatmentrv_i) \frac{\tilde{\psi}_{\what{\funcparam}_{0,k}(X_i)}(Y_i)\what{e}_{1,k}(X_i)}{\what{\nu}_{0,k}(X_i)\what{e}_{0,k}(X_i)}\right],$$
 and  $\what{\nu}_{0,k}(\cdot)$ is the nonparametric estimator of 
  $$\nu_0(X)=P(Y \le \funcparam_0(X) \mid Z=0, X) + \Gamma P(Y > \funcparam_0(X) \mid Z=0, X)$$
  based on data in ${\cal I}_{-k}.$ A simple
extension of Theorem~\ref{thm:semiparametric} shows 
$$\sqrt{n}(\what{\tau}^- - \tau^-) \to \normal(0,
\sigma_{\tau^-}^2),$$
as $n \rightarrow \infty,$ where
\begin{equation}
 \sigma_{\tau^-}^2 \defeq \var \begin{aligned}[t] \Bigg[&
\treatmentrv Y + (1-\treatmentrv)\funcparam_1(X) + \treatmentrv \frac{\psi_{\funcparam_1(X)}(Y)e_0(X)}{\nu_1(X)e_1(X)} \\
&-(1-\treatmentrv) Y - \treatmentrv \funcparam_0(X) - (1-\treatmentrv)\frac{\tilde{\psi}_{\funcparam_0(X)}(Y)e_1(X)}{\nu_0(X)e_0(X)}
\Bigg].
 \end{aligned}
 \label{eq:var-est-tau-hat-minus}
\end{equation}
 Furthermore, a consistent estimator of the variance $\sigma_{\tau^-}^2$ is
\begin{align}
\label{eq:est-tau-hat-var}
  \what\sigma_{\tau^-}^2
  =  \frac{1}{n}\sum_{k=1}^{K}\sum_{i\in \mathcal{I}_k}
  \bigg[& 
\treatmentrv_i Y_i + (1-\treatmentrv_i)\what\funcparam_{1,k}(X_i) + \treatmentrv_i \frac{\psi_{\what\funcparam_{1,k}(X_i)}(Y_i)\what{e}_{0,k}(X_i)}{\what{\nu}_{1,k}(X_i)\what{e}_{1,k}(X_i)}  \\
& - (1-\treatmentrv_i)Y_i - \treatmentrv_i \what\funcparam_{0,k}(X_i) - (1-\treatmentrv_i) \frac{\tilde{\psi}_{\what{\funcparam}_{0,k}(X_i)}(Y_i)\what{e}_{1,k}(X_i)}{\what{\nu}_{0,k}(X_i)\what{e}_{0,k}(X_i)}-\what\tau^-
\bigg]^2
  \nonumber
\end{align}
and
$[\what\tau^- - z_{1-\alpha/2}\what\sigma_{\tau^-}/\sqrt{n}, \what\tau^- + z_{1-\alpha/2}\what\sigma_{\tau^-}/\sqrt{n}] $
is a $100(1-\alpha)\%$  asymptotic confidence interval for $\tau^-.$ 
The proof is \emph{mutatis mutandis} identical to that of
Theorem~\ref{thm:semiparametric}.


Importantly, our bounds define a confidence set for $\tau =
\E[Y(1)-Y(0)]$. The same approach as in Section~\ref{sec:cate-sensitivity},
but up-weighting large values of $Y(1)$ and small values of $Y(0)$, provides
an estimate $\what{\tau}^+$ of $\tau^+$ that upper bounds the ATE. The
limiting distribution of $\what{\tau}^+$ is also normal.  With these estimators,
we may construct a confidence interval for the ATE,
\begin{equation}
  \what{\operatorname{CI}}_{\tau} = \left[\what\tau^- - z_{1-\alpha/2}\frac{\what\sigma_{\tau^-}}{\sqrt{n}}, \what\tau^+ + z_{1-\alpha/2}\frac{\what\sigma_{\tau^+}}{\sqrt{n}} \right],
  \label{eq:ci4ate}
\end{equation}
where $\what{\sigma}_{\tau^+}^2$ is a consistent estimator of the variance of $\sqrt{n}(\what{\tau}^+-\tau^+).$
Because $\tau^- \le \tau \le \tau^+$, this confidence interval has
appropriate asymptotic coverage:
\begin{restatable}{cor}{corincludeate}
  \label{cor:include-ate}
  Let $P$ satisfy the $\Gamma$-\cornfield{} condition~\eqref{eq:cornfield},
  conditional independence~\eqref{eq:indep}, and
  Assumptions~\ref{assumption:bounded-variance}--\ref{assumption:nuisance-est}. Let
  $\what{\operatorname{CI}}_{\tau}$ be defined as in \eqref{eq:ci4ate}. For
  $\tau = \E[Y(1) - Y(0)]$, we have
  \begin{equation*}
    \liminf_{n \to \infty} P(\tau \in \what{\operatorname{CI}}_\tau) \ge 1-\alpha.
  \end{equation*}
\end{restatable}

\begin{remark}
  It is possible to extend Theorem~\ref{thm:semiparametric} to provide
  confidence intervals uniform over $\mc{P}$. In other words, the coverage
  probability of the relevant confidence intervals converge to the desired
  level uniformly over all the distributions in $\mc{P}.$ To do so,
  Assumption~\ref{assumption:nuisance-est} must be uniform over a class of
  distributions $\mc{P}$ satisfying
  Assumption~\ref{assumption:problem-regularity}, for instance by assuming
  there exists sequences $\Delta_n \to 0$ and $\delta_n \to 0$ such that
  \begin{equation*}
    \sup_{P \in \mc{P}} P\left(\|\what{\eta}_1(\cdot) - \eta_1(\cdot)\|_{2,P} > n^{-1/4}\delta_n\right) < \Delta_n.
  \end{equation*}
  Previous work~\cite{chen2015optimal} shows that series estimators for
  the conditional regression function (example 1 in
  Section~\ref{sec:cate-sensitivity}) converge uniformly; extending these
  results to the estimation of $\theta_1(\cdot)$ and $\nu_1(\cdot)$ is beyond the
  scope of the present work.
\end{remark}

\subsection[Estimation of nu]{Construction of
  $\what\nu_{1,k}(\cdot)$ and its asymptotic properties}
\label{sec:nu-est}

\newcommand{\lossnu}{\bar{\loss}_{\Gamma}}

The above results assumed access to a well-behaved estimate of the
the weighted probability
$\popprob(X) = 1 + (\Gamma - 1)\P(Y(1) \ge \popfunc(X)
\mid Z=1, X)$. Here, we describe a nonparametric estimator via a
loss function: defining
\begin{equation*}
  \lossnu(\probfunc, \funcparam, \outcome) \defeq \half \left[
    1 + (\Gamma - 1)\indic{\outcome \ge \funcparam} -
    \probfunc  \right]^2,
\end{equation*}
$\popprob$ uniquely solves the optimization problem
\begin{equation}
  \label{eqn:opt-prob}
  \minimize_{\probfunc(\cdot)~\textup{measurable}}
  ~ \E[\lossnu\left\{\probfunc(X), \popfunc(X), Y(1) \right\} \mid Z=1].
\end{equation}
The natural sieve estimator for $\popprob(\cdot)$ minimizes the empirical
version of \eqref{eqn:opt-prob} under finite-dimensional sieves. However, this
requires knowledge of $\popfunc(\cdot)$, which itself must be
estimated. Therefore, consider the following (nested) cross-fitting approach:
\begin{enumerate}[1.]
\item Partition the sample $\mathcal{I}_{-k}$ into two
  \emph{independent} sets, $\mathcal{I}_{-k,1}$ and $\mathcal{I}_{-k,2}$.
\item Let $\what{\theta}_{1k}^{\nu_1}(\cdot)$ be an estimator of
  $\popfunc(\cdot)$ based on the first subset $\mathcal{I}_{-k,1};$
\item For a sequence of sieve parameter spaces $\probspace_1 \subseteq
  \cdots \subseteq \probspace_{n} \subseteq \cdots \subseteq \probspace,$
  estimate $\what{\nu}_{1,k}$ minimizing
  the plug-in version of the
  population problem~\eqref{eqn:opt-prob},
  \begin{equation}
    \label{eqn:opt-emp-prob}
    \minimize_{ \probfunc(\cdot) \in 1 + (\Gamma -1)\probspace_{n}}
    ~ \E_{n,2}^{(k)}\left[\lossnu
      \big(\probfunc(\covariaterv), \what{\theta}_{1k}^{\nu_1}(\covariaterv),
      Y \big) \mid Z = 1 \right],
  \end{equation}
\end{enumerate}
where $\E_{n,2}^{(k)}$ is the empirical expectation with respect to the second subset $\mathcal{I}_{-k,2}$.

When $\nu_1(X)$ belongs in a $\probsmooth$-smooth H\"{o}lder space, in
Proposition~\ref{prop:prob-sieve} in the Supplementary Materials, we prove that the empirical
solution $\what{\nu}_{1,k}(\cdot)$ to the problem~\eqref{eqn:opt-emp-prob}
achieves the minimax optimal nonparametric rate (up to logarithmic factors)
\begin{equation*}
  \norm{\what{\nu}_{1, k}(\cdot) - \nu_1(\cdot)}_{2, P_1} = O_P\left( \left(\frac{\log n}{n} \right)^{\frac{\probsmooth}{2\probsmooth+d}} \right).
\end{equation*}
If $\probsmooth>d/2,$ then
$\norms{\what{\nu}_{1,k} - \nu_1}_{2, P}=o_P(n^{-1/4})$, satisfying the
assumptions in Theorem~\ref{thm:semiparametric}. We defer a rigorous treatment
to Appendix~\ref{sec:nu-sieve} as our results heavily build on the standard
theory of sieve estimation~\cite{Chen07}. In
Proposition~\ref{prop:prob-sieve}, we demonstrate sufficient conditions
for the convergence of $\what{\nu}_{1,k}$ needed for the lower bound
estimator~\eqref{eq:orthogonal-estimator} and its asymptotic normality via
Theorem~\ref{thm:semiparametric}: with sufficient smoothness of $\nu_1$, it is
possible to efficiently estimate lower and upper bounds on the average
treatment effect.

\subsection[Design sensitivity and optimality of the estimator for tau-]{Design sensitivity and optimality of our bound on the ATE}
\label{sec:hypothesis-test}

We complement our methodological development so far with optimality results
for our worst-case bounds. By construction, our approach yields a tight bound on the mean of each unobserved potential outcome. We extend these results to the ATE by constructing an instance where our bound is tight. That is, we construct a family of data generating distributions such that whenever our bounds cannot infer the sign of the ATE, the ability to test whether or not the ATE is positive is intrinsically difficult. To this end, we study a pointwise asymptotic level $\alpha$
hypothesis test for the composite null
\begin{equation}
  H_0(\Gamma) : \E[Y(1)] \le \E[Y(0)] ~~\mbox{and the}~
    \Gamma\text{-\cornfield{} condition~\eqref{eq:cornfield} holds}
  \label{eq:gaussian-null-regular}
\end{equation}
under
Assumptions~\ref{assumption:bounded-variance}--\ref{assumption:nuisance-est},
and analyze its design sensitivity~\cite{Rosenbaum10}. Let $H_1 : Q$ be an
alternative with a positive average treatment effect $\tau = \E_Q[Y(1) - Y(0)] > 0$
and no confounding ($\Gamma=1$ in Eq.~\eqref{eq:cornfield}).  Let
$t_n^\Gamma = t_n^\Gamma\{(Y_i, Z_i, X_i)_{i=1}^n\} \in \{0, 1\}$ be a
pointwise asymptotic level
$\alpha$ test for the null hypothesis~\eqref{eq:gaussian-null-regular}, where $t_n^\Gamma=1$, if the null hypothesis $\tau\le 0$ is rejected.   The
\emph{design sensitivity}~\cite{Rosenbaum10,rosenbaum2011new} of the
sequence $\{t_n^\Gamma\}$ is the threshold $\Gamma_\design$ such
that the power $Q(t_n^\Gamma = 1) \to 0$ for $\Gamma > \Gamma_{\design}$ and
the power $Q(t_n^\Gamma = 1) \to 1$ for $\Gamma <
{\Gamma}_\design$. In other words, if the selection bias satisfies $\Gamma >
\Gamma_{\design}$, the test cannot differentiate the alternative $\tau>0$ from the null $\tau\le 0$ regardless of the sample size;
if $\Gamma < \Gamma_{\design}$, the test always rejects the null under the alternative $Q$ for sufficiently large $n$ (we define $\Gamma_{\design}=\infty$ when no such threshold exists). Given the
confidence interval for $\tau$ described in Section~\ref{sec:semiparametric-convergence},
a natural asymptotic level $\alpha$ test for $H_0(\Gamma)$,
the hypothesis~\eqref{eq:gaussian-null-regular},
is
\begin{equation}
  \label{eq:the-test}
  \psi_n^\Gamma\{(Y_i, Z_i, X_i)_{i=1}^n\}
  \defeq \ind{\what\tau^- > z_{1-\alpha}\frac{\what\sigma_{\tau^-}}{\sqrt{n}} }.
\end{equation}

We consider the design sensitivity of $\psi_n^\Gamma$ in the simplified
setting without covariates, which allows us to demonstrate its optimality. In
this case, $\{Y(0), Y(1)\} \independent Z \mid U$, the simplified
$\Gamma$-\cornfield{} condition \eqref{cond:simple-gamma-selection} holds,
$ \{Y(0), Y(1)\}\independent Z$ under the alternative $Q$ (recall
Eq.~\eqref{eq:indepsim}), and $\theta_1,\theta_0 \in \R$ are constants defined
in Eq.~\eqref{eq:population-theta} and
Eq.~\eqref{eq:population-mu-0-upper}. 



\begin{restatable}{prop}{thmdesignsensitivity}
  \label{prop:design-sensitivity-model}
  Let $\psi_n^\Gamma$ be defined as in Eq.\eqref{eq:the-test}, so that
  $\psi_n^\Gamma$ is asymptotically level $\alpha$ for $H_0(\Gamma)$ in
  \eqref{eq:gaussian-null-regular}. For an alternative $H_1 = \{Q\}$, define
  \begin{equation*}
    \tau^-(\Gamma)
    \defeq
    \E_Q[Z Y(1) + (1 - Z) \theta_1 - (1 - Z) Y(0) - Z \theta_0],
  \end{equation*}
  where $\theta_1,~\theta_0$ solve
  ~\eqref{eq:population-theta} and ~\eqref{eq:population-mu-0-upper}, respectively, at level $\Gamma$
  for the distribution $Q.$
  Then, either the design sensitivity $\Gamma_\design$ of
  $\psi_n^\Gamma$ is infinite or it uniquely solves the equation
  $\tau^-(\Gamma) = 0$.
\end{restatable}

\noindent See Section~\ref{sec:proof-design-sensitivity} in the Supplementary Materials for proof.  While
there is no simplified expression for $\Gamma_\design$ in general, it can be derived explicitly for some special alternatives $Q$. For instance, in
Supp. Materials, Section~\ref{sec:proof-gaussian-design-sensitivity}, we prove the following
result for Gaussian alternatives.
\begin{restatable}{cor}{cordesignsensitivity}
  \label{cor:designsensitivity}
  Let $\psi_n^\Gamma$ be as in Eq.~\eqref{eq:the-test}.  For the alternative $H_1(Q):$
  $$\left\{Y(1) \distas \normal\left(\frac{\tau}{2}, \sigma^2\right),Y(0) \distas
  \normal\left(-\frac{\tau}{2}, \sigma^2\right), Z\distas \bernoulli(\frac{1}{2}) \right\},$$
  $\psi_n^\Gamma$ has design sensitivity
  \begin{equation}
    \label{eq:design-sensitivity-gaussian}
    \Gamma_\design^\gauss
    \defeq
    -\frac{\int_0^\infty y \exp\left(- \tfrac{(y - \tau)^2}{2\sigma^2}\right)
      \dif{y}}{
      \int_{-\infty}^0 y \exp\left(- \tfrac{(y - \tau)^2}{2\sigma^2}\right)
      \dif{y}}
    = \frac{\phi(\frac{\tau}{\sigma}) + \frac{\tau}{\sigma} \Phi(\frac{\tau}{
        \sigma})}{
      \phi(\frac{\tau}{\sigma}) -
      \frac{\tau}{\sigma} \Phi(\frac{\tau}{\sigma})},
  \end{equation}
  where $\Phi$ and $\phi$ denote the standard Gaussian
  CDF and density, respectively.
\end{restatable}
  
The next proposition shows that the test $\psi_n^\Gamma$ is optimal for alternative $H_1(Q)$ given in Corollary \ref{cor:designsensitivity},
as any asymptotic level $\alpha$ test of $H_0(\Gamma)$ has
design sensitivity $\ge \Gamma^\gauss_\design$
(see Supplementary Materials Section~\ref{sec:proof-opt-design-sensitivity} for proof).
\begin{restatable}{prop}{optdesignsensitivity}
  \label{prop:opt-design-sensitivity}
  Let $H_0(\Gamma)$ be as in \eqref{eq:gaussian-null-regular}. There exists
  $a \in [1/(1+\sqrt{\Gamma}), \sqrt{\Gamma}/(1+\sqrt{\Gamma})]$ such that
  for the alternative $H_1 (Q):$
  $$\left\{Y(1) \distas \normal\left(\frac{\tau}{2},
      \sigma^2\right),Y(0) \distas \normal\left(-\frac{\tau}{2},
      \sigma^2\right), Z\distas \bernoulli(a) \right\},$$ if
  $\Gamma \ge \Gamma_\design^\gauss$, there exists a probability measure
  $P \in H_0(\Gamma)$ for $\{Y(1), Y(0), Z, U\},$ such that for all
  $n \in \N$, all tests $t_n$, and $(Y_i, Z_i)$ i.i.d.,
  \begin{equation*}
    P(t_n\{(Y_i, Z_i)_{i=1}^n\} = 1) = Q(t_n\{(Y_i, Z_i)_{i=1}^n\} = 1).
  \end{equation*}
\end{restatable}


\begin{remark}
  Our proof uses a specific choice of $a$ to simplify the algebra;
  solving a system of nonlinear equations for the
  distribution of $P_{Z|U}$ allows for any marginal
  $P(Z=1)$.\vspace{-1em}
\end{remark}

\begin{remark}
  The above optimality results for $\psi_n^\Gamma$ extend to alternatives beyond Gaussian distributions, so long as $Y(0) \overset{d}{=} C(1 - Y(1))$, for some constant
  $C > 0.$ The proof relies on this symmetry in the potential outcomes to
  construct a distribution under $H_0(\Gamma)$ matching $Q$ over the observed
  data, $\{(Y_i(Z_i), Z_i), i=1,\cdots, n\}.$ This symmetry is 
  unnecessary if one is interested in the mean (or conditional mean) of a
  single potential outcome $\E[Y(1)]$ (or $\E[Y(1) \mid X=x]$, in which case
  the test $\psi_n^\Gamma$ achieves the optimal design sensitivity for any
  alternative for which the proposed method is consistent.
\end{remark}


\vspace{-20pt}
\section{Numerical experiments}
\label{sec:experiments}

To complement our theoretical analysis in Section~\ref{sec:semiparametric}, we
examine the performance of the method using Monte-Carlo simulation and a real
dataset from an observational study examining the effect of fish consumption
on blood mercury levels. We evaluate two implementations of the methodology
developed in Sections~\ref{sec:cate-sensitivity}
and~\ref{sec:semiparametric}---one based on the sieve estimators studied in
Section~\ref{sec:nonparametric-sieves} and the other based on gradient boosted
trees fit to minimize the weighted squared loss~\eqref{eqn:opt}.

The Monte-Carlo simulations support the validity of the inference procedure in
realistic settings. We find that the semiparametric approach presented in
Section~\ref{sec:semiparametric} accurately bounds the average treatment
effect under unobserved confounding, when our assumptions about the extent of
confounding $\Gamma$ hold. We show that by using machine learning to optimize
the loss function in \eqref{eq:gamma-loss}, our method can scale to reasonably
high dimensional data. Additionally, we show that the bounds on the ATE are
tight in practice, and empirically compare their conservativeness to that of
the matching-based approach from \citet{rosenbaum2014weighted}. Finally, we
confirm our findings on a real observational study, demonstrating that our
semiparametric approach provides valid yet narrow bounds on the ATE $\tau$.

\subsection{Method Implementations}
\label{sec:implementation}

When implementing an estimator to bound the ATE $\tau$ using the method
developed in Section~\ref{sec:semiparametric}, one must choose estimators of
the nuisance parameters ${e}_z(\cdot),$ ${\theta}_z(\cdot),$ and
${\nu}_z(\cdot)$, and select their hyperparameters. In the first
implementation, we stay close to the estimators used in our theoretical
analysis with formal convergence guarantees: we estimate the propensity score
$\what{e}_1(\cdot)$ by a random forest \cite{athey2019generalized}, and
$\what\theta_z(\cdot)$ (respectively, $\what{\nu}_z(\cdot)$) by the
non-parametric estimator from Section~\ref{sec:cate-sensitivity} (respectively
Section~\ref{sec:nu-est}) using the polynomial (power series) sieve. The sieve
size and regularization were selected via 10-fold cross-validation, and then
used with 10-fold cross-fitting for the semiparametric estimation.
To estimate $\what\nu_z(\cdot)$, we use an iterative, instead of nested, form of cross-fitting that sacrifices some independence between folds to be more computationally efficient, described in Section~\ref{sec:practical-nu} of the Supplement.
Nonparametric estimation of the
propensity score $e_1(\cdot)$ leads to variability that requires weight
clipping to stabilize the semiparametric estimates
\cite{lee2011weight,tsiatis2007comment}. We clipped weights worth more than
$1/20$ of the total weight of the samples.

In one experiment below, we use a variant of this implementation where we fit
$\what{e}_{1, k}(\cdot)$ via a simple logistic regression; the logistic
regression model for the propensity score is misspecified, so the lower-order statistical bias from the Neyman orthogonality will not hold; the statistical bias of the estimator will depend on the convergence rate of the nonparametric
estimator of $\what{\theta}_{z}(\cdot)$, which will not converge
sufficiently quickly. As a result, we expect that the statistical bias will dominate the convergence of $\what\tau^{-}$ to $\tau^-$.

In the second implementation, we use \texttt{xgboost}~\cite{ChenGu16} to fit a
machine learning estimator for all of the nuisance parameters, emphasizing the
generality and scalability of our methods. \texttt{xgboost} is a gradient
boosted tree method that performs well with tabular data, despite having
little formal theory regarding its convergence guarantees. Therefore, we used
the simulations discussed below as a way to assess it's appropriateness as a
nuisance parameter estimator for our semiparametric method from
Section~\ref{sec:semiparametric-method}. In this implementation, we fit the
estimator $\what{\theta}_z(\cdot)$ to minimize the weighted squared
loss~\eqref{eqn:opt}, and fit the remaining nuisance parameters to minimize
the $\log$ loss for predicting a binary target (treatments or
the targets $\ind{Y_i \ge \theta}$ for estimating $\nu_z(\cdot)$). As with the previous
implementation, $\what\nu_{z}(\cdot)$ are fit with the iterative cross-fitting described in Section~\ref{sec:practical-nu} of the Supplement. Similarly, all tuning parameters (boosting iterations, regularization,
subsampling fraction, minimum node size) are selected via 10-fold
cross-validation. We found that when estimating a generic nuisance parameter $\eta(\cdot)$, representing either $\theta_z(\cdot), \nu_{z}(\cdot)$, or
$e_1(\cdot)$, adding an additional intercept term as follows improved performance signficantly: After fitting $\what{\eta}_z(\cdot)$ using
\texttt{xgboost}, we
fit $\beta_0$ in the model $\what{\eta}_z(X) + \beta_0$ using the
appropriate loss function for the nuisance parameter.

\subsection{Simulations}

The purpose of the simulation study is to demonstrate the good coverage of the
proposed confidence intervals for reasonable choices of sample size $n$ and
covariate dimension $d,$ and to understand some of the practical properties of
the proposed methods relative to existing methods for sensitivity analysis,
such as matching methods \cite{rosenbaum2014weighted}. In all of the
simulations, we generate the data as follows for a randomly chosen set of
coefficients $\beta$ and $\mu$: draw $X \distas \uniform[0,1]^d$, and
conditional on $X=x$, draw
\begin{equation*}
  U \distas \normal\left\{0, \left(1 + \tfrac{1}{2}\sin(2.5 x_1)\right)^2
  \right\},
  ~~~
  Y(0) = \beta^\top x + U, ~~~
  Y(1) = \tau + \beta^\top x + U.
\end{equation*}
We draw the treatment assignment according to
\begin{equation*}
  Z \distas \bernoulli\left\{\frac{\exp\left(\alpha_0+x^\top \mu +
    \log(\Gamma_{\rm data}) \ind{u > 0} \right)}{1+\exp\left(\alpha_0+x^\top \mu +
    \log(\Gamma_{\rm data}) \ind{u > 0} \right)}\right\},
\end{equation*}
where $\alpha_0$ is a constant controlling the overall treatment assignment ratio. 
This model satisfies the $\Gamma_{\rm data}$-\cornfield{} condition, since
\begin{equation*}
  \frac{P(Z = 1 \mid X=x, U=u)}{P(Z = 0 \mid X=x, U=u)}
  \frac{P(Z = 0 \mid X=x, U=\tilde{u})}{P(Z = 1 \mid X=x, U=\tilde{u})}
  = \Gamma_{\rm data}^{\ind{u>0}-\ind{\tilde{u}>0}}\in [\Gamma_{\rm data}^{-1},
  \Gamma_{\rm data}]
\end{equation*}
Across all experiments, we set $\tau = 1$ and \smash{$\Gamma_{\rm data} =
\exp(1)$}. Unless otherwise stated, we used the same $\Gamma$ in our
sensitivity analysis as the level of confounding \smash{$\Gamma_{\rm data}$} used to
generate the data. Here, unobserved confounding inflates estimates that assume
unconfoundedness: when $Z=1$, $U$ is more likely to be positive than when
$Z=0$, which inflates the mean of treated units, i.e.,
$\E[Y(1) \mid Z=1, X=x] > \E[Y(1) \mid X=x]$. We expect that the upper bound
from the sensitivity analysis is above the true ATE, while the lower bound is
only slightly below the truth, assuming that we choose $\Gamma \ge \Gamma_{\rm data}$, but not by too much. 

In the first set of simulations, we simulate data with a moderate number of
observed covariates ($d=20$), where we observe the proposed sensitivity
analysis procedure quickly approaches it's asymptotic behavior as sample size
grows. 
For these simulations, we use the \texttt{xgboost} implementation, validating
the performance of our semiparametric method when the nuisance parameters are
estimated well, even if lacking in formal convergence guarantees.

Table~\ref{tab:coverage} summarizes the empirical performance of the \texttt{xgboost} implementation based on 500 simulations. As
expected, the average lower bound estimator $\what\tau^-$ is close
to the true ATE, while the average upper bound estimator $\what\tau^+$ is
higher than the true ATE to account for unmeasured confounding.
The estimators of the standard errors of $\what\tau^-$ and
$\what\tau^+$ are fairly accurate when $n\ge 1000$. When $n$ is small, they slightly
underestimate the true standard errors.  The empirical coverage probability of the
confidence interval of ATE is conservative because of unobserved confounding.
As the unobserved confounding introduces upward bias,
the lower bound $\tau^-\approx \tau,$ and we expect that the coverage
probability of the confidence interval of $\tau$ is close to 97.5\% for
large $n$, which is confirmed by the simulation results in
Table~\ref{tab:coverage}.

\begin{table}[t]
  \caption{Simulation results of the proposed method with $20$ observed
    covariates. $\what{\tau}^-$, the empirical average of $\what{\tau}^-;$
    $\what{\sigma}_{\tau^-}$, the empirical average of
    $\what{\sigma}_{\tau^-};$ SD. of $\what\tau^-$, the empirical standard
    deviation of $\what\tau^-$; $\what{\tau}^+$, the empirical average of
    $\what{\tau}^+;$ $\what{\sigma}_{\tau^+}$, the empirical average of
    $\what{\sigma}_{\tau^+};$ SD. of $\what\tau^+$, the empirical standard
    deviation of $\what\tau^+$; and coverage, the empirical coverage
    probability of the 95\% confidence intervals
    $\what{\operatorname{CI}}_{\tau}.$ ($\mbox{ATE} =\tau = 1$ and
    $\Gamma_{\rm data} = \exp(1)$.)}
\begin{center}
\begin{tabular}{|c | c| c| c| c| c| c| c|}
\hline
     $n$   &  $\what{\tau}^-$ & SD. of $\what\tau^-$ & $\what{\sigma}_{\tau^-}$ &   $\what\tau^+$ &  SD. of $\what\tau^+$ & $\what{\sigma}_{\tau^+}$ &  Coverage \\
     \hline
500.0       &  1.008 &  0.085 &  0.081 &     1.424 &  0.082 &     0.077 &    0.952 \\
1000.0      &  1.000 &  0.059 &  0.057 &     1.404 &  0.058 &     0.053 &    0.978 \\
2000.0      &  0.998 &  0.042 &  0.040 &     1.395 &  0.040 &     0.038 &    0.966 \\
4000.0      &  0.995 &  0.029 &  0.028 &     1.387 &  0.027 &     0.027 &    0.980 \\
\hline
\end{tabular}
\end{center}
\label{tab:coverage}
\end{table}

In the second set of simulations, the dimension $d$ of the covariates, sample
size $n$, and marginal treatment probability $P(Z=1)$ match those from the
real observational study on fish consumption and blood mercury levels in the
next subsection ($d=8$, $n = 1100$, $P(Z=1)=0.21$), so that we can validate
our approach before interpreting the results on real data. We use the
nonparametric sieve implementation for estimating the nuisance parameters in
the real observational study, and so we use this implementation here.  As
estimation with sieves is challenging in this setting due to the eight
covariates and a nonlinear model, in Table~\ref{tab:realistic-coverage} we
observe that the variance estimates $\what{\sigma}^{\pm}$ underestimate the
standard deviation of $\what{\tau}^{\pm}$ by approximately $10\%$.  We also
evaluate the performance when the propensity score estimator is mis-specified,
as discussed in Section~\ref{sec:implementation}.

We compare our semiparametric methods to the $M$-estimator based matching
method \texttt{sensitivitymw}~\cite{rosenbaum2014weighted}. Note that our
simulation uses a constant treatment effect, as assumed by matching
methods. The confidence intervals for the matching approach is conditional on
the design (and assumes exact matched pairs), whereas our intervals are
unconditional. The confidence intervals for the ATE from the matching method
appear conservative, coming from having a lower design sensitivity and larger
standard errors (Table~\ref{tab:realistic-coverage}). The larger standard
errors could potentially be reduced using covariate adjustment in
matching~\citep{Rosenbaum02Cov}. 

\begin{table}[t]
\caption{Simulation results of the proposed method (parametric and nonparametric) and the existing matching method with eight observed covariates. $\what{\tau}^-$, the empirical average of $\what{\tau}^-;$  $\what{\sigma}_{\tau^-}$, the empirical average of $\what{\sigma}_{\tau^-};$ SD. of $\what\tau^-$, the empirical standard deviation of $\what\tau^-$;  $\what{\tau}^+$, the empirical average of $\what{\tau}^+;$  $\what{\sigma}_{\tau^+}$, the empirical average of $\what{\sigma}_{\tau^+};$ SD. of $\what\tau^+$, the empirical standard deviation of $\what\tau^+$; and Coverage, the empirical coverage probability of the 95\% confidence intervals $\what{\operatorname{CI}}_{\tau}.$ ($\mbox{ATE} =\tau = 1$ and
    $\Gamma_{\rm data} = \exp(1)$.)}
\begin{center}
\begin{tabular}{|l | c| c| c| c| c| c| c|}
\hline
     Approach   &  $\what{\tau}^-$ & $\what{\sigma}_{\tau^-}$ & SD. of $\what\tau^-$ &  $\what\tau^+$ & $\what{\sigma}_{\tau^+}$ &  SD. of $\what\tau^+$ & Coverage \\
     \hline
Nonparametric & 0.995 & 0.073 & 0.081 & 1.775 & 0.069 & 0.076 & 0.960 \\
Misspecified & 0.988 & 0.071 & 0.081 & 1.775 & 0.068 & 0.076 & 0.970 \\
Matching & 0.869 & - & 0.097 & 2.125 & - & 0.097 & 0.996 \\
\hline
\end{tabular}
\end{center}
\label{tab:realistic-coverage}
\vspace{-19pt}
\end{table}

In the third set of simulations, we include only a single covariate ($d=1$),
and evaluate the performance of the semiparametric method with the
\texttt{xgboost} implementation, and the matching method described above over
a range of sample sizes. One of the challenges with interpreting the above
simulations is that the results will include a mixture of errors---statistical
error from having finite observations, and population-level uncertainty on the
treatment effect.  With one covariate, the semiparametric and approximate matching methods should have a small statistical bias relative to their standard errors, so the average of the point estimates from
simulations with a large sample size should approximate the asymptotic sensitivity bounds well. This allows us to compare the asymptotic behavior of the semiparametric
method and matching methods, over a variety of values of $\Gamma$ used in
analysis (while holding $\Gamma_{\rm data}$ used in the data-generation fixed). Like previous settings, Table~\ref{tab:low-dim} shows that the
bounds from matching are more conservative than the semiparametric approach.


\begin{table}[t]
  \caption{Simulation results of the proposed method and matching with $1$
    observed covariate. For each method, $0.025$-quantile and average of the
    lower bound, followed by average and $0.975$-quantile of the upper bound,
    and the coverage of the confidence interval are reported. Comparing the
    average bounds for each method shows that the semiparametric method has a
    less conservative lower bound as $\Gamma$ varies, but is still below the
    true ATE when the appropriate $\Gamma$ is used, which is $1$ in this
    simulation; the coverage shows that it still covers the true ATE at the
    appropriate level. Varying the sample size shows that the statistical bias
    of both methods is already negligible with very small sample
    sizes. ($\mbox{ATE} =\tau = 1$ and $\Gamma_{\rm data} = \exp(1)$.)}
\footnotesize
\hspace{-2em}
\begin{tabular}{|c|ccccc|ccccc|}
\hline
& \multicolumn{5}{c|}{Semiparametric Method} & \multicolumn{5}{c|}{Matching Method}\\
\hline
& \makecell{Lower\\$0.025$-\\quantile} & Lower & Upper &  \makecell{Upper\\$0.975$-\\Quantile} &  \makecell{Cover-\\age} & \makecell{Lower\\$0.025$-\\quantile} & Lower & Upper &  \makecell{Upper\\$0.975$-\\Quantile} & \makecell{Cover-\\age}\\
\hline
 $\Gamma$ & \multicolumn{10}{c|}{Fixing $n=1000$}
\\
\hline
$1$  &       1.08 &  1.18 &      1.18 &       1.29 &     0.06 &        1.23 &  1.42 &        1.42 &       1.65 &           0.00 \\
$\exp(0.5)$  &       1.00 &  1.09 &      1.27 &       1.37 &     0.56 &        0.93 &  1.11 &        1.73 &       1.96 &           0.61 \\
$\exp(1)$  &       0.90 &  1.00 &      1.35 &       1.46 &     0.97 &        0.58 &  0.80 &        2.05 &       2.30 &           1.00 \\
$\exp(2)$  &       0.71 &  0.81 &      1.52 &       1.64 &     1.00 &       -0.13 &  0.17 &        2.69 &       3.01 &           1.00 \\
$\exp(3)$ &       0.51 &  0.63 &      1.69 &       1.82 &     1.00 &       -0.90 & -0.48 &        3.35 &       3.75 &           1.00 \\
$\exp(4)$ &       0.30 &  0.46 &      1.85 &       2.01 &     1.00 &       -1.67 & -1.16 &        4.02 &       4.49 &           1.00 \\
\hline
n & \multicolumn{10}{c|}{Fixing $\Gamma = \exp(1)$ as in simulation}
\\
\hline
100.0       &       0.65 &  1.00 &      1.37 &       1.69 &     0.97 &        0.11 &  0.82 &        2.05 &       2.83 &           0.99 \\
1000.0      &       0.90 &  1.00 &      1.35 &       1.46 &     0.97 &        0.58 &  0.80 &        2.05 &       2.30 &           1.00 \\
4000.0      &       0.94 &  1.00 &      1.35 &       1.41 &     0.98 &        0.68 &  0.81 &        2.05 &       2.19 &           1.00 \\
\hline
\end{tabular}
\label{tab:low-dim}
\end{table}

\subsection{Real observational data}
We apply our method to analyzing an observational study to infer
the effect of fish consumption on blood mercury levels and compare our
result to that of a prior analysis based on covariate
matching~\cite{zhao2017sensitivity}. The data consist of observations from
2,512 adults in the United States who participated in a single
cross-sectional wave of the National Health and Nutrition Examination Survey
(2013-2014). All participants answered a questionnaire regarding their
demographics and food consumption and had their blood mercury concentration measured (data available in the R package
CrossScreening).

High fish consumption is defined as individuals who reported $>12$ servings
of fish or shellfish in the previous month per their questionnaire, low fish
consumption as 0 or 1 servings of fish. The outcome of interest is $\log_2$
of total blood mercury concentration (ug/L). The primary objective is to
study if fish consumption causes higher mercury concentration. To match
prior analysis \cite{zhao2017sensitivity}, we excluded one individual with
missing education level and seven individuals with missing smoking status
from the analysis, and imputed missing income data for 175 individuals using
the median income. In addition, we created a supplementary binary covariate
to indicate whether the income data were missing. There are a total of 234
treated individuals (those with high fish consumption), 873 control
individuals (low fish consumption). The data include eight covariates
(gender, age, income, whether income is missing, race, education, ever
smoked, and number of cigarettes smoked last month).  Our approach uses the
same $\Gamma$-\cornfield{} model as the previous matched-pair analysis
in \cite{zhao2017sensitivity}, so results for our proposed method and
the analysis based on these 234 matched pairs are nearly comparable. However, the 
confidence intervals constructed for matching are conditional on the covariates
and choice of matched pairs. As
Table~\ref{tab:obs-study} shows (see also Fig.~\ref{fig:obs-study}), when
$\Gamma>\exp(1)$, our method achieves tighter confidence intervals around
the effect of fish consumption on blood mercury level: our confidence
intervals are nested within those based on the matching
method. For example, when $\Gamma=\exp(3)$ (representing a relatively large
selection bias), the 95\% confidence interval for the increase in average
$\log_2$-transformed blood mercury concentration caused by high fish
consumption is [0.47, 3.29] based on our new method and [-0.24, 4.48] based
on the matching method. While the former excludes zero, suggesting a
significant association in the presence of unknown confounding, the latter
includes the null association and is not statistically significant.
The confidence intervals for our method are always shorter except
when $\Gamma = 1$, ie. under unconfoundedness.

\begin{table}
\caption{Comparison to sensitivity results of \cite{zhao2017sensitivity} using the same data set. Because the same sensitivity model as the matched analysis was used, results can be compared directly. We demonstrate that the method can achieve tighter bounds on the average treatment effect both in point estimates and confidence intervals.}
\vspace{1em}
\hspace{-3.2em}\begin{tabular}{| c| c c c c c|c c c c c|}
\hline
& \multicolumn{5}{c|}{Semiparametric Method} & \multicolumn{5}{c|}{Matching Method} \\
\hline
$\Gamma$ & \makecell{Lower \\95\% CI} & Lower & Upper & \makecell{Upper \\ 95\% CI} & \makecell{Length\\ of CI} & \makecell{Lower \\95\% CI} & Lower & Upper & \makecell{Upper \\ 95\% CI} & \makecell{Length\\ of CI} \\
\hline
1                     & 1.51 & 1.74 & 1.74 & 1.97 & 0.46 & 1.9   & 2.08 & 2.08 & 2.25 & 0.35 \\
$\exp(0.5)$ & 1.31 & 1.53 & 2.03 & 2.26 & 0.95 & 1.57  & 1.75 & 2.41 & 2.59 & 1.02 \\
$\exp(1)$   & 1.07 & 1.27 & 2.27 & 2.47 & 1.4  & 1.25  & 1.45 & 2.74 & 2.94 & 1.89 \\
$\exp(2)$   & 0.74 & 0.91 & 2.77 & 2.89 & 2.15 & 0.58  & 0.87 & 3.36 & 3.65 & 3.07 \\
$\exp(3)$   & 0.47 & 0.6  & 3.19 & 3.29 & 2.82 & -0.23 & 0.28 & 3.97 & 4.48 & 4.71 \\
$\exp(4)$   & 0.18 & 0.29 & 3.55 & 3.63 & 3.45 & -     & -    & -    & -    & -    \\
\hline
\end{tabular}
\label{tab:obs-study}
\end{table}

\begin{figure}[ht]
  \centering
  \caption{Visual comparison to sensitivity results of matching method in
    \cite{zhao2017sensitivity} using the same data set. See numerical
    details in Table~\ref{tab:obs-study}. The filled areas represent the
    estimated bounds on the average treatment effect, whereas the dotted /
    dashed lines represent their confidence intervals. For
    values of $\Gamma$ larger than $\exp(0.5)$, our approach
    produces intervals with shorter length.}
  \label{fig:obs-study}
  \includegraphics[width=4in]{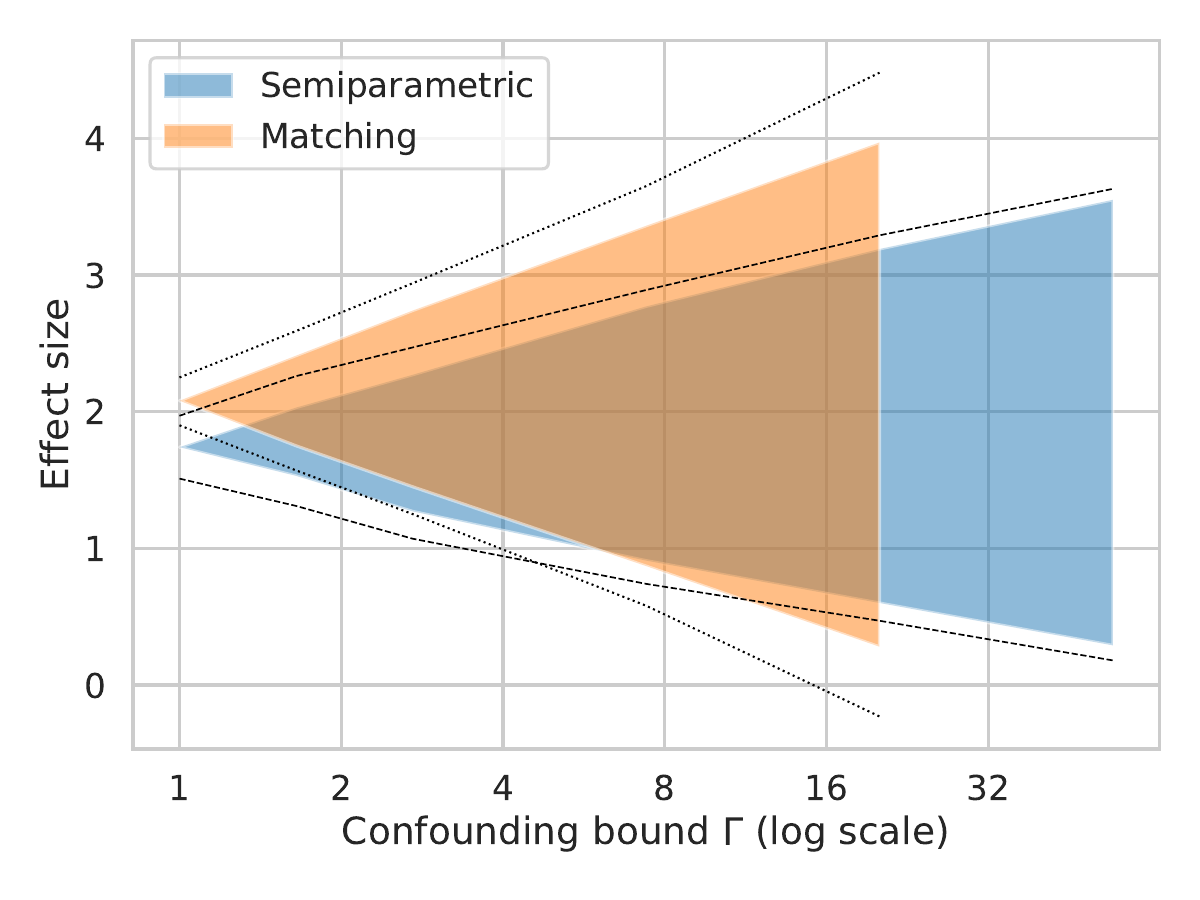}
\end{figure}

\section{Discussion}
\label{section:discussion}

The $\Gamma$-\cornfield{} model \eqref{eq:cornfield} relaxes the
unconfoundedness assumption \eqref{eq:indepsim} required for the
identification of causal treatment effects. We propose estimators
$\what{\tau}^{\pm}(\cdot)$ for upper and lower bounds on the CATE $\tau(x)$
and $\what{\tau}^{\pm}$ for the ATE $\tau$ under the $\Gamma$-\cornfield{}
condition~\eqref{eq:cornfield} and derive their asymptotic properties.  Our
loss minimization approach is practical and scalable, allowing the use of
flexible machine learning methods. Theoretically, we demonstrate the
statistical advantages of our approach, replicating the advantageous $o_{p}(n^{-p/(2p+d)})$ convergence of
series estimation procedures \cite{newey1997convergence} and root $n$ consistency of doubly robust semi-parametric
estimates~\cite{bang2005doublerobust,chernozhukov2018double} in the absence of unobserved
confounding~\eqref{eq:indepsim}. Our simulation studies and experimental
evidence from real observational data confirm these advantages exist in
practical finite sample regimes as well.

Our bounds demonstrate a few important phenomena for understanding the
robustness of causal inference with observational data. First, as we
note in
Section~\ref{sec:semiparametric}, the estimator $\what{\tau}^-$ reduces to
the AIPW estimator~\eqref{eq:aipw}
when $\Gamma=1$. Therefore, for any $\Gamma > 1$, the
confidence interval for $\tau$ estimated in \eqref{eq:ci4ate} includes the
AIPW estimate~\eqref{eq:aipw}, which serves as the center of
the interval bounding the ATE. Second,
the estimator $\what{\theta}_1(\cdot)$ minimizes a weighted squared error loss function \eqref{eqn:opt}, while the estimator
$\what{\mu}_{1,1}(\cdot)$ minimizes a unweighted mean squared error loss.
When the residual noise $Y - \mu_{1,1}(X)$ is small, the difference
between weighted and unweighted loss functions also tends to be small.
Therefore, the effect
of selection bias on the bias of the ATE $\tau$ or CATE $\tau(x)$
estimated under the no unobserved confounding assumption \eqref{eq:indepsim}
depends on the magnitude of these residuals; when these residuals are close to zero,
the risk of unobserved confounding is mitigated.

Our bounds on the ATE $\tau$ and CATE $\tau(x)$ depend on bounding the
conditional mean of the potential outcomes $\mu_1(x) = \E[Y(1) \mid X=x]$ and
$\mu_0(x) = \E[Y(0)\mid X=x]$. The proposed $\what{\tau}^-$ and
$\what{\tau}^-(x)$ employ a worst case re-weighting scheme (such as in
\eqref{eq:population-theta} and \eqref{eqn:opt}) to bound them
separately. Section~\ref{sec:hypothesis-test} establishes the optimality of
this approach under a specific symmetry condition on the distributions of the
potential outcomes. In general, our approach may not be optimal; an optimal
estimator may require worst case treatment assignments that depend on both
potential outcomes simultaneously, consistent with the independence assumption
\eqref{eq:indep} and $\Gamma$-\cornfield{} condition
\eqref{eq:cornfield}. Such joint consideration of $\mu_1(x)$ and $\mu_0(x)$
complicates the estimation procedure but is an important direction of future
research.

In practice, choosing an appropriate level of $\Gamma$ in the sensitivity
analysis is important. \citet[Chp.~6]{Rosenbaum02} discusses using known
relationships between a treatment and an auxiliary measured outcome to
detect the presence and magnitude of hidden bias. For example, suppose a drug is approved with an unbiased estimate of its effect on a primary outcome based on a randomized clinical trial,
and drug surveillance investigates the potential adverse events associated with the drug use in real world. The difference between the estimated treatment effect on the primary outcome based on observational data and that based on a randomized clinical trial can serve as an indication of the magnitude of $\Gamma$, the hidden bias in the observational data. It may then be appropriate to
perform a sensitivity analysis for adverse events with the same level of
$\Gamma$. However, in many settings, there is no such surrogate for
estimating $\Gamma$. In discussions with clinicians who often conduct
biomedical studies, we find it helpful to provide results for a number of
different values of $\Gamma$ 
to help
contextualize the strength of evidence, rather than present a single bound
with undue certainty.
While our
result is valid for each fixed $\Gamma$, providing uniform inference results over a set of $\Gamma$ would allow estimation of the smallest value of $\Gamma$ consistent with zero treatment effect in the data (a \emph{sensitivity value} analogous to the \textit{E value} for risk ratios from \citet{vanderweele2017sensitivity}).

\bibliographystyle{abbrvnat}
\setlength{\bibsep}{2pt}
\bibliography{bib_long,bib_short}

\newpage \appendix

\section{Proofs for bounds on the CATE}
\label{sec:proofs-no-cov}

\subsection{Proof of absolute continuity in Lemma~\ref{lem:bdd-lr-no-cov}}
\label{sec:proof-bounded-likelihood}


\begin{proof}
  Here, we only prove the absolute continuity result. The rest of the proof is in Section~\ref{sec:cate-sensitivity} after the statement of Lemma~\ref{lem:bdd-lr-no-cov}. Let $U \in \mc{U}$ be the unobserved confounder satisfying
  $Y(1) \indep Z \mid X, U$ and~\eqref{eq:cornfield}. Then for any set
  $A \subset \mc{U}$
  \begin{align}
    \frac{P(U \in A \mid Z = 0, X = x)}{P(U \in A \mid Z = 1, X = x)}
    & = \frac{P(Z = 0 \mid X = x, U \in A)}{P(Z = 1 \mid X = x, U \in A)}
    \cdot \frac{P(Z = 1 \mid X = x)}{P(Z = 0 \mid X = x)}
    \in [\Gamma^{-1}, \Gamma]
    \label{eqn:remember-cornfield}
  \end{align}
  by condition~\eqref{eq:cornfield} and the quasi-convexity of the ratio
  mapping
  $(a, b) \mapsto a/b$. Letting $q_z$ denote the density of $U$ (with
  respect to a base measure $\mu$) conditional on $Z = z$, we then have
  $q_0(u \mid x) / q_1(u \mid x) \in [\Gamma^{-1}, \Gamma]$, and for
  any measurable set $A \subset \R$
  \begin{align*}
    \frac{P(Y(1) \in A \mid Z = 0, X = x)}{P(Y(1) \in A \mid Z = 1, X = x)}
    & = \frac{\int P(Y(1) \in A \mid Z = 0, U = u, X = x) q_0(u \mid x) d\mu(u)}{
      \int P(Y(1) \in A \mid Z = 1, U = u, X = x) q_1(u \mid x) d\mu(u)} \\
    & \stackrel{(i)}{=}
    \frac{\int P(Y(1) \in A \mid U = u, X = x) q_0(u \mid x) d\mu(u)}{
      \int P(Y(1) \in A \mid U = u, X = x) q_1(u \mid x) d\mu(u)}
    \stackrel{(ii)}{\in} [\Gamma^{-1}, \Gamma]
  \end{align*}
  where equality~$(i)$ is a consequence of $Y(1) \indep Z \mid X, U$, and
  inequality~$(ii)$ follows again from the quasi-convexity of the ratio. This
  yields the absolute continuity claim.
\end{proof}

\subsection{Proof of Lemma~\ref{lemma:duality}}
\label{sec:proof-duality}

\begin{proof}
  As everything is conditional on $x$, we it without
  loss of generality, letting $\E_1[\cdot] = \E[\cdot \mid Z = 1]$ for
  shorthand.
  We first develop a simple duality argument. The set
  \begin{equation*}
    \mc{L}_\Gamma \defeq \{L : \mc{Y} \to \R_+, L ~ \mbox{measurable}, ~
    L(y) \le \Gamma L(\wt{y}) ~ \mbox{for~all~} y, \wt{y} \}
  \end{equation*}
  is convex, contains the constant function $L \equiv 1$ in its interior,
  and for $L \equiv 1$ we have $\E_1[L(Y(1))] = 1$. Thus,
  strong duality~\cite[Thm.~8.6.1 and Problem 8.7]{Luenberger69}
  implies
  \begin{equation}
    \inf_{L \in \mc{L}_\Gamma}
    \left\{\E_1[L(Y(1))] \mid \E_1[L(Y(1))] = 1 \right\}
    = \sup_{\mu \in \R}
    \inf_{L \in \mc{L}_\Gamma} \left\{\E_1[(Y(1) - \mu) L(Y(1))]
    + \mu \right\}
    \label{eq:lagrange-ver}.
  \end{equation}

  Now, we show that for each $\mu \in \R$, 
  \begin{equation}
    L^\ast(y) \propto \Gamma \ind{y - \mu \le 0} + \ind{y - \mu > 0}
    \label{eq:best-lr}
  \end{equation}
  attains the minimum value of $\inf_{L \in \mc{L}_\Gamma} \E_1[(Y(1) - \mu)
    L(Y(1))]$. That is, the minimizer takes on only the values $L^*(y) \in
  \{c, c\Gamma\}$ for some $c \ge 0$. The constraint
  $L \in \mc{L}_\Gamma$ guarantees that $L^*(y) \in [c, c \Gamma]$ for some
  $c \ge 0$.
  Assume that $c \le L(y) \le c \Gamma$,
  but $L(y) \not\in \{c, c \Gamma\}$. Then
  letting $L\opt(y) = c$ if $(y - \mu) > 0$ and $L\opt(y) = c \Gamma$
  if $(y - \mu) \le 0$, we have
  $(y - \mu) L\opt(y) \le (y - \mu) L(y)$, with strict inequality
  if $y \neq \mu$. Thus, any function $L \in \mc{L}_\Gamma$ can be modified
  to be of the form~\eqref{eq:best-lr} without increasing the objective
  $\E_1[(Y(1) - \mu) L(Y(1))]$.

  Substituting the minimizer~\eqref{eq:best-lr} into
  the right objective~\eqref{eq:lagrange-ver}, we recall
  that $\psi_t(y) = \hinge{y - t} - \Gamma \neghinge{y-t}$ to obtain
  \begin{equation*}
    \theta_1(x) =
    \sup_{\mu} \inf_{c \ge 0}
    \left\{\E_1\left[ c \psi_\mu(Y(1)) \mid X=x\right] + \mu\right\}.
  \end{equation*}
  This gives the final result~\eqref{eq:dual-constrained}, as
  \begin{equation*}
    \inf_{c \ge 0}
    \E_1\left[ c \psi_\mu(Y(1)) \mid X=x\right]
    = \begin{cases} -\infty & \mbox{if~} \E_1[\psi_\mu(Y(1)) \mid X = x]
      < 0 \\
      0 & \mbox{otherwise}.
    \end{cases}
  \end{equation*}
  Since
$\theta \mapsto \E[ \psi_\theta\left(Y(1)\right)\mid Z = 1, X]$ is a
decreasing function, $\theta_1(X)$ is the only zero crossing of the function
for almost every $X$.
\end{proof}
\vspace{-1em}

\subsection{Proof of Lemma~\ref{lemma:opt-is-good}}
\label{section:proof-of-opt-is-good}

\lemoptisgood*

\begin{proof}
  Let $\wb{\R} = \R \cup \{+\infty\}$.
  Normal integrand theory~\cite[Section 14.D]{RockafellarWe98} allows
  swapping integrals and infimum over measurable mappings. A map
  $f: \R \times \covspace \to \wb{\R}$ is a normal integrand if its
  epigraphical mapping $x \mapsto S_f(x) \defeq
  \epi f(\cdot; \covariate) = \{ (t,
  \alpha) \in \R \times \R: f(t; \covariate) \le \alpha\}$ is closed-valued
  and
  measurable, that is, for $\sigalg$ the Borel sigma-algebra on $\R$,
  $S_f^{-1}(O) \in \sigalg$ for all open $O \subset \R^2$. We have
  \begin{lemma}[{\citet[Theorem 14.60]{RockafellarWe98}}]
    \label{lemma:inf-int-interchange}
    If $f: \R \times \covspace \to \wb{\R}$ is a normal integrand, and
    $\int_{\covspace} f(\funcparam_1(\covariate); \covariate) ~
    dP(\covariate) < \infty$ for some measurable $\funcparam_1$, then
    \begin{equation*}
      \inf_{\funcparam} \left\{
      \int_{\covspace} f(\funcparam(\covariate); \covariate) ~d P(\covariate)
      \mid
      \funcparam: \covspace \to \R~\mbox{measurable}
      \right\}
      = \int_{\covspace} \inf_{\funcscalar \in \R} f(\funcscalar; \covariate) ~ dP(\covariate).
    \end{equation*}
    If this common value is not $-\infty$, a measurable function $
    \funcparam^*: \covspace \to \R$ attains the minimum of the left-hand side
    iff $ \funcparam^*(\covariate) \in \argmin_{\funcscalar \in \R}
    f(\funcscalar; \covariate)$ for $P$-almost every $\covariate
    \in \covspace$.
  \end{lemma}
  \noindent
  Let $f(t, \covariate) \defeq \half \treatedE[ \hinge{Y(1) - t}^2 + \Gamma
    \neghinge{Y(1) - t}^2 \mid \covariaterv = \covariate]$.  Since
  $(t, x) \mapsto f(t, x)$ is
  continuous by assumption, $f$ is a normal integrand~\cite[Example
    14.31]{RockafellarWe98}.  Rewrite the minimization
  problem~\eqref{eqn:opt} using the tower property
  \begin{equation*}
    \inf_{\funcparam} \left\{
    \treatedE\left[ \treatedE[
        \loss_\Gamma(\funcparam; (\covariaterv, Y(1))) \big| \covariaterv]
      \right] = \treatedE[f\left\{\funcparam(\covariaterv), \covariaterv\right\}]
    \mid
    \funcparam: \covspace \to \R~\mbox{measurable}
    \right\}.
  \end{equation*}
  Apply Lemma~\ref{lemma:inf-int-interchange} to obtain
  $\funcparam_1(\covariate)
  = \argmin_{\funcscalar \in \R} f(\funcscalar; \covariate)$. Since $t
  \mapsto f(t, \covariate)$ is convex, the first order condition
  $\frac{d}{dt} f(t; \covariate)=0$ shows that $\funcparam_1(x)$
  solves $\E_1[\psi_{\theta(x)}(Y(1)) \mid X=x]=0$. The uniqueness (up
  to measure-zero transformations) of $\theta_1$ is immediate by
  the strong convexity of $t \mapsto \loss_\Gamma(t, y)$.
\end{proof}

\section{Sieve estimation}
\label{sec:sieve-method}

\subsection{Convergence rates for $\what{\theta_1}$, the empirical
  minimizer~\eqref{eqn:opt-emp}}

In this section, we establish asymptotic convergence rates for minimizers
$\what{\theta}_1(\cdot)$ of \eqref{eqn:opt-emp}.  We consider two examples to
make this concrete.

\begin{example}[Polynomials]
  \label{example:polynomials}
  Let $\pol{\comp}$ be the space of $\comp$-th order polynomials on
  $[0, 1]$,
  \begin{equation*}
    \pol{\comp} \defeq
    \left\{ [0,1] \ni x \mapsto \sum_{k=0}^{\comp} a_k x^k
    : a_k \in \R \right\}.
  \end{equation*}
Define the sieve $\funcspace_n \defeq \left\{ x \mapsto \Pi_{k=1}^d f_k(x_k) \mid f_k \in \pol{J_n}, k = 1,\dots,d \right\},$ for $J_n \to \infty$.
\end{example}

\begin{example}[Splines]
  \label{example:splines}
  Let $0 = t_0 < \ldots < t_{\comp+1} = 1$ be knots that satisfy
  \begin{equation*}
    \frac{\max_{0 \le j \le \comp} (t_{j+1} - t_j)}
    {\min_{0  \le j \le \comp} (t_{j+1} - t_j)}
    \le c
  \end{equation*}
  for some $c > 0$. Then, the space of $r$-th order splines with $\comp$
  knots is
  \begin{equation*}
    \spl{r}{\comp} \defeq
    \left\{ x \mapsto
    \sum_{k=0}^{r-1} a_k x^k + \sum_{j=1}^\comp b_j \hinge{x-t_j}^{r-1}, x\in [0, 1]:
    a_k, b_k \in \R
    \right\}.
  \end{equation*}
  Define the sieves $\funcspace_n \defeq \left\{ x \mapsto
  f_1(x_1)f_2(x_2)\dots f_d(x) \mid f_k \in \spl{r}{J_n}, k = 1,\dots,d
  \right\}$ for some integer $r \ge \floor{\holdersmooth} + 1$ and $J_n \to
  \infty.$
\end{example}

We require (standard) regularity conditions. Let
$\holderball{\holdersmooth}{\holderradius}$ denote the H\"{o}lder class of
$\holdersmooth$-smooth functions, defined for
$\holdersmooth_1 = \ceil{\holdersmooth} - 1$ and
$\holdersmooth_2 = \holdersmooth - \holdersmooth_1$ by
\begin{equation*}
  \holderball{\holdersmooth}{\holderradius}
  \defeq \left\{ h \in C^{\holdersmooth_1}(\covspace):
  \sup_{\tiny \begin{array}{c}x\in \mc{X} \\ \sum_{l=1}^d \alpha_l<p_1 \end{array}}|D^\alpha h(x)| +
  \sup_{\tiny \begin{array}{c} x \neq x' \in \mc{X} \\ \sum_{l=1}^d \beta_l=p_1\end{array}}\frac{|D^\beta h(x) - D^\beta h(x')|}{\norm{x - x'}^{\holdersmooth_2}}
  \le \holderradius
  \right\},
\end{equation*}
where $C^{\holdersmooth_1}(\covspace)$ denotes the space of
$\holdersmooth_1$-times continuously differentiable functions on $\covspace,$ and
$D^{\alpha} = \frac{\partial^{\alpha}}
{\partial^{\alpha_1} \ldots \partial^{\alpha_\covdim}}$,
for any $\covdim$-tuple of nonnegative integers $\alpha = (\alpha_1, \ldots,
\alpha_\covdim)$.
We make a few concrete assumptions on smoothness and other properties of
parameters of interest.
\begin{assumption}
  \label{assumption:holder-smooth}
  Let
$\covspace = \covspace_1 \times \cdots \times \covspace_\covdim$ be the
Cartesian product of compact intervals
$\covspace_1, \ldots, \covspace_\covdim$, and assume
  $\popfunc \in \holderball{\holdersmooth}{\holderradius} \eqdef \funcspace$
  for some $\holderradius > 0$.
\end{assumption}

\newcommand{\condmoment}{\sigma^2_{\rm shift}}

\begin{assumption}
  \label{assumption:bdd-error}
  There exists $\condmoment < \infty$ such that for all $\covariate
  \in \covspace$, $\E[\{Y(1) - \popfunc(\covariaterv)\}^2 \mid Z=1, X = x] \le
  \condmoment$.
\end{assumption}

\begin{assumption}
  \label{assumption:lebesgue-equiv}
  $P_{\covariaterv | \treatmentrv = 1}$ has a density 
  $\treatedp(x)$ with respect to the Lebesgue measure and $0 < \inf_{\covariate \in
    \covspace} \treatedp(\covariate) \le \sup_{\covariate \in \covspace}
  \treatedp(\covariate) < \infty$.
\end{assumption}

Assumption~\ref{assumption:holder-smooth} assumes that $\popfunc(\cdot)$ is
in a $\holdersmooth$-smooth H\"{o}lder space.
Sufficient conditions for satisfying this assumption include when the conditional mean function $\mu_z(x) = \E[Y(z)  \mid X=x]$ is in a $\holdersmooth$-smooth H\"older space, and the residuals $Y(1) - \mu_1(x)$ are homoskedastic or $Y(1)$ is binary: in these cases, $\theta_1(x)$ is a simple affine transformation of $\mu_1(x)$, preserving its smoothness. Assumption~\ref{assumption:holder-smooth} allows for more general models where the residuals may be heteroskedastic but $\theta_1$ is still smooth. 
Assumption~\ref{assumption:bdd-error} is a standard condition
to ensure convergence of the empirical loss function by bounding the
second moment.  Finally, Assumption~\ref{assumption:lebesgue-equiv} asserts
that $P_{\covariaterv | \treatmentrv = 1}$ has upper- and lower-bounded
density, so that it is equivalent to the Lebesgue measure on $\mc{X}$. This assumption, along with the symmetric assumption on $P_{\covariaterv | \treatmentrv = 0}$ needed to estimate $\theta_0(\cdot)$, implies strong ignorability, as well as bounds on the marginal density of $P_{\covariaterv}$ under the Lebesgue measure. Assumption~\ref{assumption:lebesgue-equiv}
allows us to relate the $L^2(P)$ norm, $\|\cdot\|_{2,P},$ to the supremum norm, $\|\cdot\|_{\infty, P},$ of
$\what{\theta}_1 - \theta_1 \in \holderball{\holdersmooth}{\holderradius}$,
which is important for proving the convergence of sieve
estimators~\cite{Chen07}. Although outside the scope of this paper, adapting other nonparametric estimators such as the partitioning estimates described in \citet{GyorfiKoKrWa02} may admit good $\|\cdot\|_{2,P}$ convergence rates without this assumption.


The tradeoff between the random estimation error and approximation precision
of the sieve space $\funcspace_n$ (see Lemma~\ref{lemma:sieve} in the
Appendix) dictates the accuracy of $\what\theta_1(\cdot)$.  The following
theorem guarantees that finite dimensional linear sieves considered
yield standard non-parametric rates for estimating $\popfunc(\cdot)$ by
balancing different sources of error.
\begin{restatable}{thm}{thmsieve}
  \label{thm:sieve}
  For $\covspace = [0, 1]^d$, let $\funcspace_n$ be given by the finite
  dimensional linear sieves in Example~\ref{example:polynomials}
  or~\ref{example:splines} with $\comp_n \asymp n^{\frac{1}{2 \holdersmooth
      + \covdim}}$. Define $\epsilon_n = (\frac{\log
    n}{n})^\frac{\holdersmooth}{2\holdersmooth + \covdim}$.  Let
  Assumptions~\ref{assumption:holder-smooth}, \ref{assumption:bdd-error},
  and~\ref{assumption:lebesgue-equiv} hold, and let $\empfunc$ satisfy
  \begin{equation*}
    \empE\left[
      \loss_\Gamma\big(\empfunc (\covariaterv), Y(1)\big) \mid Z=1 \right]
    \le   \inf_{ \funcparam \in \funcspace_n}
    \empE\left[
      \loss_\Gamma\left(\funcparam(\covariaterv), Y(1)\right) \mid Z=1\right]
    + O_P(\epsilon_n^2).
  \end{equation*}
  Then
  $\norms{\empfunc - \popfunc}_{2, P_1} = O_{P}(\epsilon_n)$
  and $\norms{\empfunc - \popfunc}_{\infty,P_1} =
  O_{P}(\epsilon_n^\frac{2 \holdersmooth}{2 \holdersmooth + \covdim})$.
\end{restatable}
\noindent See the Supplementary Materials Section~\ref{section:proof-of-sieve} for proof.  The key
property of the function spaces $\funcspace_n$ in
Examples~\ref{example:polynomials} and~\ref{example:splines} is that $\inf_{\theta \in \funcspace_n} \linf{\theta -
  \popfunc} = O(\comp_n^{-\holdersmooth})$ (cf.~\cite[Sec.~5.3.1]{Timan63}
or \cite[Thm.~12.8]{Schumaker07}), which allows appropriate balance between
approximation and estimation error. Similar guarantees hold for wavelet
bases and other finite-dimensional sieves~\cite{Daubechies92, Chen07},
allowing generalization of Theorem~\ref{thm:sieve} beyond the
explicit examples provided.

\subsection{Convergence rates for $\what{\nu}_{1, k}$, the empirical
  minimizer~\eqref{eqn:opt-emp-prob}}
\label{sec:nu-sieve}

To show convergence of the empirical minimizer~\eqref{eqn:opt-emp-prob},
$\what{\nu}_{1, k}$, we need two assumptions.
\begin{assumption}
  \label{assumption:prob-smooth}
  There exist $\probsmooth, \probradius > 0$, and a set
  $S \subset \holderball{\holdersmooth}{\holderradius}$ with
  $\popfunc \in S$
  such that
  (a) $P(\what{\theta}_{1k}^{\nu_1} \in S \mid \treatmentrv=1) \to 1$ as $n \to \infty$,
  (b) for all $\funcparam \in S$,
    $x \mapsto P(Y(1) \ge \funcparam(x) \mid  \treatmentrv=1, X=x)$
    belongs to
    $ \probspace \defeq \holderball{\probsmooth}{\probradius} \cap \left\{
      \probfunc: \covspace \to [0, 1] \right\}$.
\end{assumption}
\begin{assumption}
  \label{assumption:lip-cdf}
  Let $S$ be as in Assumption~\ref{assumption:prob-smooth}.
  There is a constant $L_{\probfunc} <\infty$
  such that for $f, g \in S$,
  \begin{equation*}
    \int \left[P(Y(1) \ge f(X) \mid Z=1, \covariaterv=x)-
    P(Y(1) \ge g(X) \mid Z=1, \covariaterv=x)\right]^2 \dif{P_1}(x)
    \le L_{\probfunc}^2 \|f - g\|_{2,P_1}^2.
  \end{equation*}
\end{assumption}
\noindent Assumption \ref{assumption:prob-smooth} ensures that the map $x
\mapsto P(Y(1) \ge \funcparam(x) \mid \treatmentrv=1, X=x)$ is
sufficiently smooth for functions $\funcparam(\cdot)$ close to
$\popfunc(\cdot).$ In the current case, since $\|\what{\theta}_{1k}^{\nu_1} - \popfunc\|_{2,P_1}
= o_p(1)$, $S$ in Assumption \ref{assumption:prob-smooth} can be a
$\|\cdot\|_{2,P_1}$-neighborhood of $\popfunc \in
\holderball{\holdersmooth}{\holderradius}.$ This condition is necessary, as
$\hat{\nu}_{1,k}(x)$ solves an empirical version of the optimization
problem~\eqref{eqn:opt-prob} using the estimator
$\what{\theta}_{1k}^{\nu_1}(\cdot)$ instead of the true
$\funcparam_1(\cdot).$ Assumption \ref{assumption:lip-cdf} guarantees that
the map
$\funcparam \mapsto P\left(Y(1) \ge \funcparam(x) \mid \treatmentrv=1, X=x\right)$
is also sufficiently smooth. A simple sufficient condition for Assumption
\ref{assumption:lip-cdf} is that $Y(1) \mid \treatmentrv = 1, X=x$ has a
bounded density for almost every $x\in \mathcal{X}.$ If the density in certain
regions is not bounded, but we know a priori that $\theta(x) \not= y$ in these
regions, then we choose $S$ so that $\what{\theta}_{1k}^{\nu_1} \not= y$ in
these regions, as well. For instance, if $Y(1) \in \{0, 1\}$, then unless it
is deterministic, $\theta_1(x) \in (0,1)$, so
$\theta_1 \in S = \holderball{\probsmooth}{\probradius} \cap \{ f: \covspace
\to (0,1) \}$, and $P_{Y(1) \mid X=x, Z=1}$ has a density
$p_{Y(1) \mid X=x, Z=1}(\theta_1(x)) = 0$, which implies
Assumption~\ref{assumption:lip-cdf}.

Under these additional assumptions, the following proposition gives the
convergence rate of the proposed sieve estimator. See the proof
in the Supplementary Materials Section~\ref{section:proof-of-prob-sieve}.
\begin{restatable}{prop}{propprobsieve}
  \label{prop:prob-sieve}
  For $\covspace = [0, 1]^d$, let $\probspace_n$ be the finite dimensional
  linear sieves considered in Examples~\ref{example:polynomials} or
  \ref{example:splines}. Let
  $\epsilon_n = (\frac{\log n}{n})^\frac{q}{2q + d}$ and
  $\comp_n \asymp \epsilon_n^{-1/q}$. Let
  Assumptions~\ref{assumption:lebesgue-equiv},~\ref{assumption:prob-smooth},
  and~\ref{assumption:lip-cdf} hold.  Assume that
  $\norms{\what\theta_{1k}^{\nu_1} - \popfunc}_{2,P_1}=O_p(\epsilon_n)$, and
  let $\what{\nu}_{1,k}$
  satisfy 
   \begin{equation*}
    \mathbb{E}_{n,2}^{(k)} \left[  \lossnu\left(\what{\nu}_{1,k}(\covariaterv), \what{\theta}_{1k}^{\nu_1}(\covariaterv), Y(1)\right)\right]
    \le   \inf_{ \probfunc \in 1 + (\Gamma-1)\probspace_n}
    \mathbb{E}_{n,2}^{(k)}\left[ \lossnu\left(\probfunc(\covariaterv), \what{\theta}_{1k}^{\nu_1}(\covariaterv), Y(1)\right)\right]
    + O_p(\epsilon_n^2).
  \end{equation*}
  Then $\norms{\what{\nu}_{1,k} - \popprob}_{2, P} = O_p(\epsilon_n)$.
\end{restatable}


\section{A Practical Procedure for Estimating \lowercase{$\nu_{z}(\cdot)$}}
\label{sec:practical-nu}

The nested cross-fitting procedure proposed in Section~\ref{sec:nu-est} enjoys
strong theoretical properties, but it may be computationally expensive in
practice. For instance, when using $10$-fold cross-fitting, this requires
fitting the nonparametric sieve estimator $100$ times on different
subsamples. To reduce the computational complexity, we use an iterative
cross-fitting procedure. We describe our algorithm for the treated units
$z=1$, as the case for control units $z=0$ is symmetric.
\begin{enumerate}
\item Select hyperparameters for estimating $\theta_{1}(\cdot)$
  using cross-validation on the weighted squared
  loss~\eqref{eq:gamma-loss}.
\item Fit models $\what\theta_{1,k}(\cdot),~k=1,\dots,K$ using cross-fitting
  with the selected hyperparamters.
\item Compute the binary targets $V_i = 1\{Y_i \le \what\theta_{1,k}(X_i)\}$
  for each observation $i$ based on the corresponding $k$ such that
  $i \in \mathcal{I}_k$.
\item Using the targets $V_i$, choose hyperparameters for estimating
  $\nu_1(\cdot)$ using cross-validation; use the loss described in
  Section~\ref{sec:nu-est} with observations $(V_i, X_i)$.
\item With the selected hyperparameters, fit final estimators
  $\what\nu_{1,k}(\cdot)$ using the observations $(V_i, X_i)$ for
  $i \in \mathcal{I}_{-k}$.
\item Using these nuisance parameter estimates alongside estimates of the
  propensity score, $\what{e}_{1,k}(\cdot)$, calculate the cross-fitted
  semiparametric estimate $\what\tau^{-}$
\end{enumerate}
Due to the cross-fitting construction, the $\what\theta_{1,k}(\cdot)$ used to
construct the targets $V_i$ are independent of the $i$-th observation,
capturing the key property of the nested cross-fitting needed for estimation
of $\nu_1(\cdot)$ as described in Section~\ref{sec:nu-est}. Additionally,
$\what{\theta}_{1,k}$ is independent of the observations $(X_i, Y_i, Z_i)$ for
$i \in I_k$ when plugged in to the cross-fit estimate $\what\tau^{-}$.

However, the iterative construction does not guarantee the independence of
$\what\nu_1(\cdot)$ and the observations $(X_i, Y_i, Z_i)$ for $i \in I_k$.
To see this, let $k$ and $k'$ be two distinct fold indices between $1$ and
$K$, and let $i$ be an observation in $\mathcal{I}_k$ and $i'$ be an
observation in $\mathcal{I}_{k'}$. Observation $i$ is used to fit
$\what\theta_{1, k'}(\cdot)$, which is then used to compute
$V_{i'}$. Therefore, the $i$-th observation is not independent of $V_{i'}$,
which is used to estimate $\what\nu_{1,k}(\cdot)$. In summary, the $i$-th
observation will generally not be independent of $\what\nu_{1,k}(\cdot)$. In
all of our numerical experiments, this dependence does not have a
noticeable effect on the distribution of $\what{\tau}^{-}$.



\section{Proofs for Sieve Estimation}
\subsection{Proof of Theorem~\ref{thm:sieve}}
\label{section:proof-of-sieve}

We require a few notions of complexity to give this proof.
Let $\mc{V}$ be a vector space with (semi)norm
$\norm{\cdot}$ on $\mc{V}$, and let $V \subset \mc{V}$. A collection $v_1,
\ldots, v_\covnum \subset V$ is an \emph{$\epsilon$-cover} of $V$ if for
each $v \in V$, there exists $v_i$ such that $\norm{v - v_i} \le
\epsilon$. The \emph{covering number} of $V$ with respect to $\norm{\cdot}$
is then $\covnum(V, \epsilon, \norm{\cdot}) \defeq \inf\left\{\covnum \in \N
: ~ \mbox{there~is~an~} \epsilon \mbox{-cover~of~} V ~
\mbox{with~respect~to~} \norm{\cdot} \right\}$. For some fixed
$b > 0$, define the sequence
\begin{equation}
  \label{eqn:covering-mod}
  \covmod_n \defeq
  \inf\left\{ \delta \in (0, 1):
    \frac{1}{\sqrt{n} \delta^2}
    \int_{b\delta^2}^{\delta} \sqrt{\log \covnum\left( 
        \epsilon^{1 + \covdim / 2\holdersmooth }, \funcspace_n, \norm{\cdot}_{2, \treatedP}
      \right)} d\epsilon
    \le 1 \right\}.
\end{equation}
The following convergence result is a consequence of general results on
sieve estimators~\cite{ChenSh98, Huang98, Chen07} adapted for the
optimization problem~\eqref{eqn:opt}.
\begin{restatable}{lemma}{lemsieve}
  \label{lemma:sieve}
  Let
  Assumptions~\ref{assumption:holder-smooth}, \ref{assumption:bdd-error},
  and~\ref{assumption:lebesgue-equiv}
  hold, and let $\empfunc$ minimize the empirical risk~\eqref{eqn:opt-emp}
  to accuracy
  \begin{equation*}
    \E_{1,n}\left[  \loss_\Gamma\left\{\empfunc (\covariaterv), Y(1)\right\}\right]
    \le   \inf_{ \funcparam \in \funcspace_n}
    \E_{1,n}\left[  \loss_\Gamma\left\{\funcparam (\covariaterv), Y(1)\right\} \right]
    + O_p\left(\epsilon_n^2\right)
  \end{equation*}
  where
  $\epsilon_n \defeq \max\{ \covmod_n, \inf_{\funcparam \in \funcspace_n}
    \norm{\popfunc - \funcparam}_{2, \treatedP}\}$.  Then,
  $\norms{\empfunc - \popfunc}_{2, \treatedP} = O_p( \epsilon_n)$.
\end{restatable}
\begin{proof}
  To prove the result, it is sufficient to verifying the assumptions of the
  following general result for sieve estimation due to~\citet{ChenSh98}
  (see also~\cite{Chen07, Huang98}).
\begin{lemma}[{\citet[Theorem 3.2]{Chen07}}]
  \label{lemma:sieve-chen}
  Let $\popfunc \in \holderball{\holdersmooth}{\holderradius}$ for some
  $\holdersmooth > 0, \holderradius < \infty$, and for $\funcparam$ in some
  neighborhood of $\popfunc$ assume that
  \begin{equation*}
    \treatedE[ \loss_\Gamma(\funcparam(\covariaterv), Y(1))] - \treatedE[
  \loss_\Gamma(\popfunc(\covariaterv), Y(1))] \asymp \norm{\funcparam -
    \popfunc}_{2, \treatedP}^2.
  \end{equation*}
  For $\delta$ small enough, assume there exists
  a function $M : \mc{X} \times \R \to \R_+$ such that
  \begin{align}
    & \sup_{\funcparam \in \funcspace_n:
    \norm{\funcparam - \popfunc}_{2, \treatedP} \le \delta}
    \var_{\treatedP} \left(
      \loss_\Gamma(\funcparam(\covariaterv), Y(1))
      - \loss_\Gamma(\popfunc(\covariaterv), Y(1))
    \right) \lesssim \delta^2 \label{eqn:var-modulus} \\
    & \sup_{\funcparam \in \funcspace_n:
    \norm{\funcparam - \popfunc}_{2, \treatedP} \le \delta}
      \left|
      \loss_\Gamma(\funcparam(\covariaterv), Y(1))
      - \loss_\Gamma(\popfunc(\covariaterv), Y(1))
      \right|
      \le \delta^s M(\covariaterv, Y(1))
      \label{eqn:loss-modulus}
  \end{align}
  for some $s \in (0, 2)$ where
  $\treatedE[M(\covariaterv, Y(1))^2] < \infty$. Then
  $\norms{\empfunc - \popfunc}_{2, \treatedP} = O_p( \epsilon_n)$.
\end{lemma}

To verify these assumptions, first we check that
\begin{equation}
  \treatedE\left[ \loss_\Gamma \left\{\funcparam (\covariaterv), Y(1)\right\}\right] - \treatedE \left[
  \loss_\Gamma\left\{\popfunc(\covariaterv), Y(1)\right\} \right] \asymp \norm{\funcparam -
    \popfunc}_{2, \treatedP}^2.
  \label{eqn:loss-differences-like-ltwo}
\end{equation}
Indeed, $\loss_\Gamma(\cdot, \outcome)$ is $1$-strongly
convex and has $\Gamma$-Lipschitz derivative, so for any $t, t' \in \R$,
\begin{equation}
  \half (t - t')^2
  + \psi_{t'}(\outcome) (t - t')
  \le \loss_\Gamma(\funcscalar, \outcome) -   \loss_\Gamma(\funcscalar', \outcome)
  \le \psi_{\funcscalar'}(\outcome) (\funcscalar - \funcscalar')
  + \frac{\Gamma}{2} (\funcscalar - \funcscalar')^2,
  \label{eqn:loss-upper-lower}
\end{equation}
where we have used that $\psi_{\funcscalar}(\outcome) =\hinge{\outcome -
  \funcscalar} - \Gamma \neghinge{\outcome-\funcscalar} =
\frac{\partial}{\partial \funcscalar}\loss_\Gamma(\funcscalar,
y)$. Recalling that $\treatedE[ \psi_{\popfunc(\covariaterv)}(Y(1)) \mid
  \covariaterv] = 0$ almost surely, taking expectations
yields~\eqref{eqn:loss-differences-like-ltwo} as
\begin{equation*}
  \half \norm{\theta - \theta_1}_{2,\treatedP}^2
  \le \treatedE\left[ \loss_\Gamma\left\{\funcparam (\covariaterv), Y(1)\right\}\right] - \treatedE \left[
  \loss_\Gamma\left\{\popfunc(\covariaterv), Y(1)\right\} \right]
  \le \frac{\Gamma}{2} \ltwotp{\funcparam - \popfunc}^2.
\end{equation*}

Next, we verify~\eqref{eqn:var-modulus} and \eqref{eqn:loss-modulus}.
By substituting in inequality~\eqref{eqn:loss-upper-lower}, we have
\begin{equation}
  \label{eqn:modulus-bound}
  |\loss_\Gamma(\funcparam(\covariate), \outcome)
  - \loss_\Gamma(\popfunc(\covariate), \outcome) |
  \le \Gamma |\outcome - \popfunc(\covariate)| |\funcparam(\covariate) - \popfunc(\covariate)|
  + \Gamma |\funcparam(\covariate) - \popfunc(\covariate)|^2.
\end{equation}
The following lemma~\cite{ChenSh98, Gabushin67} connects the
$L^2(\lambda)$-norm of
$\funcparam \in \holderball{\holdersmooth}{\holderradius}$ to its supremum
norm (where $\lambda$ denotes the Lebesgue measure).
\begin{lemma}[{\citet[Lemma 2]{ChenSh98}}]
  \label{lemma:two-sup-norm}
  For $\funcparam \in \holderball{\holdersmooth}{\holderradius}$, we have
  $\linf{\funcparam} \le 2 c^{1-\frac{2\holdersmooth}{2\holdersmooth + \covdim}} \norm{\funcparam}_{2,
    \lambda}^{\frac{2\holdersmooth}{2\holdersmooth + \covdim}}$.
\end{lemma}
\noindent Note that
$\norm{\cdot}_{2, \lambda} \asymp \norm{\cdot}_{2, \treatedP}$ by
Assumption~\ref{assumption:lebesgue-equiv}, and so
$\linf{\funcparam} \lesssim \norm{\funcparam}_{2, \treatedP}^{\frac{2\holdersmooth}{2\holdersmooth + \covdim}}$.
Taking squares on both sides in the inequality~\eqref{eqn:modulus-bound} and
using convexity of $t \mapsto t^{2}$ gives
\begin{equation*}
  |\loss_\Gamma(\funcparam(\covariate), \outcome))
  - \loss_\Gamma(\popfunc(\covariate), \outcome)) |^2
  \le 2\Gamma^2 |\outcome - \popfunc(\covariate)|^2 |\funcparam(\covariate) -
  \popfunc(\covariate)|^2
  + 2\Gamma^2 |\funcparam(\covariate) - \popfunc(\covariate)|^4.
\end{equation*}
Recalling from Assumption~\ref{assumption:bdd-error} that $\treatedE[(Y(1) -
  \popfunc(\covariaterv))^2 \mid \covariaterv] \le \condmoment$ for some
$\condmoment < \infty$, Lemma~\ref{lemma:two-sup-norm} implies that
\begin{align*}
  \sup_{\funcparam \in \funcspace_n:
    \norm{\funcparam - \popfunc}_{2, \treatedP} \le \delta}
    \var_{\treatedP} \left(
      \loss_\Gamma(\funcparam(\covariaterv), Y(1))
      - \loss_\Gamma(\popfunc(\covariaterv), Y(1))
  \right)  \lesssim \Gamma^2 M \delta^2 + 
  \Gamma^2 \delta^{2+ \frac{4\holdersmooth}{2\holdersmooth + \covdim}}
  \lesssim \delta^2
\end{align*}
whenever $\delta \in (0, 1)$. This verifies the condition~\eqref{eqn:var-modulus}.
Similarly, for $\delta$ small
\begin{equation*}
  \sup_{\funcparam \in \funcspace_n:
    \norm{\funcparam - \popfunc}_{2, \treatedP} \le \delta}
  \left|
    \loss_\Gamma(\funcparam(\covariaterv), Y(1))
    - \loss_\Gamma(\popfunc(\covariaterv), Y(1))
  \right|
  \lesssim \Gamma  \delta^{\frac{2\holdersmooth}{2\holdersmooth + \covdim}} c^{1-\frac{2\holdersmooth}{2\holdersmooth+\covdim}} 
  \left(|Y(1) - \popfunc(X)| + c^{1-\frac{2\holdersmooth}{2\holdersmooth+\covdim}}\right).
\end{equation*}
Noting $\treatedE (Y(1) - \popfunc(X))^2 <\infty$
verifies the condition~\eqref{eqn:loss-modulus} with
$s = 2\holdersmooth / (2\holdersmooth + \covdim).$
\end{proof}

  It now suffices to bound $\covmod_n$ and the approximation error
  $\inf_{\funcparam \in \funcspace_n} \norm{\popfunc - \funcparam}_{2,
    \treatedP}$ in Lemma ~\ref{lemma:sieve}.  First, note
  from~\citet{ChenSh98} and \citet{vandeGeer00} that
\begin{equation*}
  \log\covnum\left(
    \epsilon, \funcspace_n, \ltwotp{\cdot}
  \right)
  \lesssim \mbox{dim}(\funcspace_n) \log \frac{1}{\epsilon},
\end{equation*}
where $\mbox{dim}(\funcspace_n) = \comp_n^d$.
Then
\begin{equation*}
  \frac{1}{\sqrt{n} \delta^2}
  \int_{b\delta^2}^{\delta} \sqrt{\log \covnum\left( 
      \epsilon^{1 + \covdim / 2\holdersmooth }, \funcspace_n, \norm{\cdot}_{2, \treatedP}
    \right)} d\epsilon
  \lesssim \frac{1}{\delta} \sqrt{\frac{\mbox{dim}(\funcspace_n)}{n}\log \frac{1}{\delta}}, 
\end{equation*}
which implies that
$$\covmod_n \asymp \sqrt{\frac{\mbox{dim}(\funcspace_n) \log n}{n}} =
\sqrt{\frac{\comp_n^d \log n}{n}}.$$

For $\funcspace_n$ defined as in Examples~\ref{example:polynomials} or
\ref{example:splines} with $\comp=\comp_n$, standard function approximation
results yield $\inf_{\funcparam \in \funcspace_n} \linf{\funcparam -
  \popfunc} = O(\comp_n^{-\holdersmooth})$. (See
\citet[Section 5.3.1]{Timan63} and \citet[Theorem
  12.8]{Schumaker07}, respectively.)
Therefore, for any of these choices of approximating functions,
\begin{equation*}
  \inf_{\funcparam \in \funcspace_n} \norm{\popfunc - \funcparam}_{2,
  \treatedP} = O(\comp_n^{-p}).
\end{equation*}
Set $\comp_n \asymp n^{\frac{1}{2\holdersmooth + \covdim}} (\log
n)^{-\frac{1}{2\holdersmooth + \covdim}}$ in Lemma~\ref{lemma:sieve}, so
that $\|\what\theta_1-\theta_1\|_{2,P_1}=O_p((\frac{\log n}{n}
)^{\frac{\holdersmooth}{2\holdersmooth + \covdim}})$.

Finally, Lemma~\ref{lemma:two-sup-norm} gives the comparison between
$\norm{\cdot}_{2, \treatedP}$ and $\norm{\cdot}_{\infty, \treatedP}$.

\subsection{Proof of Proposition~\ref{prop:prob-sieve}}
\label{section:proof-of-prob-sieve}

To simplify notation, we drop the dependence on cross-fit folds, and write as
if $\what{\nu}_{1,k}$ is estimated using a sample of size $n$, with empirical
expecatation $\E_n$, and $\what{\theta}_{1k}^{\nu_1}$ is estimated using an
independent sample.

For
some fixed $b > 0$, let
\begin{equation}
  \label{eqn:prob-covering-mod}
  \bar{\covmod}_n \defeq
  \inf\left\{ \delta \in (0, 1):
    \frac{1}{\sqrt{n} \delta^2}
    \int_{b\delta^2}^{\delta} \sqrt{\log \covnum\left( 
        \epsilon^{1 + \covdim / 2\holdersmooth }, \probspace_n,
        \norm{\cdot}_{2, \treatedP}
      \right)} d\epsilon
    \le 1 \right\}.
\end{equation}
The following lemma shows that $\bar{\covmod}_n$ quantifies the trade-off
between estimation error, approximation error, and
$\norms{\what{\theta}_{1k}^{\nu_1} - \popfunc}_{2,P_1}$ when approximating
$\popprob.$ Proposition~\ref{prop:prob-sieve} follows directly from
the lemma.
\begin{restatable}{lemma}{lemprobsieve}
  \label{lemma:prob-sieve}
  Let Assumptions~\ref{assumption:lebesgue-equiv},
  \ref{assumption:prob-smooth}, and~\ref{assumption:lip-cdf} hold, and let
  $\what{\theta}_{1k}^{\nu_1}$ be a consistent estimator of $\theta_1$ based on
  an independent external data. Let $\what{\nu}_{1,k}$ minimize the
  empirical loss to accuracy
  \begin{equation*}
    \E_{n}\left[\lossnu\left(\what{\nu}_{1,k}(\covariaterv), \what{\theta}_{1k}^{\nu_1}(\covariaterv), Y(1)\right)\right]
    \le   \inf_{ \probfunc \in 1+(\Gamma-1)\probspace_n}
    \E_{n}\left[ \lossnu\left(\probfunc(\covariaterv), \what{\theta}_{1k}^{\nu_1}(\covariaterv), Y(1)\right)\right]
    + O_p\left(\epsilon_n^2\right)
  \end{equation*}
  where
  $\epsilon_n \defeq \max\{ \bar{\covmod}_n, \inf_{\probfunc \in
      1+(\Gamma-1)\probspace_n} \norm{\popprob - \probfunc}_{2, \treatedP}
  \}$.  If $n \epsilon_n^2 \to \infty$, then
  $$\norm{\what{\nu}_{1,k} - \popprob}_{2, \treatedP} = O_p \left( \epsilon_n +
    \ltwotp{\what{\theta}_{1k}^{\nu_1} - \popfunc} \right).$$
\end{restatable}
\begin{proof}
  \newcommand{\nugiven}[1]{\nu_1(\cdot \mid {#1})}
  \newcommand{\nuxgiven}[1]{\nu_1(x \mid {#1})}
  \newcommand{\nuXgiven}[1]{\nu_1(X \mid {#1})}
  For $\theta : \mc{X} \to \R$,
  let $\nugiven{\theta}$ be the solution to the optimization
  problem~\eqref{eqn:opt-prob} using $\theta$ in place
  of $\popfunc$, i.e.
  \begin{equation*}
    \nuxgiven{\theta}
    \defeq 1 + (\Gamma-1)
    \treatedP(Y(1) \ge \theta(\covariaterv) \mid
    \covariaterv=x, \theta).
  \end{equation*}
  By the triangle inequality,
  \begin{equation}
    \label{eqn:triangle-prob}
    \ltwotp{\what{\nu}_{1,k} - \popprob}
    \le \ltwotp{\what{\nu}_{1,k} - \nugiven{\what{\theta}_{1k}^{\nu_1}}}
    + \ltwotp{\nugiven{\what{\theta}_{1k}^{\nu_1}} - \popprob}.
  \end{equation}

  We bound the second term in \eqref{eqn:triangle-prob}
  by using Assumption~\ref{assumption:lip-cdf} to obtain
  \begin{align*}
    \ltwotp{\nugiven{\what{\theta}_1} - \popprob}^2
    & = (\Gamma-1)^2 \E_1\left[\left(
      \treatedP(Y(1) \ge \what{\theta}_1 (\covariaterv)
      \mid \covariaterv)
      -
      \treatedP(Y(1) \ge \popfunc(\covariaterv) \mid \covariaterv)
      \right)^2\right] \\
    & \le (\Gamma-1)^2 L_{\probfunc}^2
    \ltwotp{\what{\theta}_{1k}^{\nu_1} - \popfunc}^2.
  \end{align*}

  It remains to show that
  $\norms{\what{\nu}_{1,k} - \nugiven{\what{\theta}_{1k}^{\nu_1}}}_{2,P_1}
  = O_p(\epsilon_n)$. To this end, we use the
  fact that $\nugiven{\what{\theta}_{1k}^{\nu_1}}$ uniquely minimizes
  $\treatedE[  \lossnu(\probfunc(\covariaterv),
    \what{\theta}_{1k}^{\nu_1}(\covariaterv), Y(1))]$
  and apply a general result for sieve estimators.
  In this case, a variant of
  Lemma~\ref{lemma:sieve} with high-probability guarantees
  establishes appropriate convergence rates.
  \begin{lemma}[{\citet[Corollary 1]{ChenSh98}}]
    \label{lemma:sieve-concentration}
    Let $\theta : \mc{X} \to \R$ be such that $\nugiven{\theta} \in 1 +
    (\Gamma-1)\holderball{\probsmooth}{\probradius}$ for some $\probsmooth,
    \probradius > 0$. Assume that
    for any $\nu : \mc{X} \to \R$,
    \begin{equation}
      \label{eq:approximate}
      \treatedE[  \lossnu(\nu(\covariaterv), \theta(\covariaterv), Y(1))]
      - \treatedE[  \lossnu(\nuXgiven{\theta}, \theta(\covariaterv),
        Y(1))]
      \asymp \norm{\nu -  \nugiven{\theta}}_{2, \treatedP}^2.
    \end{equation}
    Let $\mc{V}_n(\delta) = \{\nu \in 1 + (\Gamma - 1) \probspace_n:
    \norms{\probfunc - \nugiven{\theta}}_{2,\treatedP} \le \delta\}$ for
    shorthand, and assume that for small enough $\delta > 0$ and some $s \in (0,
    2)$,
    \begin{align}
      & \sup_{\probfunc \in \mc{V}_n(\delta)}
      \var\left(
      \lossnu\left(\probfunc(\covariaterv), \theta(\covariaterv), Y(1)\right)
      - \lossnu\left(\nuXgiven{\theta},
      \theta(\covariaterv), Y(1)\right) \mid Z=1
      \right) \le C \delta^2 \label{eqn:var-modulus-prob} \\
      & \sup_{\probfunc \in \mc{V}_n(\delta)}
      \left|
      \lossnu\left(\probfunc(x), \theta(x), y\right)
      - \lossnu\left(\nuxgiven{\theta}, \theta(x), y\right)
      \right|
      \le C\delta^s
      \label{eqn:loss-modulus-prob}
    \end{align}
    for all $x, y$,
    where $C < \infty$ is independent of $\theta$.  If the loss is uniformly
    bounded over $\Pi$ so that $\sup_{\nu \in 1+ (\Gamma-1)\Pi, \theta \in
      \R, y \in \R} \lossnu(\nu(x), \theta, y) < \infty$, then there exist
    constants $0 < c, C < \infty$ (independent of
    $\theta$) such that for any $t > 0$,
    if $\what{\nu}_{1,k}(\cdot \mid \theta)$ is the minimizer to the problem
    \begin{equation*}
    \minimize_{ \probfunc(\cdot) \in 1 + (\Gamma -1)\probspace_{n}}
    ~ \E_{n}\left[\lossnu
      \big(\probfunc(\covariaterv), \theta(\covariaterv),
      Y \big) \mid Z = 1 \right],
  \end{equation*}
  then we have
    \begin{equation*}
      \treatedP\left(
      \ltwotp{\what{\nu}_{1,k}(\cdot \mid \theta)
        - \nugiven{\theta}} \ge t \epsilon_n
      \right)
      \le C \exp\left( -c n \epsilon_n^2 t^2 \right).
    \end{equation*}
  \end{lemma}

  On the event $ \event_n \defeq \{ \what{\theta}_{1k}^{\nu_1} \in S \}$,
  Assumption~\ref{assumption:prob-smooth} implies
  $\nugiven{\what{\theta}_{1k}^{\nu_1}} \in 1 + (\Gamma-1)\probspace$.  Since
  the event $\event_n$ is independent of the samples used in the sieve
  procedure~\eqref{eqn:opt-emp-prob} for computing $\what{\nu}_{1,k}$ and
  $\P_1(\event_n) \to 1$, applying Lemma~\ref{lemma:sieve-concentration}
  conditioned on this event gives the desired result. We now verify each of
  the hypotheses of Lemma~\ref{lemma:sieve-concentration} for any $\theta \in S$.

  Recall the definition $\lossnu(\nu, \theta, y)
  = \half (1 + (\Gamma - 1) \indic{y \ge \theta} - \nu)^2$. We
  first demonstrate boundedness:
  we have
  \begin{equation*}
    \sup_{\nu \in 1+ (\Gamma-1)\probspace, x, \theta \in \R, y \in \R}
    \lossnu(\nu(x), \theta, y) \le \half \Gamma^2,
  \end{equation*}
  as $\probspace \subset \{\mc{X} \to [0, 1]\}$.
  Since
  $\treatedE[\nuXgiven{\theta} - 1 - (\Gamma-1)\indic{Y(1) \ge
      \theta(\covariaterv)} \mid \covariaterv] = 0$ almost surely,
  for any $\nu, \theta : \mc{X} \to \R$ we have
  \begin{align*}
    \treatedE[  \lossnu(\nu(\covariaterv), \theta(\covariaterv), Y(1))]
    - \treatedE[ \lossnu(\nuXgiven{\theta}, \theta(\covariaterv), Y(1))]
    = \half \treatedE [(\nu(\covariaterv) - \nuXgiven{\theta})^2],
  \end{align*}
  and condition~\eqref{eq:approximate} holds. To verify
  conditions~\eqref{eqn:var-modulus-prob} and~\eqref{eqn:loss-modulus-prob},
  we note that
  \begin{align*}
    \lefteqn{
      \lossnu(\probfunc(x), \theta(x), y)
      -   \lossnu\left(\nuxgiven{\theta}, \theta(x), y\right)}
    \\
    & = \half \left( \probfunc(\covariate)^2 - \nuxgiven{\theta}^2\right)
    - \left( 1 + (\Gamma-1) \indic{\outcome \ge (\covariate)}\right)
    \left(
    \probfunc(\covariate) - \nuxgiven{\theta}    \right).
  \end{align*}
  By squaring both sides, taking expectations, and using
  Lemma~\ref{lemma:two-sup-norm} to see that $\linf{\nu},
  \linf{\nugiven{\theta}} \le C$ for elements $\nu \in \probspace_n$ (and
  applying Assumption~\ref{assumption:lebesgue-equiv}), we obtain
  \begin{equation*}
    \var_{\treatedP} \left(
    \lossnu(\probfunc(\covariaterv), \theta(\covariaterv), Y(1))
    -   \lossnu(\nuXgiven{\theta}, \theta(\covariaterv), Y(1))
    \right)
    \le C \ltwotp{\probfunc - \nugiven{\theta}}^2
    \le C \delta^2
  \end{equation*}
  for any $\probfunc \in \probspace_n$ such that $\norm{\probfunc -
    \nugiven{\theta}}_{2, \treatedP} \le
  \delta$. Condition~\eqref{eqn:var-modulus-prob} holds.  For
  condition~\eqref{eqn:loss-modulus-prob}, we similarly note that
  \begin{align*}
    \left|
    \lossnu(\probfunc(\covariate), \theta(\covariate), y)
    -   \lossnu(\nuxgiven{\theta}, \theta(x), y)\right|
    \le C \linf{\probfunc - \nugiven{\theta}}.
  \end{align*}
  Again applying Lemma~\ref{lemma:two-sup-norm} and
  Assumption~\ref{assumption:lebesgue-equiv}, we
  see
  condition~\eqref{eqn:loss-modulus-prob} holds with
  $s = 2\probsmooth / (2\probsmooth+ \covdim)$.
\end{proof}


\section{Proofs for properties of bounds on the ATE}
\subsection{Proof of the consistency of the ATE estimator}
\label{sec:proof-consistency}

\thmconsistency*
\begin{proof}
  To establish the convergence of $\what{\mu}_1^-$, we use the cross-fitting
  construction to split $\what{\mu}_1^-$ into $K$ estimators and establish
  the convergence of each separately. Because $K$ is fixed, $\what{\mu}_1^-$
  will converge if each does. For the estimator in the $k$-th fold, the
  samples used to estimate $(\what{\theta}_{1,k}, \what{\nu}_{1,k},
  \what{e}_{1,k})$ are independent of the samples in $\mc{I}_k$ and
  hence the sum
  \begin{equation*}
    \what{\mu}_{1,k}^- \defeq
    \frac{1}{|\mathcal{I}_k|} \sum_{i \in \mathcal{I}_k} Z_i Y_i
    + (1-Z_i)\what{\theta}_{1,k}(X_i) +
    Z_i \frac{\psi_{\what{\theta}_{1,k}(X_i)}(Y_i)(1-\what{e}_{1,k}(X_i))}{
      \what{\nu}_{1,k}(X_i)\what{e}_{1,k}(X_i)}.
  \end{equation*}
  For notational simplicity, let $m = |\mathcal{I}_k|$. By
  construction, $m \approx n/K$, so establishing convergence as $m \to
  \infty$ among the samples in $\mathcal{I}_k$ is equivalent to establishing
  the convergence as $n \to \infty$.

  The argument for establishing consistency is relatively
  standard~\cite{chernozhukov2018double}. The only challenge is that
  $\what{\nu}_{1,k}$ (recall the
  definition~\eqref{eqn:weight-normalization-def-a} and
  estimator~\eqref{eqn:opt-emp-prob}) changes as a function of $m$, so that
  for $i\in \mathcal{I}_k$
  \begin{equation*}
    Z_i Y_i + (1-Z_i)\what{\theta}_{1,k}(X_i) + Z_i \frac{\psi_{\what{\theta}_{1,k}(X_i)}(Y_i)(1-\what{e}_{1,k}(X_i))}{\what{\nu}_{1,k}(X_i)\what{e}_{1,k}(X_i)}
  \end{equation*}
  are only i.i.d.\ conditional on $\what{\nu}_{1,k}$. Define the
  $\sigma$-algebra $\mc{F}_{\infty, k}$ generated by the samples in the set
  $\mc{I}_{-k}$ as $n \to \infty$, so that the elements in the preceding
  display are i.i.d.\ conditional on $\mc{I}_{-k}$.

  Because $\P(|\what{\mu}_1^- - \mu_1^-|> \epsilon \mid \mc{F}_{\infty, k})
  \cas 0$ implies $\P(|\what{\mu}_1^- - \mu_1^-|> \epsilon ) \to 0$ by
  dominated convergence, it suffices to show that $\what{\mu}_1^- \cp
  \mu_1^-$ conditionally almost surely on $\mc{F}_{\infty, k}$. To check
  convergence conditional on $\mc{F}_{\infty, k}$, note that the elements in
  the preceeding display form a triangular array for which the weak law of
  large numbers still holds. Then, we apply the following weak
  law for triangular arrays.
  \begin{thm}[Dembo~\cite{Dembo16}, Corollary 2.1.14]
    \label{thm:wlln-tri}
    Suppose that for each $m$, the random variables
    $\xi_{m,i},\,i=1,\dots,m$ are pairwise independent, identically
    distributed for each $m$, and $\E[|\xi_{m,1}|] < \infty$.  Then setting
    $S_m = \sum_{i = 1}^m \xi_{m,i}$ and $a_m = \sum_{i=1}^m
    \E\overline{\xi}_{m,i}$,
    \begin{equation*}
      m^{-1}(S_m - a_m) \cp 0~\mbox{as}~m \to \infty.
    \end{equation*} 
  \end{thm}

  \newcommand{\expectinfty}{\E_{\infty,k}}
  
  For the constant $\constlow > 0$ in
  Assumption~\ref{assumption:bounded-variance}. Define the event
  \begin{equation*}
    A \defeq \left\{ \inf_{x} \what{e}_{1,k}(x) \ge \epsilon,
    \inf_{x} \what{\nu}_{1,k}(x) \ge \constlow,
    ~\mbox{and}~
    \norms{\theta_1 - \what{\theta}_{1,k}}_{1,P} \le 1
    \right\},
  \end{equation*}
  which is $\mc{F}_{\infty,k}$-measurable. Let
  \begin{equation*}
    \xi_{m,i} \defeq \ind{A} \left( Z_i Y_i + (1-Z_i)\what{\theta}_{1,k}(X_i)
    + Z_i \frac{\psi_{\what{\theta}_{1,k}(X_i)}(Y_i)(1-\what{e}_{1,k}(X_i))}{
      \what{\nu}_{1,k}(X_i)\what{e}_{1,k}(X_i)}\right),
  \end{equation*}
  so that $\what{\mu}_{1,k}^- \ind{A} = \frac{1}{m} \sum_{i=1}^m \xi_{m,i}$.
  Assume without loss of generality that $\mathcal{I}_k = \{1,\dots,m\}$.
  Defining the conditional expectation $\expectinfty[\cdot] \defeq \E[\cdot
    \mid \mc{F}_{\infty, k}]$ for shorthand, we have
  \begin{align*}
    \expectinfty[|\xi_{m,i}|]
    & = \ind{A}\expectinfty\left[\bigg| Z Y(1) + (1-Z)\what{\theta}_{1,k}(X)
      + Z \frac{\psi_{\what{\theta}_{1,k}(X)}(Y(1))(1-\what{e}_{1,k}(X))}{
        \what{\nu}_{1,k}(X)\what{e}_{1,k}(X)}
      \bigg|\right] \\
    & \le \ind{A}\left(\expectinfty\left| Y(1) \right| +
    \expectinfty\left|\what{\theta}_{1,k}(X)\right| + \expectinfty\left|
    \frac{\psi_{\what{\theta}_{1,k}(X)}(Y(1))(1-\what{e}_{1,k}(X))}{\what{\nu}_{1,k}(X)\what{e}_{1,k}(X)}
    \right| \right) \\
    & \le \ind{A} \left( \expectinfty\left| Y(1) \right| +
    \expectinfty\left|\what{\theta}_{1,k}(X)\right| + \frac{1}{\epsilon}\expectinfty\left|
    \psi_{\what{\theta}_{1,k}(X)}(Y(1)) \right|\right) \\
    & \le \ind{A} \left(
    \expectinfty\left| Y(1) \right| + \expectinfty\left|\what{\theta}_{1,k}(X)\right| +
    \frac{\Gamma}{\epsilon}\expectinfty\left| Y(1) \right| +
    \frac{\Gamma}{\epsilon}\expectinfty\left|\what{\theta}_{1,k}(X)\right|\right) \\
    & < \infty,
  \end{align*}
  because $\indic{A} \expectinfty|\what{\theta}_{1,k}(X)| \le \indic{A}
  (\E|\theta_1(X)| + 1) < \infty$ and $\expectinfty|\theta_1(X)| \le \Gamma
  \expectinfty[|Y(1)|]
  = \E[|Y(1)|] < \infty$, all
  condtitionally almost surely on $\mc{F}_{\infty,k}$.

  Adopt the notation $S_m = \sum_{i \in \mathcal{I}_k} \xi_{m,i}$ and
  $a_m = \expectinfty[S_m]$.
  Theorem~\ref{thm:wlln-tri} implies that $\frac{1}{m}(S_m - a_m) \cp 0$
  (conditionally a.s.). Next, we show that $\frac{1}{m}a_m -
  \expectinfty[\xi_{m,1}] = \expectinfty[\overline{\xi}_{m,1}] -
  \expectinfty[\xi_{m,1}] \cp 0$. Recall that $|\xi_{m,1}|\ind{|\xi_{m,1}| >
    m} \to 0$ a.s.\ and
  $\expectinfty[|\xi_{m,1}|\ind{|\xi_{m,i}| > m}] \le
  \expectinfty[|\xi_{m,1}|] < \infty$, so dominated convergence
  implies $\expectinfty[|\xi_{m,1}|\ind{|\xi_{m,i}| > m}] \to
  0$. Together, these imply that
  \begin{align*}
    \lefteqn{\ind{A} \what{\mu}_{1,k}^- - \expectinfty[\xi_{m,i}]} \\
    & = \ind{A}\frac{1}{m} \sum_{i=1}^m
    \left(Z_i Y_i + (1-Z_i)\what{\theta}_{1,k}(X_i) + Z_i \frac{\psi_{\what{\theta}_{1,k}(X_i)}(Y_i)(1-\what{e}_{1,k}(X_i))}{\what{\nu}_{1,k}(X_i)\what{e}_{1,k}(X_i)}
    \right)
    - \expectinfty[\xi_{m,i}] \cp  0
  \end{align*}
  conditionally almost surely, and
  $\ind{A} \what{\mu}_{1,k}^- - \what{\mu}_{1,k} \cp 0$ conditionally
  almost surely, as $\P(A) \to 1$.
  Jensen's inequality,
  Assumption~\ref{assumption:bounded-variance}(d,e), and that
  $\E[\psi_{\theta(X)}(Y(1)) \mid Z = 1, X] = 0$ almost surely
  allow us to bound the error
  \begin{align*}
    \lefteqn{|\expectinfty[\xi_{m,i}] - \mu_1^-|} \\
    & = \left|\ind{A}
    \expectinfty\left[(1 - Z) \left(\what{\theta}_{1,k}(X)
      - \theta_1(X)\right) + Z \frac{\psi_{\what{\theta}_{1,k}}(Y(1))
        (1 - \what{e}_{1,k}(X))}{\what{\nu}_{1,k}(X) \what{e}_{1,k}(X)}\right]
    \right| + o_P(1) \\
    &=
    \begin{aligned}[t]
      &\Bigg|\ind{A}\expectinfty\Bigg[ (1-Z)\left(\what{\theta}_{1,k}(X)
        - \theta_{1}(X)\right) + Z \frac{\left(\psi_{\what{\theta}_{1,k}(X)}(Y(1))-\psi_{{\theta}_{1}(X)}(Y(1))\right)(1-\what{e}_{1,k}(X))}{\what{\nu}_{1,k}(X)\what{e}_{1,k}(X)}  \\
        &- Z \frac{\psi_{{\theta}_{1}(X)}(Y(1))(1-\what{e}_{1,k}(X))}{\what{\nu}_{1,k}(X)\what{e}_{1,k}(X)} \Bigg] \Bigg| + o_P(1)
    \end{aligned} \\
    & \le
    \expectinfty\left[
      \left|\what{\theta}_{1,k}(X)- \theta_1(X)\right| \right]
    + \expectinfty\left[ \frac{\Gamma}{c \epsilon}\left|\psi_{\what{\theta}_{1,k}(X)}(Y(1))-\psi_{{\theta}_{1}(X)}(Y(1))\right| \right]
    + o_P(1),
  \end{align*}
  where the $o_P(1)$ term comes from the fact that
  Assumption~\ref{assumption:bounded-variance}(b,d,e) imply $\P(A) \to
  1$. Finally, the Lipschitz continuity $|\psi_t(y) - \psi_s(y)| \le
  \Gamma|s - t|$
  and Assumption~\ref{assumption:bounded-variance}(b)
  that $\norms{\what{\theta}_{1,k} - \theta_1}_{1,P} \cp 0$ give that
  $|\expectinfty[\xi_{m,i}] - \mu_1^-| \cp 0$.
\end{proof}

\subsection{Proof of Theorem~\ref{thm:semiparametric}}
\label{sec:proof-semiparametric}

\newcommand{\subopt}{_\star}
\newcommand{\score}{\dot{\ell}}

\semiparametricnormality*

The proof depends heavily on Theorems 3.1 and 3.2 and Corollary 3.1 of
\citet{chernozhukov2018double}; primarily, we check their
Assumptions 3.1 and 3.2 for our proposed estimator. Stating their
results requires a bit of notation, which we
introduce here briefly.
We begin with their general assumptions
about a score function. Their result allows multi-dimensional
estimates, so we allow a $d$-dimensional score in the assumptions.
\begin{assumption}[\citet{chernozhukov2018double}, Assumption 3.1]
  \label{assumption:dml-score}
  Let $\mc{W}$ be a measurable space and
  $\mc{T}$ be a collection of (nuisance) functions mapping
  $\mc{W} \to \R^{d_0}$. Let $d \in \N$ and
  $\Theta \subset \R^d$, and let
  $m : \mc{W} \times \Theta \times \mc{T} \to \R^d$ be a score function.
  Let $\theta\subopt \in \Theta$ be the true parameter of interest
  and $\eta\subopt \in \mc{T}$ be the true nuisance.
  
  There exist $0 < c_0 \le c_1 < \infty$ such that for all $n \ge 3$, the
  following conditions hold.
  \begin{enumerate}[(a)]
  \item \label{item:mean-true-param}
    The true parameter $\theta\subopt$ obeys
    \begin{equation*}
      \E\left[\score(W, \theta\subopt, \eta\subopt) \right] = 0.
    \end{equation*}
    
  \item \label{item:linearization-score}
    There exist functions $\score^a$ and $\score^b$ such that
    the score $m$ is linear in $\theta$, satisfying
    \begin{equation*}
      \score(w, \theta, \eta) = \score^a(w; \eta)\theta + \score^b(w; \eta),
    \end{equation*}
    for all $w \in \mathcal{W}, \theta \in \Theta, \eta \in \mathcal{T}$.

  \item \label{item:twice-gateaux}
    The map $\eta \mapsto \E[\score(W, \theta, \eta)]$ is twice
    continuously Gateaux-differentiable on $\mathcal{T}$.

  \item \label{item:neyman-orthogonality}
    The score function $\score$ obeys the Neyman orthogonality
    condition~\cite[Def.~2.1]{chernozhukov2018double}, i.e.\ the
    derivative $\frac{\dif{}}{\dif{r}} \E[ \score(W, \theta\subopt, \eta\subopt
      + r(\eta - \eta\subopt)]$ exists for all $\eta \in \mathcal{T}$ and
    $r$ near 0, and
    \begin{equation}
      \label{eqn:neyman-orthogonality}
      \frac{\dif{}}{\dif{r}} \E\left[ \score(W, \theta\subopt, \eta\subopt +
        r(\eta - \eta\subopt)) \right]
      \bigg|_{r=0} = 0.
    \end{equation}
  \item\label{item:linear-score-matrix}
    Let $J_0 \defeq \E[\score^a(W; \eta\subopt)]$. The singular values of
    $J_0$ lie in $[c_0, c_1]$.
  \end{enumerate}
\end{assumption}
\noindent
We also require assumptions on properties of the nuisance variables.
\begin{assumption}[\citet{chernozhukov2018double}, Assumption 3.2]
  \label{assumption:dml-nuisance}
  Let the notation of Assumption~\ref{assumption:dml-score} hold.  Let
  $\delta_n$ and $\Delta_n$ be sequences with $\delta_n \ge n^{-1/2}$ and
  $\lim_{n\to\infty} \delta_n = \lim_{n\to\infty} \Delta_n = 0$.  There
  exist $0 < \epsilon$, $0 < c_0, c_1 < \infty$ and a set $\mc{T}_n$ with
  the following properties.
  \begin{enumerate}[(a)]
  \item \label{item:def-of-tn}
    Given a random subset $I$ of $[n]$ of size $ n / K$, the nuisance
    parameter $\what{\eta}$
    estimated on $\{W_i\}_{i \in I^c}$ belongs to $\mathcal{T}_n$
    with probability at least $1-\Delta_n$. $\mc{T}_n$ contains
    $\eta\subopt$.
  \item \label{item:moment-conditions}
    The following moment conditions hold:
    \begin{align*}
      m_n & = \sup_{\eta \in \mathcal{T}_n}\left( \E[\|\score(W, \theta\subopt,
        \eta)\|^q]
      \right)^{1/q} \le c_1
      ~~ \mbox{and} ~~
      m_n' = \sup_{\eta \in \mathcal{T}_n}\left( \E[\|\score^a(W; \eta)\|^q]
      \right)^{1/q} \le c_1.
    \end{align*}
  \item \label{item:score-rates}
    Define the rates
    \begin{align*}
      r_n & \defeq \sup_{\eta \in \mathcal{T}_n} \E[\|\score^a(W; \eta)
        - \score^a(W; \eta\subopt)\|], ~~
      r_n' \defeq \sup_{\eta \in \mathcal{T}_n} \E[\|\score(W, \theta\subopt, \eta)
        - \score(W, \theta\subopt, \eta\subopt)\|^2]^{\half}
      \\
      \lambda_n' & \defeq \sup_{0 < r < 1, \eta \in \mathcal{T}_n}
      \left\|\E\left[\frac{d^2}{dr^2}
        \score(W, \theta\subopt, \eta\subopt + r(\eta -\eta\subopt)) \right]
      \right\|.
    \end{align*}
    Then $r_n \le \delta_n$, $r_n' \le \delta_n$, and $\lambda_n'
    \le \delta_n / \sqrt{n}$.
  \item \label{item:score-variance}
    The score has positive variance:
    $\lambda_{\min}(\E[ \score(W, \theta\subopt, \eta\subopt)
      \score(W, \theta\subopt, \eta\subopt)^\top]) \ge c_0$.
  \end{enumerate}
\end{assumption}

It is possible to relax Assumption~\ref{assumption:dml-score}(d), the
equality~\eqref{eqn:neyman-orthogonality}, to require only that
$|\frac{\partial}{\partial r} \E[\score(W, \theta\subopt, \eta\subopt + r(\eta -
  \eta\subopt))]|_{r = 0}| \le \lambda_n$ for $\lambda_n \to 0$ sufficiently
quickly; we shall not need such generality.  With these assumptions in
place, we have the following theorem.

\begin{thm}[Chernozhukov et al.~\cite{chernozhukov2018double}, Theorems 3.1
  and 3.2]
  \label{thm:dml-score}
  \newcommand{\dmlrest}{\what{\theta}_{\textup{DML}}} Let
  Assumptions~\ref{assumption:dml-score} and~\ref{assumption:dml-nuisance}
  hold. In addition, suppose that $\delta_n \ge n^{-1/2}$ for all $n \ge
  1$. Let the sets $\{\mc{I}_k\}_{k=1}^K$ equally (at random) partition
  $n$. Define the empirical estimators $\what{\eta}_k$ based on $i \not
  \in \mc{I}_k$ and the double-robust-double-machine-learning (DMLR2)
  estimator $\dmlrest$ as the root of
  \begin{equation*}
    \frac{1}{n} \sum_{k = 1}^K
    \sum_{i \in \mc{I}_k} \score(W_i, \dmlrest, \what{\eta}_k)
    = 0.
  \end{equation*}
  Then
  $\dmlrest$ concentrates in a $1/\sqrt{n}$ neighborhood of
  $\theta\subopt$ and is approximately linear and centered Gaussian:
  \begin{equation*}
    \sqrt{n}(\dmlrest - \theta\subopt)
    = \frac{1}{\sqrt{n}} \sum_{i=1}^n \score(W_i, \theta\subopt, \eta\subopt)
    + O_P(\rho_n) \cd \normal(0, \Sigma),
  \end{equation*}
  where the
  remainder term $\rho_n = n^{-1/2} + r_n + r_n' + n^{1/2}\lambda_n +
  n^{1/2}\lambda_n'$ and the variance $\Sigma = J_0^{-1}\E[ \score(W,
    \theta\subopt, \eta\subopt) \score(W, \theta\subopt,
    \eta\subopt)^\top] J_0^{-1 \top}$, where $J_0$ is as in
  Assumption~\ref{assumption:dml-score}(\ref{item:linear-score-matrix}).

  If additionally $\delta_n \ge n^{-(1 - 2/q) \wedge 1/2}$
  and
  \begin{equation*}
    \what{\Sigma}
    = \frac{1}{n}
    \sum_{k = 1}^K \sum_{i \in \mc{I}_k} \what{J}_0^{-1}
    \score(W_i, \dmlrest, \what{\eta}_k) \score(W_i,
    \dmlrest, \what{\eta}_k)^\top \what{J}_0^{-\top}
    ~~ \mbox{and} ~~
    \what{J}_0 = \frac{1}{n} \sum_{k = 1}^K \sum_{i \in \mc{I}_k}
    \score^a(W_i; \what{\eta}_k)
  \end{equation*}
  then
  \begin{equation}
    \sqrt{n} \what{\Sigma}^{-\half}
    (\dmlrest - \theta\subopt)
    = -\frac{1}{\sqrt{n}} \sum_{i = 1}^n
    \Sigma^{-\half} J_0^{-1} \score(W_i, \theta\subopt, \eta\subopt)
    + O_P(\rho_n) \cd \normal(0, I_d)
  \end{equation}
  where $\rho_n = n^{-(1 - 2/q)\wedge 1/2} + r_n + r_n'
  + n^{1/2} \lambda_n + n^{1/2} \lambda_n'$.
\end{thm}






We now construct a score function for our estimator $\what{\mu}_1^-$ and
demonstrate that it satisfies the conditions necessary
to apply Theorem~\ref{thm:dml-score}.
Let $Y = Y(Z)$ be the observed potential outcome.  Define the
triple $W_i = (Y_i,
X_i, \treatmentrv_i)$ as the $d+2$ dimensional random vector containing
all the observed random variables. Similarly let $w = (y, x, \treatment)$
for a fixed variable; we define our score
\begin{align}
  \label{eq:orthogonal-score}
  \score(w, \mu, \eta) \defeq  z y + (1-z)\funcparam(x) - \mu
  + z\frac{\hinge{y-\funcparam(x)} - \Gamma\neghinge{y - \funcparam(x)}}{
    \nu(x)}\frac{1-e(x)}{e(x)},
\end{align}
where the nuisance parameter $\eta \defeq (\funcparam, \nu, e)$.  First,
it is clear by definition~\eqref{eq:orthogonal-estimator} of
$\what{\mu}_1^-$ it is the root (in $\mu$) of $\sum_{k = 1}^K \sum_{i \in
  \mc{I}_k} \score(W_i, \mu, \what{\eta}_{1,k}) = 0$, so that the
definitional part of Theorem~\ref{thm:dml-score} holds. Additionally,
by construction $\what{\eta}_{1,k}$ is estimated based on
$[n] \setminus \mc{I}_k$.

We now turn to verifying that Assumptions~\ref{assumption:dml-score}
and \ref{assumption:dml-nuisance} hold for the estimator $\what{\mu}_1^-$.
To that end, define the space
\begin{equation}
  \label{eqn:funcparam-space}
  \mc{T} \defeq \{ \eta = (\theta, \nu, e) \mid
  \eta~\text{measurable},~1 \le \nu(x) \le \Gamma,
  \epsilon \le
  e(x) \le 1-\epsilon ~ \mbox{for~} x \in \mc{X}\},
\end{equation}
so that $\eta_1 \defeq (\theta_1, \nu_1, e_1) \in \mathcal{T}$
(recall that
$\theta_1(x) = \E[Y(1) \mid X = x, Z = 1]$,
$e_1(x) = \P(Z = 1 \mid X = x)$, and 
$\nu_1(x) = P(Y(1) \ge \theta_1(x) \mid Z = 1, X= x)
+ \Gamma P(Y(1) < \theta_1(x) \mid Z = 1, X=x)$).
We verify the assumptions in Sections~\ref{sec:verify-dml-score}
and \ref{sec:verify-dml-nuisance}, respectively.
Once we verify these, the proof of Theorem~\ref{thm:semiparametric}
is complete, as the confidence interval statements are immediate.

\subsubsection{Verifying Assumption~\ref{assumption:dml-score}}
\label{sec:verify-dml-score}

By construction, the score satisfies
\begin{align*}
  \E[\score(W, \mulower, \eta_1)] = & \E[ZY(1) + (1-Z)\funcparam_1(X)] - \mu_1^- \\
  &  + \E\left[\frac{\E[Z(\hinges{Y(1) - \funcparam_1(X)}-
        \Gamma \neghinges{Y(1) - \funcparam_1(X)}) \mid X]}{\nu_1(X)}
    \frac{e_0(X)}{e_1(X)}\right] = 0
\end{align*}
because $\E[Z(\hinges{Y(1) - \funcparam_1(X)} - \Gamma \neghinges{Y(1) -
    \funcparam_1(X)}\mid X] =
\E[Z \mid X]
\E[\psi_{\funcparam_1(X)}(Y(1)) \mid X = x, Z = 1] = 0$ almost everywhere
by Lemma~\ref{lemma:opt-is-good}, $\nu_1(x) \ge 1$, and
$e_1(x) > \epsilon$. Thus Assumption
\ref{assumption:dml-score}\eqref{item:mean-true-param} holds. The
score~\eqref{eq:orthogonal-score} is linear in $\mu$, so
Assumption~\ref{assumption:dml-score}\eqref{item:linearization-score}
holds with $\score^a\equiv - 1$.

We turn to verifying the twice Gateaux differentiability with respect to
$\funcparam$ in
Assumption~\ref{assumption:dml-score}\eqref{item:twice-gateaux}.  Because
$\psi_t(Y(1)) \le \Gamma Y(1)$, and $\E[|Y(1)|^q] < \infty$, dominated convergence
will hold for all interchanges of $\E$ and differentiation in what
follows.  First note that $\nu(x)$ and $e(x)$ are bounded below and above
by definition~\eqref{eqn:funcparam-space} of $\mc{T}$. Therefore, it
suffices to prove that for an arbitrary measurable function $g$ with
values in $[-\epsilon^{-1}, \epsilon^{-1}]$,
\begin{equation*}
  \funcparam \mapsto \left[
    \E\left[ g(x)Z\psi_{\theta(x)}(Y(1))  \mid X=x  \right]\right]_{x \in \mc{X}}
\end{equation*}
is twice Gateaux differentiable for $\funcparam \in \mc{T}$. When $Y(1)|Z=1,
X=x$ has a density $p_{Y(1)}(y\mid X=x, Z=1),$ for almost every $x$ and $P(Y(1)=y
\mid X=x,Z=1)=0$,
\begin{equation}
  \label{eqn:first-psi-derivative}
  \frac{\dif}{\dif{t}} \E[ Z\psi_t(Y) \mid Z=1, X=x] = -P_1(Y(1) > t \mid X=x) - \Gamma P_1(Y(1) < t \mid X=x),
\end{equation}
and so
\begin{align*}
  \frac{\dif{}^2}{\dif{t}^2} \E[ Z\psi_t(Y) \mid Z=1, X=x ]
  & =
  (1-\Gamma)p_{Y(1)}(t \mid X=x,Z=1).
\end{align*}
Coupled with the boundedness of $g(\cdot)$, this implies twice Gateaux
differentiability, satisfying
Assumption~\ref{assumption:dml-score}\eqref{item:twice-gateaux} for
$\theta_1(x)$. As $(\nu, e) \in \mc{T}$ are uniformly bounded, twice
Gateaux differentiability with respect to the remaining terms is immediate
by~\citet[Proof of Thms.~5.1 \& 5.2, step 1]{chernozhukov2018double}.

Next, we verify Neyman orthogonality,
Assumption~\ref{assumption:dml-score}\eqref{item:neyman-orthogonality}.
Note that if $y \neq \theta_1(x)$, then
\begin{align}
  \lefteqn{\left.\frac{\partial}{\partial r}
    \score(w, \mu, \eta_1 + r(\eta - \eta_1))\right|_{r = 0}} \nonumber \\
  & = (1 - z) (\theta(x) - \theta_1(x))
  - z \frac{\indic{y > \theta_1(x)}
    + \Gamma \indic{y < \theta_1(x)}}{\nu_1(x)}
  \cdot (\theta(x) - \theta_1(x))
  \cdot \frac{1 - e_1(x)}{e_1(x)} \nonumber \\
  & \qquad + z \frac{\psi_{\theta_1(x)}(y)}{
    \nu_1(x)^2 e_1(x)^2}(1 - e_1(x))
  \left.\frac{\partial}{\partial r}
  (\nu_1(x) + r (\nu(x) - \nu_1(x)))(
  (e_1(x) + r (e(x) - e_1(x)))\right|_{r = 0} \nonumber \\
  & \qquad ~ - z \frac{\psi_{\theta_1(x)}(y)}{\nu_1(x) e_1(x)}
  (e(x) - e_1(x)) \nonumber \\
  & = (1 - z) (\theta(x) - \theta_1(x))
  - z \frac{\indic{y > \theta_1(x)}
    + \Gamma \indic{y < \theta_1(x)}}{\nu_1(x)}
  \cdot (\theta(x) - \theta_1(x))
  \cdot \frac{1 - e_1(x)}{e_1(x)} \nonumber \\
  & \qquad + z \frac{\psi_{\theta_1(x)}(y)}{
    \nu_1(x)^2 e_1(x)^2}(1 - e_1(x))
  \left[e_1(x) (\nu(x) - \nu_1(x))
    + \nu_1(x) (e(x) - e_1(x))\right]
  \label{eqn:big-gateaux-expansion}
  \\
  & \qquad ~ - z \frac{\psi_{\theta_1(x)}(y)}{\nu_1(x) e_1(x)}
  (e(x) - e_1(x)).
  \nonumber
\end{align}
As $\nu, e \in \mc{T}$ are uniformly bounded, $\psi_t(\cdot)$ is Lipschitz,
and $Y(1)$ has $q > 2$ moments,
the dominated convergence theorem implies
\begin{equation*}
  \left.\frac{\partial}{\partial r}
  \E[\score(W, \mu, \eta_1 + r (\eta - \eta_1))]\right|_{r = 0}
  = \E\left[\left.\frac{\partial}{\partial r}
    \score(W, \mu, \eta_1 + r (\eta - \eta_1))\right|_{r = 0}
    \right].
\end{equation*}
Now, note that
by definition~\eqref{eqn:weight-normalization-def-a}, we have
\begin{align*}
  \lefteqn{\E\left[Z \left(\indic{Y(1) > \funcparam_1(X)}
      + \Gamma \indic{Y(1) < \funcparam_1(X)} \right) \mid X \right]} \\
  & =
  \E\left[Z \left(P(Y(1) > \funcparam_1(X) \mid X, Z = 1)
    + \Gamma P(Y(1) < \funcparam_1(X) \mid X, Z = 1)\right)
    \mid X \right]
  = \E[Z \mid X] \nu_1(X).
\end{align*}
Substituting Eq.~\eqref{eqn:big-gateaux-expansion} for the derivative
of the score and applying iterated expectations
with $\E[\cdot \mid X = x]$,
we thus obtain
\begin{align*}
  \lefteqn{\E\left[\left.\frac{\partial}{\partial r}
      \score(W, \mu, \eta_1 + r (\eta - \eta_1))\right|_{r = 0}
      \right]} \\
  & = \E\left[e_0(X) (\theta(X) - \theta_1(X))\right]
  - \E\left[\frac{e_1(X) \nu_1(X) (\theta(X) - \theta_1(X))}{\nu_1(X)}
    \frac{e_0(X)}{e_1(X)}\right] \\
  & \quad
  + \E\left[\frac{\E[Z \psi_{\funcparam_1(X)}(Y(1)) \mid X]}{
      \nu_1(X)^2 e_1(X)^2}
    e_0(X)
    \left[e_1(X) (\nu(X) - \nu_1(X))
      + \nu_1(X) (e(X) - e_1(X))\right]\right] \\
  & \quad
  - \E\left[\frac{\E[Z \psi_{\funcparam_1(X)}(Y(1))
        \mid X]}{e_1(X) \nu_1(X)} (e(X) - e_1(X))\right] \\
  & = \E\left[e_0(X) (\theta(X) - \theta_1(X))\right]
  - \E\left[e_0(X) (\theta(X) - \theta_1(X))\right]
  = 0,
\end{align*}
where we have used that $\E[Z \psi_{\funcparam_1(X)}(Y(1)) \mid X] = 0$ with
probability 1 by Lemma~\ref{lemma:opt-is-good} and that $Y(1) = \theta_1(X)$ with probability 0.





Finally, Assumption~\ref{assumption:dml-score}\eqref{item:linear-score-matrix}
clearly holds, as $J_0 = \score^a = -1$.

\subsubsection{Verifying Assumption~\ref{assumption:dml-nuisance}}
\label{sec:verify-dml-nuisance}

Assumption~\ref{assumption:nuisance-est} establishes that
$\what\funcparam$, $\what{\nu}$ and $\what{e}$ satisfy $\|
\what\funcparam(\cdot) - \funcparam_1 (\cdot)\|_{2,P} = o_P(n^{-1/4})$, $\|
\what\nu(\cdot) - \nu_1(\cdot) \|_{2,P} = o_P(n^{-1/4})$, and $\|
\what{e}(\cdot) - e_1(\cdot) \|_{2,P} = o_P(n^{-1/4})$. Therefore, there
exists sequences $a_n \to 0$ and $\Delta_n' \to 0$ such that for each $n$,
\begin{equation*}
  \| \what{\funcparam}(\cdot) - \funcparam_1(\cdot) \|_{2,P} \le a_n n^{-1/4},
  ~~
  \| \what{\nu}(\cdot) - \nu_1(\cdot) \|_{2,P} \le a_n n^{-1/4},
  ~~ \mbox{and} ~~
  \| \what{e}(\cdot) - e_1(\cdot) \|_{2,P} \le a_n n^{-1/4}
\end{equation*}
with probability $1-\Delta_n'/2$. We may choose $a_n$ so that these hold when
$\what{\eta}$ is estimated using only $(1- \frac{1}{K})n$ (as opposed to
$n$) samples. Similarly, Assumption~\ref{assumption:nuisance-est} implies
that there exists a constant $C_1$ such that
\begin{equation*}
  \| \what{\funcparam}(\cdot) - \funcparam_1(\cdot) \|_{\infty,P} \le C_1,
  ~~
  \| \what{\nu}(\cdot) - \nu_1(\cdot) \|_{\infty,P} \le C_1,
  ~~ \mbox{and} ~~
  \| \what{e}(\cdot) - e_1(\cdot) \|_{\infty,P} \le C_1
\end{equation*}
with probability $1-\Delta_n'/2$.
Now, for constants
$\constlow$ and $\constup$ to be chosen,
we define the relevant
set $\mc{T}_n$ by
\begin{equation}
  \label{eqn:def-Tn}
  \begin{split}
    \mathcal{T}_n
    \defeq
    \big\{ \eta = (\funcparam, e, \nu) & :
    \norm{\eta - \eta_1}_{2,P} \le a_n n^{-1/4},
    ~ \norm{\eta - \eta_1}_{\infty,P} \le C_1,
    \\
    &
    ~~ \epsilon \le e \le 1 - \epsilon,
    ~ \constlow \le \nu \le \constup \big\}.
  \end{split}
\end{equation}
By Assumption~\ref{assumption:bounded-variance}, there exists $\Delta_n''$ such that $P(\epsilon \le \inf_x \what{e}(x) \le \sup_x \what{e}(x) \le 1-\epsilon) \ge 1- \Delta_n'' / 2$ such that $P(\constlow \le \inf_x \what{\nu}(x) \le \sup_x \what{\nu}(x) \le \constup \ge 1- \Delta_n''/2$, and $\Delta_n'' \to 0$.
Then $P(\what\eta \in \mathcal{T}_n)\ge 1-\Delta_n'/2-\Delta_n'/2-\Delta_n''=1-\Delta_n$, for $\Delta_n = \Delta_n' + \Delta_n''$
by Assumptions~\ref{assumption:bounded-variance} and
\ref{assumption:nuisance-est}, so
Assumption~\ref{assumption:dml-nuisance}\eqref{item:def-of-tn} holds.

To bound the moments in
Assumption~\ref{assumption:dml-nuisance}\eqref{item:moment-conditions},
first we bound the score at the true nuisance parameter $\eta_1$, 
applying the triangle inequality to bound the difference. By
the equality~\eqref{eq:population-theta} defining $\funcparam_1$,
for any $x$ and $\delta > 0$, there exists
$L^\delta_x(y)$ such that
\begin{equation*}
  \theta_1(x) \in \left[ \E[L^\delta_x(Y(1)) Y(1) \mid X=x] - \delta,
    \E[L^\delta_x(Y(1)) Y(1) \mid X = x]\right]
\end{equation*}
where $\E[L^\delta_x(Y(1))] = 1$, and $0\le L_x^\delta(y) \le \Gamma
L_x^\delta(\wt{y})$ for all $y,\wt{y}$.  Together, these imply
$\Gamma^{-1} \le L^\delta_x(y) \le \Gamma$. Therefore,
Assumption~\ref{assumption:problem-regularity}(b) and H\"{o}lder's
inequality imply that for $C_q = \E[|Y(1)|^q]$,
\begin{align}
  \nonumber
  2^{1 - q}
  \E\left[ |\theta_1(X)|^q \right]
  & \le \E\left[ \left| \E\left[ L^\delta_X(Y(1)) Y(1) | X\right] \right|^q + \delta^q \right]
  \le
  \E\left[ \Gamma^q \E\left[ \left|Y(1)\right|^q \mid X\right] + \delta^q \right]
  \\
  & = \Gamma^q \E\left[ \left|Y(1)\right|^q \right] + \delta^q
  = \Gamma^q C_q + \delta^q.
  \label{eqn:bound-theta-1-q-moment}
\end{align}
To bound the moments of the score, we therefore have
\begin{align*}
  \lefteqn{4^{1 - q} \E\left[|\score(W, \mulower, \eta_1)|^q\right]} \\
  & \le
  |\mulower|^q
  +   \E\left[|ZY(1)|^q + (1-Z)|\funcparam_1(X)|^q +Z\biggm|\frac{\psi_{\funcparam_1(X)}(Y(1))( 1 - e_1(X))}{\nu_1(X) e_1(X)}\biggm|^q \right] \\
  &\le
  |\mu_1^-|^q + (1-\epsilon)\E_1[|Y(1)|^q] + (1 - \epsilon)\E_0\left[|\funcparam_1(X)|^q\right] + (1-\epsilon) \E_1\left[ |\psi_{\funcparam_1(X)}(Y(1))|^q \right] \\
  &\le
  |\mu_1^-|^q + (1-\epsilon)\E_1[|Y(1)|^q] + (1 - \epsilon)\E_0\left[|\funcparam_1(X)|^q\right] + (1-\epsilon)2^{q-1} \Gamma^q \left( \E_1\left[ |Y(1)|^q \right] +  \E_1\left[ |\funcparam(X)|^q \right]  \right) \\
  & \le \wt{C}_1
\end{align*}
for a finite $\wt{C}_1$.
For $\eta \in \mc{T}_n$, so that $\norm{\eta - \eta_1}_{q, P} \le C_1$
by definition~\eqref{eqn:def-Tn},
we may bound the difference
\begin{align*}
  \lefteqn{2^{1-q}
    \E\left[|\score(W, \mulower, \eta_1)-\score(W, \mulower, \eta)|^q\right]} \\
  &\le
  \E\left[(1-Z)|\funcparam_1(X) - \funcparam(X)|^q +Z\left|\frac{\psi_{\funcparam_1(X)}(Y(1))(1 - e_1(X))}{\nu_1(X) e_1(X)} - \frac{\psi_{\funcparam(X)}(Y(1))(1 - e(X))}{\nu(X) e(X)}\right|^q \right] \\
  & \le C_1^q + \E\left[ Z\left| \frac{\psi_{\theta_1(X)}(Y(1))(1 - e_1(X))\nu(X)e(X)-\psi_{\funcparam(X)}(Y(1))(1-e(X))\nu_1(X)e_1(X)}{\nu_1(X)\nu(X)e_1(X)e(X)} \right|^q\right] \\
  & \stackrel{(i)}{\le}
  C_1^q + \frac{1}{\constlow^2\epsilon^2}
  \E\left[Z \left| \psi_{\funcparam_1(X)}(Y(1))(1 - e_1(X))\nu(X) e(X)
    - \psi_{\funcparam(X)}(Y(1))(1 - e(X)) \nu_1(X) e_1(X) \right|^q \right] \\
  & \stackrel{(ii)}{\le}
  C_1^q + \frac{2^{q-1}}{\constlow^2 \epsilon^2}
  \E\left[Z \left| \left\{\psi_{\funcparam_1(X)}(Y(1))(1 - e_1(X))
    - \psi_{ \funcparam(X)}(Y(1))(1 - e(X))\right\} \nu(X) e(X) \right|^q \right]\\
  & \qquad ~ + \frac{2^{q-1}}{\constlow^2 \epsilon^2}
  \E\left[Z\bigg| \psi_{\funcparam(X)}(Y(1))(1 - e(X))
    \left(\nu(X) e(X) - \nu_1(X) e_1(X)\right) \bigg|^q\right]\\
  & \le \wt{C}_2
\end{align*}
for some constant $\wt{C}_2$, uniformly over $\eta \in \mc{T}_n$.  Above,
inequality~$(i)$ follows by the lower bound conditions~\eqref{eqn:def-Tn} on
$e, \nu \in \mc{T}_n$, inequality $(ii)$ is convexity and the triangle
inequality, and the final inequality follows as
$\sup_x |\nu(x)| < \infty$, $e(\cdot) \subset [\epsilon, 1 - \epsilon]$
for $\nu, e \in \mc{T}_n$, and
\begin{align*}
  \E[Z|\psi_{\theta(X)}(Y(1))|^q]
  & = \E[Z \hinge{Y(1) - \theta(X)}^q]
  + \Gamma^q \E[Z \neghinge{Y(1) - \theta(X)}^q] \\
  & \le 2^q \left((1 + \Gamma^q) \E[|Y(1)|^q]
  + (1 + \Gamma^q) \E[|\theta(X)|^q]\right)
  \le \wt{C}_3
\end{align*}
for a finite $\wt{C}_3$ because $\norm{\theta - \theta_1}_{q,P} \le C_1$ by
definition~\eqref{eqn:def-Tn} of $\mc{T}_n$ and
inequality~\eqref{eqn:bound-theta-1-q-moment}.
Assumption~\ref{assumption:dml-nuisance}\eqref{item:moment-conditions}
follows by the triangle inequality.


We turn to verifying
Assumption~\ref{assumption:dml-nuisance}\eqref{item:score-rates}. Because
$\score^a \equiv -1$, the rate $r_n = 0$. The
construction~\eqref{eqn:def-Tn} of $\mathcal{T}_n$ implies that
\begin{align*}
  \lefteqn{\sup_{\eta \in \mathcal{T}_n} E\left[ \left\{\score(W, \mulower, \eta) - \score(W, \mulower, \eta_1)\right\}^2 \right]^\half} \\
  & \le
  \sup_{\eta \in \mathcal{T}_n}
  \| \funcparam(\cdot) - \funcparam_1(\cdot) \|_{2,P} +
  \frac{\Gamma(1-\epsilon)}{\epsilon} \sup_{\eta \in \mathcal{T}_n} \| \funcparam(\cdot) - \funcparam_1(\cdot) \|_{2,P} \\
  & \quad + 2\frac{\Gamma^2(1-\epsilon)}{\epsilon}\left(\|\funcparam(\cdot)\|_{2,P} + \E[Y(1)^2]\right)
  \cdot \left(\left\| e(\cdot) - e_1(\cdot) \right\|_{2,P} + \left\| \nu(\cdot) - \nu_1(\cdot) \right\|_{2,P}\right)
\end{align*}
Therefore, $r_n'\le \tilde{C}_4 a_n n^{-1/4}$ for a constant $\tilde{C}_4 <
\infty$.  Bounding the second derivative moments
$\lambda_n'$ is more involved;
the next lemma controls this term.
We defer the proof to section~\ref{sec:proof-second-deriv-rate}.

\begin{restatable}{lemma}{lemsecondderiv}
  \label{lem:second-deriv-rate}
  Assume the conditions of Theorem~\ref{thm:semiparametric}. Let
  $\mc{T}_n$ be as defined~\eqref{eqn:def-Tn}.
  Then
  \begin{equation}
    \sup_{r \in (0,1),\eta \in \mathcal{T}_n} \frac{d^2}{dr^2} \E\left[ \score(W, \mulower, \eta_1 + r(\eta - \eta_1)) \right] \le C a_n^2 n^{-1/2}.
    \label{eq:rate-cond}
\end{equation}
\end{restatable}

\noindent In summary,
$\lambda_n' \le C a_n^2/\sqrt{n}$.
Let $\delta_n = \max\{\tilde{C}_3 a_n, \tilde{C}_4 a_n^2, n^{-1/2}\}$.
The sequences $\{a_n\}$, $\{a_n^2\}$, and $\{n^{-1/2}\}$ all
converge to 0, and $\delta_n$ satisfies the conditions
of Assumption~\ref{assumption:dml-nuisance}. Assumption
\ref{assumption:dml-nuisance}\eqref{item:score-rates} follows.

For Assumption~\ref{assumption:dml-nuisance}\eqref{item:score-variance},
note that
\begin{align*}
  \E\left[ \score(W, \mulower, \eta_1)^2 \mid Z=1, X=x \right] &= \var\left[ Y(1) + \frac{\psi_{\funcparam_1(X)}(Y(1))}{\nu_1(X)}\frac{1-e_1(X)}{e_1(X)} \biggm| Z=1, X=x \right] \\
  &\ge \var(Y(1)\mid Z=1, X=x),
\end{align*}
because $\cov(Y(1), \psi_{\funcparam_1(X)}(Y(1))) \ge 0$.
Taking expectations over $X$ completes the proof.

\subsubsection{Proof of Lemma~\ref{lem:second-deriv-rate}}
\label{sec:proof-second-deriv-rate}

\lemsecondderiv*
\begin{proof}
  For notational convenience, and with some abuse,
  throughout this proof for any function
  $f : \mc{X} \times \R \to \R$ of $x \in \mc{X}$ and a scalar
  $r \in \R$, we let $f'(x, r) = \frac{\partial}{\partial r} f(x, r)$
  and $f''(x, r) = \frac{\partial^2}{\partial r^2} f(x, r)$.
  Let 
  \begin{equation}
    \label{eqn:define-shitty-h}
    h(x,r) = \E\left[ Z \psi_{\funcparam_1(X) + r\left\{\funcparam(X) -
        \funcparam_1(X)\right\}}(Y(1))  \mid X=x \right].
  \end{equation}
  Note that $\sup_{r \in (0, 1)} h(x, r)$ is integrable because
  $t \mapsto \psi_t(Y(1))$ is Lipschitz and $\funcparam \in \mc{T}_n$.
  Differentiate once to get
  \begin{align*}
    h'(x, r)
    &= \left\{\funcparam(X) - \funcparam_1(X)\right\}\frac{d}{dt} \E\left[  Z \psi_t(Y(1))  \mid X=x \right] \bigg|_{t = \funcparam_1(X)+r\left\{\funcparam(X) - \funcparam_1(X)\right\}} \\
    &= -\left\{\funcparam(X) - \funcparam_1(X)\right\} \left[1 + (\Gamma-1)P_1\left(Y(1) < \funcparam_1(X)+r\left\{\funcparam(X) - \funcparam_1(X)\right\} | X=x\right) \right] e_1(X),
  \end{align*}
  and again to get
  \begin{align*}
    h''(x, r)
    &=
    -(\Gamma-1)e_1(X) \left\{\funcparam(X) - \funcparam_1(X)\right\}^2 p_{Y(1)}(\funcparam_1(X)+r\{\funcparam(X) - \funcparam_1(X)\} \mid Z=1, X=x).
  \end{align*}
  Assumption~\ref{assumption:problem-regularity}
  (or the relaxation~\eqref{eqn:condition-replacing-density})
  guarantee that
  the density term $p_{Y(1)}$ above is uniformly bounded.
  Additionally, because $\theta \in \mc{T}_n$, we may use the Lipschitz
  continuity of $t \mapsto \psi_t(Y(1))$ to obtain
  \begin{equation}
    \label{eqn:get-h-bounded-ish}
    h(x, r) \in h(x, 0)
    \pm \Gamma r |\theta(x) - \theta_1(x)|,
  \end{equation}
  where we have used that $h(x, 0) = 0$ for almost all $x$ (recall
  Lemma~\ref{lemma:opt-is-good}).

  Let $f(x, r) = h(x, r)\left\{1 - e_1(x) - r(e(x) - e_1(x))\right\}.$ Then
  \begin{equation}
    \label{eqn:fp-calc}
    f'(x,r) = \left[1 - e_1(x) - r\{e(x) - e_1(x)\}\right]
    h'(x,r)  - \{e(x) - e_1(x)\} h(x,r),
  \end{equation}
  and
  \begin{align*}
    f''(x, r)
    &= -2\{e(x) - e_1(x)\} h'(x,r)
    + \left[1 - e_1(x) - r\{e(x) - e_1(x)\}\right]
    h''(x,r).
  \end{align*}
  Because $e(x) \in (\epsilon, 1-\epsilon)$ by definition~\eqref{eqn:def-Tn}
  of $\mc{T}_n$, we have
  \begin{equation}
    \label{eqn:get-fp-bounded}
    \left|f'(x,r)\right| \le (1-\epsilon)\Gamma
    |\funcparam(x) - \funcparam_1(x)| + |e(x) - e_1(x)| |h(x,r)|
    \le \Gamma |\funcparam(x) - \funcparam_1(x)|(1 + |e(x) - e_1(x)|)    
  \end{equation}
  by inequality~\eqref{eqn:get-h-bounded-ish},
  and
  \begin{align*}
    \left|f''(x,r)\right|
    & \le 2\Gamma|e(x) - e_1(x)||\funcparam(x) - \funcparam_1(x)| \\
    & \quad ~
    + (1-\epsilon)(\Gamma-1) \left(\funcparam(X) - \funcparam_1(X)\right)^2 p_{Y(1)}(\funcparam_1(X)+r(\funcparam(X) - \funcparam_1(X)) | Z=1, X=x).
  \end{align*}
  As we note above, there is $R < \infty$ such
  that $\sup_t p_{Y(1)}(t \mid Z = 1, X = x) \le R$,
  giving the bound
  \begin{align}
    \left| f''(x,r)\right|
    & \le 2\Gamma|e(x) - e_1(x)||\funcparam(x) - \funcparam_1(x)| +
    \Gamma R \left\{\funcparam(X) - \funcparam_1(X)\right\}^2.
    \label{eqn:uniform-bound-on-fpp}
  \end{align}
  Now, let
  \begin{equation*}
    g(x, r) = \left[\nu_1(x) + r\{\nu(x) - \nu_1(x)\} \right]\left[e_1(x) + r\{e(x) - e_1(x)\} \right],
  \end{equation*}
  where we recall that
  $\sup_x |\nu(x) - \nu_1(x)| < \infty$ and
  $|e(x) - e_1(x)| \le 2 - 2 \epsilon$ by definition~\eqref{eqn:def-Tn}
  of $\mc{T}_n$.  
  Differentiating, we have
  \begin{align}
    \nonumber g'(x, r)
    & =
    \{\nu(x) - \nu_1(x)\} \left[e_1(x) + r\{e(x) - e_1(x)\} \right]
    +
    \left[\nu_1(x) + r\{\nu(x) - \nu_1(x)\} \right]\{e(x) - e_1(x)\} \\
    g''(x, r) & = 2\{\nu(x) - \nu_1(x)\}\{e(x) - e_1(x)\}.
    \label{eqn:gp-gpp}
  \end{align}
  Therefore,
  \begin{align}
    \label{eqn:gp-super-bound}
    \sup_{x \in \mc{X}, r \in (0, 1)}
    |g'(x, r)|
    & \le \sup_{x \in \mc{X}} \left\{3|\nu(x) - \nu_1(x)| + 
    3(\Gamma+1) |e(x) - e_1(x)|\right\}
    \le C
  \end{align}
  where $C < \infty$ by the boundedness conditions on $\nu$ and $e$
  that $\mc{T}_n$ guarantees~\eqref{eqn:def-Tn}.

  The following lemma abstracts the technical challenge in bounding the
  second derivatives via dominated convergence. We
  defer its proof temporarily to
  Sec.~\ref{sec:proof-second-deriv}.
  \begin{restatable}{lemma}{lemquotient}
    \label{lem:second-deriv}
    Let $f, g$ be as above.
    Then
    $\frac{\partial^2}{\partial r^2} \E[\frac{f(X, r)}{g(X, r)}]
    = \E[\frac{\partial^2}{\partial r^2} \frac{f(X, r)}{g(X, r)}]$ and
    there exists a finite $C < \infty$ such that for all $r \in (0, 1)$,
    \begin{equation*}
      \left|\frac{\partial^2}{\partial r^2} \E\left[
        \frac{f(X,r)}{g(X,r)} \right] \right|
      \le C\E\left[  \left| f''(X, r)\right| +
        |f(X, r)|\left( \left|g''(X, r)\right|
        + g'(X, r)^2 \right) + \left| f'(X, r) g'(X, r)\right|
        \right].
    \end{equation*}
  \end{restatable}

  As a consequence of Lemma~\ref{lem:second-deriv} and
  Young's inequality that $ab \le \half a^2 + \half b^2$,
  we apply (variously) inequality~\eqref{eqn:uniform-bound-on-fpp}
  to $f''$,
  inequality~\eqref{eqn:get-h-bounded-ish}
  to obtain $|f(x, r)| \le \Gamma |\theta(x) - \theta_1(x)|
  \le \Gamma C_1$ by definition~\eqref{eqn:def-Tn} of $\mc{T}_n$,
  Eq.~\eqref{eqn:gp-gpp} to get
  $|g''(x, r)| \le (\nu(x) - \nu_1(x))^2 + (e(x) - e_1(x))^2$,
  again Eq.~\eqref{eqn:gp-gpp}
  to get $|g'(x, r)|
  \le C |\nu(x) - \nu_1(x)|
  + C|e(x) - e_1(x)|$, and
  inequality~\eqref{eqn:get-fp-bounded} to bound
  $|f'(x, r)| \le C \Gamma |\funcparam(x) - \funcparam_1(x)|$,
  yielding
  \begin{align*}
    \left|\frac{\partial^2}{\partial r^2} \E\left[\frac{f(X, r)}{g(X, r)}
      \right]\right|
    & \le C \E\left[
      |e(X) - e_1(X)|^2 + |\funcparam(X) - \funcparam_1(X)|^2
      + |\nu(X) - \nu_1(X)|^2\right].
  \end{align*}
  By construction~\eqref{eqn:def-Tn} of $\mc{T}_n$,
  $\norm{\eta - \eta_1}_{2,P} \le a_n n^{-1/4}$,
  so that noting
  that
  $\frac{\partial^2}{\partial r^2}
  \E[\score(W, \mu_1^-, \eta_1 + r(\eta - \eta_1))]
  = \frac{\partial^2}{\partial r^2}
  \E[\frac{f(X, r)}{g(X, r)}]$
  gives the result.
\end{proof}

\subsubsection{Proof of Lemma~\ref{lem:second-deriv}}
\label{sec:proof-second-deriv}

By applying the quotient rule, we have
\begin{align*}
  \frac{\partial^2}{\partial r^2}
  \frac{f(x, r)}{g(x, r)}
  & = \frac{f''(x, r)}{g(x, r)}
  - 2 \frac{f'(x, r) g'(x, r)}{g(x, r)^2}
  - \frac{f(x, r) g''(x, r)}{g(x, r)^2}
  + \frac{2 f(x, r) (g'(x, r))^2}{g(x, r)^3}.
\end{align*}
If we can exhibit a function $R(x)$ such that $R(x) \ge
|\frac{\partial^2}{\partial r^2} \frac{f(x, r)}{g(x, r)}|$ for all $r \in (0,1)$ with $\E[R(X)] < \infty$, this is sufficient for interchanging
expectation and differentiation via the dominated convergence theorem.

Because $g(x, r) =
(\nu_1(x) + r(\nu(x) - \nu_1(x))) (e_1(x) + r(e(x) - e_1(x)))$,
the definition~\eqref{eqn:def-Tn} of $\mc{T}_n$ guarantees
that $\inf_{x,r\in[0,1]} g(x, r) > 0$. Therefore,
for some $C < \infty$ whose value may change from line to line but
which is independent of $x$ and $r$ for $r\in(0,1)$,
\begin{align*}
  \left|\frac{\partial^2}{\partial r^2}
  \frac{f(x, r)}{g(x, r)}\right|
  \le C \left[
    |f''(x, r)| + |f'(x, r)| |g'(x, r)|
    + |f(x, r) g''(x, r)| + |f(x, r)| (g'(x, r))^2\right].
\end{align*}
We bound each of the terms in the inequality
in turn.

Assume $r \in (0,1)$. Inequality~\eqref{eqn:uniform-bound-on-fpp} guarantees the existence of an
integrable $R_1$ such that $R_1(x) \ge |f''(x, r)|$. For
the $f' g'$ term, we have $|f'(x, r)| \le C (|h(x, r)| + |h'(x, r)|)$ for
some $C < \infty$ by Eq.~\eqref{eqn:fp-calc}, while
Eq.~\eqref{eqn:gp-super-bound} gives $\linf{g'} \le C$. Thus, as
$\funcparam$ is integrable by assumption~\eqref{eqn:def-Tn}, there is some
integrable $R_2$ such that $R_2(x) \ge |f'(x, r)| |g'(x, r)|$.  Now, we consider the term $|f(x, r) g''(x, r)|$. By
Eq.~\eqref{eqn:gp-gpp} we have $|g''(x, r)| \le C$ because $\nu, e \in
\mc{T}_n$. By the discussion after the
definition~\eqref{eqn:define-shitty-h} of $h$, there exists integrable $R_3$
such that $R_3(x) \ge |f(x, r)| C \ge |f(x, r)| |g''(x, r)|$.
Finally, we consider
the final term. As $\linf{g'} \le C$,
we may again use the function $R_3$, and so we have
\begin{equation*}
  \left|\frac{\partial^2}{\partial r^2}
  \frac{f(x, r)}{g(x, r)}\right|
  \le C (R_1(x) + R_2(x) + R_3(x)),
\end{equation*}
which is integrable. This gives the lemma.

\subsection{Design sensitivity proofs}
\label{sec:proof-design-sensitivity}

In this section, we prove Proposition~\ref{prop:design-sensitivity-model}.
We begin with a technical lemma, returning
to prove the proposition in Sec.~\ref{sec:proof-design-sensitivity-model}.
\begin{restatable}{lemma}{lemmamonotonects}
  \label{lem:monotone-continuous}
  Let $\theta_1^\Gamma(x)$ be the optimum~\eqref{eq:population-theta}
  for a fixed $\Gamma \ge 1$. If $\theta_1^\Gamma(x)$ is finite
  for some $\Gamma$, then $\Gamma \mapsto \theta_1^\Gamma(x)$ is
  continuous and monotone decreasing.
\end{restatable}
\begin{proof}
  To check that $\Gamma \mapsto \theta_1^\Gamma(x)$ is strictly monotone, we
  use the choice of $L$ that attains the minimum in
  equation~\eqref{eq:best-lr} to write
  \begin{equation}
    \label{eqn:get-theta-via-infimum}
    \theta_1^\Gamma(x) =
    \inf_\mu \E \left[\frac{\ind{Y(1) \ge \mu} + \Gamma\ind{Y(1) < \mu}}{
        C_\Gamma(\mu)} Y(1) \mid Z=1, X=x \right],
  \end{equation}
  where $C_\Gamma(\mu) = P(Y(1) \ge \mu \mid Z = 1, X = x) + \Gamma P(Y(1) < \mu
  \mid Z = 1, X = x)$ normalizes $\ind{Y(1) \ge \mu} + \ind{Y(1) < \mu}$ so
  that it is a valid likelihood ratio.  Lemma~\ref{lemma:duality} implies
  that $\theta_1^\Gamma(x)$ itself achieves the
  infimum~\eqref{eqn:get-theta-via-infimum}. Then, for $\widetilde{\Gamma} >
  \Gamma$, if $\var(Y(1) \mid X = x) > 0$, then
  \begin{align}
    \lefteqn{\theta_1^{\widetilde{\Gamma}}(x) - \theta_1^{\Gamma}(x)} \nonumber \\
    & = \inf_\mu \E \left[ \frac{\ind{Y(1) \ge \mu}
        + \widetilde{\Gamma}\ind{Y(1) < \mu}}{C_{\wt{\Gamma}}(\mu)}
      Y(1) \mid Z=1, X=x \right]
    - \theta_1^{\Gamma}(x) \nonumber \\
    &\le \E \left[ \frac{\ind{Y(1) \ge \theta_1^{\Gamma}(x)}
        + \widetilde{\Gamma}
        \ind{Y(1) < \theta_1^{\Gamma}(x)}}{C_{\wt{\Gamma}}(\theta_1^{\Gamma}(x))}
      (Y(1)-\theta_1^{\Gamma}(x)) \mid Z=1, X=x \right]
    \label{eq:before-weight-change} \\
    & < \E \left[ \frac{\ind{Y(1) \ge \theta_1^{\Gamma}(x)}
        + \Gamma\ind{Y(1) < \theta_1^{\Gamma}(x)}}{
        C_{\wt{\Gamma}}(\theta_1^{\Gamma}(x))}
      (Y(1)-\theta_1^{\Gamma}(x)) \mid Z=1, X=x \right]
    \label{eq:after-weight-change} \\
    & = 0. \nonumber
  \end{align}
  The strict inequality~\eqref{eq:after-weight-change} follows by
  considering the signs of $Y(1) - \theta_1^\Gamma(x)$ in
  expression~\eqref{eq:before-weight-change}, that $\wt{\Gamma} > \Gamma$,
  and that $\var(Y(1) \mid X = x) > 0$. The final equality is simply the
  definition of $\theta_1^\Gamma$ via the
  expectation~\eqref{eq:population-theta}.
  
  The function $t \mapsto f_\Gamma(t) \defeq \E[(Y(1) - t)_+ - \Gamma(Y(1) - t)_-
    \mid Z=1, X=x]$ is strictly monotone with slope $ \le -1$. Therefore,
  for $1 \le \Gamma \le \widetilde{\Gamma} < \infty$, using that
  $f_\Gamma(\theta_1^{\Gamma}) = 0$ and
  $f_{\widetilde{\Gamma}}(\theta_1^{\widetilde{\Gamma}}) = 0$,
  \begin{align*}
      |\theta_1^{\Gamma}(x) - \theta_1^{\widetilde{\Gamma}}(x)| &\le f_{\widetilde{\Gamma}}(\theta_1^\Gamma(x)) - f_{\widetilde{\Gamma}}(\theta_1^{\widetilde{\Gamma}}(x))
      = f_{\widetilde{\Gamma}}(\theta_1^{\Gamma}(x)) - f_{\Gamma}(\theta_1^{\Gamma}(x)) 
      \le (\widetilde{\Gamma} - \Gamma)\E\left[ (Y(1) - \theta_{\Gamma}(x))_- \right].
  \end{align*}
  When $\theta^{\Gamma}_1(x)$ is finite, this implies $\Gamma \mapsto
  \theta_1^\Gamma(x)$ is continuous.
\end{proof}

\subsubsection{Proof of Proposition~\ref{prop:design-sensitivity-model}}
\label{sec:proof-design-sensitivity-model}

\thmdesignsensitivity*
\begin{proof}
  That $\tau^-$ takes the form claimed in the proposition is a consequence
  of Eq.~\eqref{eqn:get-theta-via-infimum}.  Let us consider the power of
  the tests $\psi_n^\Gamma$.  As $n \to \infty$, we have
  $\what{\tau}^--z_{1-\alpha}\what{\sigma}_{\tau^-}/\sqrt{n} \cp \tau^{-}$,
  so that $\lim_{n \to \infty} Q(\psi_n^\Gamma = 0) = 0$ if $\tau^- > 1$,
  and otherwise $Q(\psi^\Gamma_n = 0) \to 1$.  By
  Lemma~\ref{lem:monotone-continuous}, $\tau^-(\Gamma)$ is strictly
  decreasing in $\Gamma$, so the design sensitivity $\Gamma_{\design}$ for
  this test is the choice of $\Gamma$ such that $\tau^-(\Gamma) = 0$. If
  this equation has no roots, then no choice of $\Gamma$ makes $\tau^-(\Gamma)$
  negative, so we set $\Gamma_{\design}=\infty$.
\end{proof}

\subsubsection{Proof of Corollary~\ref{cor:designsensitivity}}
\label{sec:proof-gaussian-design-sensitivity}

\cordesignsensitivity*
\begin{proof}
  For notational convenience, we use $\tilde{\Gamma} \equiv
  \Gamma_\design$ in the proof.
  Proposition~\ref{prop:design-sensitivity-model} shows that the design
  sensitivity $\tilde{\Gamma}$ satisfies
  \begin{align*}
    0 & = \E_Q[Y(1)] + \funcparam_1 - \E_Q[Y(0)] - \theta_0
    \stackrel{(i)}{=} \tau + 2 \theta_1,
  \end{align*}
  for $\theta_1$ solving $\E_Q[\hinge{Y(1) - \theta_1} - \tilde{\Gamma}
    \neghinge{Y(1) - \theta_1}] = 0$ (recall
  Lemmas~\ref{lemma:duality} and~\ref{lemma:opt-is-good})
  where equality~$(i)$
  uses that $Y(1)
  \overset{d}= -Y(0)$ under $Q$. The design sensitivity
  $\tilde{\Gamma}$ thus solves
  \begin{equation*}
    \E_Q\left[\hinge{Y(1) + \frac{\tau}{2}} - \tilde{\Gamma}
      \neghinge{Y(1) + \frac{\tau}{2}} \right] = 0.
  \end{equation*}
  
  Substituting the density of $Y(1)$ under $Q$ gives
  \begin{align*}
    0 & = \int_{-\infty}^\infty
    \left(\indic{y \ge -\tau/2} + \tilde{\Gamma} \indic{y < \tau/2}\right)
    \frac{y+\tau/2}{\sqrt{2 \pi \sigma^2}}
    \exp\left(-\frac{1}{2 \sigma^2}(y - \tau/2)^2\right)\dif{t}\\
    &= \int_{-\infty}^\infty
    \left(\indic{t \ge 0} + \tilde{\Gamma} \indic{t < 0}\right)
    \frac{t}{\sqrt{2 \pi \sigma^2}}
    \exp\left(-\frac{1}{2 \sigma^2}(t - \tau)^2\right) \dif{t}
  \end{align*}
  by a change of variables.
  This immediately implies that
  \begin{equation*}
    \tilde{\Gamma} \int_{-\infty}^0 t \exp\left(-\frac{1}{2 \sigma^2}
    (t - \tau)^2 \right) \dif{t}
    = -\int_0^\infty t \exp\left(-\frac{1}{2 \sigma^2} (t - \tau)^2\right) \dif{t},
  \end{equation*}
  which gives the first equality in the corollary.
  The second equality is just a change of variables and computation of
  the integral.
\end{proof}

\subsection{Proof of Proposition~\ref{prop:opt-design-sensitivity}}
\label{sec:proof-opt-design-sensitivity}

\optdesignsensitivity*
\begin{proof}
  We assume for simplicity that $\sigma = 1$; replacing $\tau$ by
  $\tau/\sigma$ gives an equivalent problem and does not
  change the quantity~\eqref{eq:design-sensitivity-gaussian}.
  We show that for $\Gamma \ge \Gamma_\design^\gauss$, there
  exists $P \in H_0(\Gamma)$ such that $\tvnorms{P_{Y(Z), Z} - Q_{Y(Z), Z}} = 0$,
  where the notation indicates that we only observe pairs $Y(Z), Z$.
  Our choice of $P$ will be independent of the sample size $n$, so
  that it gives both asymptotic and finite sample results.
  We proceed in three steps: (1) we construct $P$, (2) we verify the constructed $P$ belongs
  to $H_0(\Gamma)$, and (3) we show the variation distance is zero.
  
  \paragraph{Step 1: Constructing $P$}
  We construct $P$ using the distribution $Q$, augmenting with an unobserved
  confounding variable $U$ to have the Markov structure
  \begin{equation*}
    Z \longleftarrow U \longrightarrow (Y(1), Y(0)).
  \end{equation*}
  Let
  \begin{equation*}
    t^* = t^*(\Gamma)
    \defeq \argmin_t \E_Q\left[\frac{\ind{Y(1)\ge t}
        + \Gamma\ind{Y(1) < t}}{Q(Y(1)\ge t) + \Gamma Q(Y(1) < t)}Y(1)\right],
  \end{equation*}
  where Lemma~\ref{lemma:duality} shows that $t^\ast=\theta_1$ attains this
  minimum.  Denote the densities of $Y(1)$ and $Y(0)$ under $Q$ by $q_1$ and
  $q_0$, respectively, where $q_0(t) = q_1(-t)$.
  Under $P$, let $Y(1)$ have density
  \begin{subequations}
    \label{eqn:lr-treated-p}
    \begin{equation}
      \label{eq:lr-treated-p1}
      p_1(t) \propto \left( \ind{t > t^\ast} + \sqrt{\Gamma} \ind{t \le t^\ast}\right) q_1(t),
    \end{equation}
    and set $Y(0) = -Y(1)$, so that marginally $Y(0)$ has density
    \begin{equation}
      p_0(t) \propto \left( \ind{t <  -t^\ast} + \sqrt{\Gamma} \ind{t \ge -t^\ast}\right) q_0(t).
      \label{eq:lr-treated-p0}
    \end{equation}
  \end{subequations}
  Define the unobserved confounding variable
  \begin{equation*}
    U = \ind{Y(1) > t^\ast} = \ind{Y(0) < -t^\ast}
  \end{equation*}
  and let $Z$ be a random variable based on the conditional probabilities
  \begin{equation*}
    P(Z = z \mid U = u) = \begin{cases}
      \sqrt{\Gamma} / (1 + \sqrt{\Gamma}) & \mbox{if~} z = u \\
      1 / (1 + \sqrt{\Gamma}) & \mbox{if~} z = 1 - u,
    \end{cases}
  \end{equation*}
  so $Z$ is independent of $Y(1)$ and $Y(0)$ conditional on $U$.

  Define $a$ by the marginal distribution of $Z$ under this conditional
  distribution,
  \begin{align*}
    a \defeq P(Z=1)&=P(Z=1\mid U=1)P(U=1)+P(Z=1 \mid U=0)P(U=0)
    \in \left[\frac{1}{1 + \sqrt{\Gamma}}, \frac{\Gamma}{1 + \sqrt{\Gamma}}
      \right].
  \end{align*}

  \paragraph{Step 2: Verifying that $P \in H_0(\Gamma)$}
  The $\Gamma$-\cornfield{} condition~\eqref{eq:cornfield} holds, as
  \begin{align*}
    \frac{P(Z = 1 \mid U = 1)}{P(Z = 0 \mid U = 1)}\frac{P(Z = 0 \mid U = 0)}{P(Z = 1 \mid U = 0)} &= \frac{\frac{\sqrt{\Gamma}}{1 + \sqrt{\Gamma}}}{\frac{1}{1 + \sqrt{\Gamma}}}\frac{\frac{\sqrt{\Gamma}}{1 + \sqrt{\Gamma}}}{\frac{1}{1 + \sqrt{\Gamma}}} = \Gamma.
  \end{align*}
  Therefore, $P \in H_0(\Gamma)$ if $\E_P[Y(1) - Y(0)] \le 0$. To verify
  this condition, we first calculate the conditional likelihood ratios using
  Bayes rule:
  \begin{align*}
    \frac{p_{Y(1) | Z=1}(t)}{p_{Y(1)}(t)} &= \frac{P(Z=1 \mid Y(1) = t)}{P(Z=1)}
    \\
    &=
    \frac{P(Z=1 \mid U = 1)P(U=1 \mid Y(1) = t) + P(Z=1 \mid U = 0)P(U=0 \mid Y(1) = t)}{P(Z=1)}
    \\
    &\propto
    \sqrt{\Gamma}\ind{t > t^\ast} +  \ind{t \le t^\ast},
  \end{align*}
  and similarly
  $$
  \frac{p_{Y(0) | Z=0}(t)}{p_{Y(0)}(t)} = \frac{P(Z=0 \mid Y(0) = t)}{P(Z=0)}\propto
  \ind{t < -t^\ast} + \sqrt{\Gamma}\ind{t \ge -t^\ast}.
  $$
  Thus, the definitions~\eqref{eqn:lr-treated-p} of
  $p_z(t)$ yield
  \begin{align}
    \label{eq:equal-marginal-y-1}
    \frac{p_{Y(1) | Z=1}(t)}{q_1(t)} &= \frac{p_{Y(1) | Z=1}(t)}{p_{Y(1)}(t)}  \frac{p_{Y(1)}(t)}{q_1(t)} = 1, \\
    \frac{p_{Y(0) | Z=0}(t)}{q_0(t)} &= 1,
    \label{eq:equal-marginal-y-0}
  \end{align}
  and
  \begin{align*}
    \frac{p_{Y(1) | Z=0}(t)}{q_1(t)} &\propto \ind{t > t^\ast} + \Gamma \ind{t \le t^\ast} \\
    \frac{p_{Y(0) | Z=1}(t)}{q_0(t)} &\propto \ind{t < -t^\ast} + \Gamma \ind{t \ge -t^\ast}.
  \end{align*}
  These imply that
  \begin{align*}
    \E_P\left[Y(1) - Y(0) \right] &= a\E_P\left[Y(1) - Y(0) \mid Z=1 \right] + (1-a)\E_P\left[Y(1) - Y(0) \mid Z=0 \right]
    \\
    & \stackrel{(i)}{=}
    \frac{\tau}{2} + \E_Q\left[\frac{\ind{Y(1) \ge t^\ast} + \Gamma \ind{Y(1) < t^\ast}}{Q(Y(1) \ge t^\ast) + \Gamma Q(Y(1) < t^\ast)}Y(1) \right] \\
    & \stackrel{(ii}{\le}
    \frac{\tau}{2} + \frac{1}{C_0} \int_{-\frac{\tau}{2}}^{\infty} y q_1(y) \dif{y} + \frac{\Gamma}{C_0} \int_{-\infty}^{-\frac{\tau}{2}} y q_1(y) \dif{y}
    \\
    &=
    \frac{1}{C_0} \int_{0}^{\infty} y\, q_1\!\left(y-\frac{\tau}{2}\right) \dif{y} + \frac{\Gamma}{C_0} \int_{-\infty}^{0} y \, q_1\!\left(y-\frac{\tau}{2}\right) \dif{y}
    \\
    & \stackrel{(iii)}{\le} 0,
  \end{align*}
  where inequality~$(i)$ follows by~\eqref{eqn:get-theta-via-infimum}
  defining the threshold and in inequality~$(ii)$ we replace the
  threshold $t^*$ by $\tau/2$, which only increases the objective,
  and the inequality~$(iii)$ holds when $\Gamma \ge \Gamma_\design^\gauss$
  (and $C_0 = Q(Y(1) \ge -\tau/2) + \Gamma Q(Y(1) \le -\tau/2)$).

  \paragraph{Step 3: Proving Zero Variation Distance}
  Now, we verify that observed quantity $(Y(Z), Z)$ has the same
  distribution under $Q$ and $P$. If $Q(Z=1) = a$, then $P(Z=1) = a =
  Q(Z=1)$. Furthermore, the conditional distribution $Y(Z) \mid Z$ under $Q$ equals that under $P$
  by Eqs.~\eqref{eq:equal-marginal-y-1} and
  \eqref{eq:equal-marginal-y-0}. Therefore, the marginal distributions $(Y(Z),Z))$ under $P$ and $Q$ are identical, and
  $\tvnorms{P_{Y(Z), Z} - Q_{Y(Z), Z}} = 0$.
  Consequently, for any test $t_n^\Gamma,$ we have
  $P(t_n^\Gamma = 1) = Q(t_n^\Gamma = 1)$.
\end{proof}


\end{document}